\documentclass[floatfix, twocolumn,superscriptaddress,10pt]{revtex4-2}
\usepackage[T1]{fontenc}
\usepackage{times}
\usepackage[utf8]{inputenc}
\usepackage[resetlabels,labeled]{multibib}
\usepackage{amsmath}
\usepackage{amsthm}
\usepackage{amsfonts}
\usepackage{amssymb}
\usepackage{bm}
\usepackage{bbm}
\usepackage{bbold}
\usepackage{graphicx,epsfig,color,tcolorbox}
\usepackage{physics}
\usepackage{mathtools}
\usepackage{mathrsfs}
\usepackage{latexsym}
\usepackage[colorlinks=true,linkcolor=teal,citecolor=purple]{hyperref}
\usepackage{hyperref}% add hypertext capabilities
\usepackage{cleveref}
\usepackage{lipsum} 
\usepackage[ruled,vlined]{algorithm2e}

\newcommand{\ot}{\otimes}

\newcommand{\diag}{\text{diag}}
\newcommand{\eps}[1]{\epsilon_{\rm \footnotesize #1}}

\SetCommentSty{mycommfont}

\SetKwInput{KwInput}{Input}                % Set the Input
\SetKwInput{KwOutput}{Output}              % set the Output
\SetKwInput{KwParams}{Parameters}

\newcommand{\chg}{\textcolor{black}}
\newtheorem{theorem}{Theorem}
\newtheorem{definition}{Definition}
\newtheorem{lemma}{Lemma}

\begin{document}

\title{All states are universal catalysts in quantum thermodynamics}
\author{Patryk Lipka-Bartosik}
 \affiliation{H. H. Wills Physics Laboratory, University of Bristol, Tyndall Avenue, Bristol, BS8 1TL, United Kingdom}
 \affiliation{Institute of Theoretical Physics and Astrophysics,
National Quantum Information Centre, Faculty of Mathematics, Physics and Informatics,
University of Gdańsk, Wita Stwosza 57, 80-308 Gdańsk, Poland}
\author{Paul Skrzypczyk}
\affiliation{H. H. Wills Physics Laboratory, University of Bristol, Tyndall Avenue, Bristol, BS8 1TL, United Kingdom}

\date{\today}

\begin{abstract}
Quantum catalysis is a fascinating concept which demonstrates that certain transformations can only become possible when given access to a specific resource that has to be returned unaffected. It was first discovered in the context of entanglement theory and since then applied in a number of resource-theoretic frameworks, including quantum thermodynamics. Although in that case the necessary (and sometimes also sufficient) conditions on the existence of a catalyst are known, almost nothing is known about the precise form of the catalyst state required by the transformation. In particular, it is not clear whether it has to have some special properties or be finely tuned to the desired transformation. In this work we describe a surprising property of multi-copy states: we show that in resource theories governed by majorization all resourceful states are catalysts for all allowed transformations. In quantum thermodynamics this means that the so-called “second laws of thermodynamics” do not require a fine-tuned catalyst but rather any state, given sufficiently many copies, can serve as a useful catalyst. \chg{These analytic results are accompanied by several numerical investigations that indicate that neither a multi-copy form nor a very large dimension catalyst are required to activate most allowed transformations catalytically.}
\end{abstract}

\maketitle

\section{Introduction}
The laws of physics are often expressed as limitations on what physical systems can and cannot do. The second law of thermodynamics is a cardinal example of this approach: it says which thermodynamic transformations can be performed under given conditions. Specifically, at a constant background temperature and volume the transition between two \chg{equilibrium} states can occur if and only if the Helmholtz free energy decreases during the process. The second law describes a relationship between average quantities (energy and entropy) and hence specifies the \emph{typical} thermodynamic behavior, i.e. justified in the limit of a large number of identically distributed and weakly interacting systems. 

Recent experiments provide evidence that with our current technology we can control and manipulate systems at much smaller scales than those governed by the second law \cite{Koski2014,Koski2015,Chida2017,Camati2016,Peterson2016,Zanin2019}. Therefore understanding thermodynamic behavior and, in particular, finding the correct way in which the standard laws of thermodynamics translate into this domain, is of crucial practical and theoretical importance. Very recently this translation into the microscopic regime was made possible using powerful tools  derived within the field of classical and quantum information theory \cite{shannon1948mathematical,Ruch1976,Marshall2011,Janzing2000}.

One of the most striking differences between standard thermodynamics and its microscopic counterpart is that transformations between states can become significantly more demanding. More specifically, there are paradigms where they are no longer described by a single second law, but an entire family of conditions, the so-called ``second laws of quantum thermodynamics'' \cite{Brand_o_2015}. In this way the free energy loses its meaning as the unique indicator of which state transitions are possible --- its role is replaced by a family of generalized free energies, a collection of information-theoretic quantities closely related to the Renyi entropies \cite{renyi1961}. This captures the idea that for microscopic systems more structure of the energy distribution must be specified in order to determine their thermodynamic properties. Importantly, by invoking typicality arguments it can be shown that in the limit of identically distributed and weakly interacting systems all members of this family of quantities approach the Helmholtz free energy, thus recovering the standard second law as a special case. 

However, these results rely on a specific assumption: that there exists some thermal machine or `catalyst' which is not consumed by the protocol but nonetheless makes the transformation possible. More specifically, if the second laws are satisfied for a pair of states $\rho$ and $\sigma$ then there is a quantum state $\omega$ which is unchanged by the protocol but still enables the joint transformation $\rho \ot \omega \rightarrow \sigma \ot \omega$. This becomes more natural once we realize that standard treatments implicitly adopt an analogous assumption; to perform a thermodynamic transformation one always needs to supply additional devices which can be cyclically reused (e.g. engines, refrigerators or heat pumps). In this way the ancillary state $\omega$ models the behavior of a thermal machine or an experimental apparatus which facilitates or even enables the transformation. This phenomenon of ``lifting restrictions without being consumed'' is called \emph{quantum catalysis}. 

As may be expected, quantum catalysis is not exclusively related to thermodynamics. The basic idea was introduced for the first time by Jonathan and Plenio in the context of entanglement transformations using local operations and classical communication (LOCC) \cite{Jonathan_1999}. However, the ability to borrow an ancillary state (the catalyst) which remains unchanged can allow for otherwise impossible transformations regardless of the specific physical situation. Because of this generality, the scenario of catalysis can be effectively described using the general tools developed within the framework of  quantum resource theories (QRTs) \cite{HORODECKI2012,Chitambar2019,Nielsen1999,Vidal1999,PhysRevA.91.052120,Horodecki2003,Streltsov2018,Brandao2013,Horodecki2013,Ng2018,Horodecki2015,Cavalcanti2016,Piani2016,Buscemi2020,Supic2019,Cavalcanti2017,magic2017,bhattacharya2018convex,wakakuwa2017operational}. More precisely, the problem has a particularly elegant and conceptually simple description for a class of theories referred to as \emph{majorization-based} quantum resource theories (MB-QRTs). In such theories quantum states are represented by probability vectors which encode their affiliation to the specific resource. The problem of conversion can be then formulated purely in terms of these vectors and answered using the concept of majorization \cite{Marshall2011}. Arguably the most well-studied examples of such theories are the resource theory of entanglement \cite{PhysRevA.53.2046,Nielsen1999,kumagai2016second}, coherence \cite{PhysRevA.91.052120,Napoli2016,Zhu_2017,chitambar2016comparison}, purity \cite{Gour_2015,streltsov2018maximal}, asymmetry \cite{Piani2016,ding2020amplifying} and thermodynamics (or athermality) \cite{Horodecki2013,PhysRevLett.111.250404,Brand_o_2015}. To focus attention we describe our findings in terms of the resource theory of quantum thermodynamics, however, the results which we present here are general and hold for any majorization-based resource theory. 

Returning to quantum thermodynamics, the second laws emerge when the catalyst is returned perfectly undisturbed. In reality, however, every thermodynamic protocol will modify the catalyst's state and so a realistic notion of catalysis must be robust against such perturbations. This leads to the notion of \emph{inexact catalysis}  where the catalyst is allowed to be returned up to some small error $\eps{C}$ \cite{vanDam2003,Brandao2013,Ng_2015}. A natural and operationally motivated error quantifier is the trace distance which also quantifies the best average probability of discriminating quantum states \cite{nielsen_chuang}. Surprisingly, states which are close in trace distance may have very different thermodynamic properties. This allows to ``cheat'' when using such catalysts, e.g. returning them with a small error as quantified by the trace distance but also much lower work content \cite{Ng_2015}. Protocols acting in this way extract work from the catalyst in order to lift the limitations imposed by the second laws and perform the transformation. In this sense the catalyst is used as a work source (or an entropy sink) rather than a device genuinely catalyzing the transformation, leading to the phenomenon known as (thermal) embezzlement \cite{Ng_2015}.

As a result, the partial order quantified by the second laws vanishes and everything becomes possible --- there are no longer any laws. One promising way to amend this situation is to quantify how the error on the catalyst scales with its dimension \cite{Brandao2015}. This approach naturally leads to two different regimes of catalysis: the \emph{embezzlement} regime in which the partial order between states completely vanishes and the \emph{genuine catalysis} regime where the partial order collapses to a certain extent  (so that only a subset of the second laws remains) or is even fully retained. The  boundary between these two regimes in terms of the trace distance error has been studied in \cite{Ng_2015}. There it was found that for systems with fully degenerate Hamiltonians all state transformations become possible when the error exceeds a certain threshold which scales linearly with the number of catalyst particles $n$ (or with the logarithm of the dimension). Furthermore, once the error scaling is better than linear in $n$, some of the generalized free energies are recovered, ultimately leading to a full partial order when no error is allowed. 

Arguably, one of the most important problems within this approach to thermodynamics is how to find a catalyst which can be useful for a given transformation. Many of the existing results are based on constructing a very specific catalyst. This, however, may be obscuring the true physical mechanism behind catalysis. Furthermore, it is still not well understood which properties of quantum states are relevant for catalysis. The second laws only guarantee the existence of the catalyst; even if they are satisfied by a pair of states, it may still be difficult to find which state catalyzes a particular process. This intuition comes from our macroscopic experience: chemical reactions can be catalyzed only by appropriately chosen chemical compounds; similarly thermal machines need to be carefully tuned so that the desired transformation may happen. In this way a natural question appears: how can we find a state which catalyzes a given transformation and how special are these states?

In this work we push forward our understanding of catalysis by reporting a surprising property of multi-copy catalysts which we term \emph{catalytic universality}. We show that any state, as long as enough copies of it are available, can serve as a catalyst for all allowable transformations. For the case of genuine catalysis this means that if the two states obey the second laws, then a catalyst formed from sufficiently many copies of \emph{any} state $\omega$ can catalyze the transformation from $\rho$ to $\sigma$ approximately, i.e with a disturbance on the catalyst decreasing almost exponentially with the number of copies. Furthermore, by employing a recent result from quantum information theory called the convex-split lemma, we show that the universality phenomenon manifests also into the embezzlement regime, i.e. when the partial order between states fully collapses. In this case sufficiently many copies of any state can catalyze \emph{any} state transformation with a vanishing disturbance, although much slower than in the regime of genuine catalysis. 

We also emphasize that these results are valid for any QRT whose transitions are governed by majorization. In this way the phenomenon of catalytic universality appears naturally in the resource theory of entanglement, coherence or purity. Therefore, our results also lead to new insights in the theory of entanglement by characterizing new and broad families of universal embezzling states.

The paper is structured as follows. In Sec. \ref{sec2} we introduce the relevant framework for thermodynamics and the main mathematical tools used to prove our results. In Sec. \ref{sec3} we describe the general protocol and then specify it to the two catalysis regimes: embezzlement and genuine catalysis. Then in Sec. \ref{sec4} we provide a numerical evidence that the phenomenon of catalytic universality can be even more general and conjecture that it holds for arbitrary states with a sufficiently large dimension.  In Sec. \ref{sec5} we provide a brief summary of our main results and finally, in Sec. \ref{sec6}, we discuss potential implications and describe several open problems which follow naturally from these findings.  

\section{Framework}
\label{sec2}
We begin by describing the resource theory of quantum thermodynamics. As a starting point we define a restricted set of quantum operations known as thermal operations (TO's) which were introduced in \cite{Horodecki2013}, and subsequently studied in \cite{Perry2015,Richens2017,Brandao2013,Brand_o_2015,Lostaglio2016,Perry2018,Faist2015gibbs,Masanes2017,Lostaglio2015,Korzekwa2016,Horodecki2014,Kwon2018clock,Mueller2018,obejko2020thermodynamics,binder2019thermodynamics}. This is an established setting which allows for exploring fundamental thermodynamic limitations by assuming perfect control over the environment. Furthermore, recent studies indicate that these operations can be achieved experimentally with a coarse-grained control \cite{Perry2018}. A readable introduction to this field of quantum thermodynamics can be found in \cite{Goold2015,Vinjanampathy_2016,Lostaglio2019}. One of the main benefits of using this formal framework is that it readily allows to apply the results of quantum and classical information theory in the thermodynamic context. This can be then adapted to specific physical scenario by considering a more demanding dynamics. 

The setting studied by the TO framework consists of a system $\rm S$ with Hamiltonian $H_{\rm S} = \sum_{i=1}^{d_{\rm S}} E_i \dyad{i}_{\rm S}$ and a heat bath $\rm{B}$ at temperature $T$ with Hamiltonian $H_{\rm{B}}$ satisfying a few reasonable assumptions about its energy spectrum (see \cite{Horodecki2013} for the details). We will always assume that the heat bath starts in a thermal state $\tau_{\rm{B}}^{\,} = e^{-\beta H_{\rm{B}}}/Z_{\rm{B}}$, where $Z_{\rm{B}} = \Tr e^{-\beta H_{\rm{B}}}$ is the partition function and $\beta = 1 / k_{\rm B}T$ is the inverse temperature. The interaction of the system with the heat bath is modelled using a unitary $U_{\rm SB}$ which conserves the total energy,  i.e. $[U_{\rm SB}, H_{\rm S} + H_{\rm B}] = 0$. The map $\mathcal{T}_{\rm S}$ arising from this unitary after tracing out the ancillary degrees of freedom is called a \emph{thermal operation} (TO) and can be written as:
\begin{align}
    \label{def_to_eq}
    \mathcal{T}_{\rm S}[\rho_{\rm S}] = \Tr_{\rm B'} \left[U_{\rm SB} \left(\rho_{\rm S} \ot \tau_{\rm B}\right) U_{\rm SB}^{\dagger}\right],
\end{align}
where the trace can be performed over any system $\rm B'$ inside $\rm SB$. In general, a complete characterization of the set of operations (\ref{def_to_eq}) is not known. However, for states $\rho$ and $\sigma$ block-diagonal in the energy eigenbasis there exists a simple criterion determining when there exists a TO such that $\mathcal{T}[\rho] = \sigma$. To present this criterion let us first construct a resource representation of the two states, i.e:
\begin{align}
    \label{eq:res_rep}
    \bm{p} = (p_1, p_2, \ldots, p_{d_{\rm S}}), \qquad \bm{q} = (q_1, q_2, \ldots, q_{d_{\rm S}})
\end{align}
where $p_i = \langle E_i |\rho| E_i \rangle$ and $q_i = \langle E_i |\sigma| E_i \rangle$ are the state's occupations in the energy eigenbasis. Similarly we denote the system's thermal state with $\tau_{\rm S}^{\,} = \text{diag}[\bm{g}] = (g_1, g_2, \ldots, g_{d_{\rm S}})$ with $g_i = e^{- \beta E_i} / Z_{\rm S}$. Let $\pi(i)$ be a permutation of the indices $i$, such that the vector with elements $p_{\pi(i)} / g_{\pi(i)}$ is sorted in a non-increasing order (\emph{beta-ordered}). Following \cite{Horodecki2013}, for such an ordered state one then constructs a curve (\emph{thermo-majorization curve}) by drawing points of the form:
\begin{align}
    \{(\sum_{i=1}^k g_{\pi(i)}, \sum_{i=1}^k p_{\pi(i)})\}_{k=1}^{d_{\rm S}}
\end{align}
together with $\{(0, 0)\}$, and connecting them piece-wise linearly to form a convex curve. As it was proved in \cite{Horodecki2013}, transformation between block-diagonal states $\rho \rightarrow \sigma$ via thermal operations is possible if and only if the thermo-majorization curve of $\rho_{\rm S}$ is never below the curve of $\sigma_{\rm S}$. This relation is known as thermo-majorization and will be denoted with ``$\succ_{\rm T}$''. Notice that this notion also recovers, as a special case, the standard majorization relation either by considering the limit of infinite temperature ($\beta \rightarrow 0$) or fully degenerate system's Hamiltonian ($H_{\rm S} \propto \mathbb{1}_{\rm S}$).

The framework of thermal operations  can naturally accommodate the phenomenon of catalysis. To do so let us consider an ancillary system $\rm{C}$ prepared in a state $\omega_{\rm C}$ of dimension $d_{\rm C}$ and Hamiltonian $H_{\rm C}$. It turns out that any transformation between \chg{diagonal} states which can be performed using a catalyst with a nontrivial energy spectrum can also be accomplished using a catalyst with a fully degenerate spectrum. In this sense, to describe all possible state transformations, without loss of generality we can always choose a trivial Hamiltonian $H_{\rm C} \propto \mathbb{1}_{\rm C}$ \cite{Brand_o_2015}. Having that said, we consider catalytic thermal operations (CTOs) to be transformations of the following form:
\begin{align}
    \label{eq:cat}
    \mathcal{T}_{\rm SC}[\rho_{\rm S } \ot \omega_{\rm C}] = \sigma_{\rm S} \ot \omega_{\rm C},
\end{align}
where now $\mathcal{T}_{\rm SC}$ is a thermal operation (\ref{def_to_eq}) with $\rm S$ being replaced by the joint system $\rm SC$. A fundamental question to ask is when there exists a catalyst $\omega_{\rm C}$ which can facilitate a given transformation $\rho \rightarrow \sigma$. In this case the necessary conditions for the existence of a transformation between two states is captured by a set of quantities called generalized or $\alpha$-free energies $F_{\alpha}$. An important result of Ref. \cite{Brand_o_2015} states that there exists a catalyst $\omega_{\rm C}$ which enables the transformation $\rho_{\rm S} \rightarrow \sigma_{\rm S}$ as in (\ref{eq:cat}) only if:
\begin{align}
    \label{eq:sec_laws}
    F_{\alpha}(\rho_{\rm S}, \tau_{\rm S})  \geq F_{\alpha}(\sigma_{\rm S}, \tau_{\rm S})  \qquad \forall\, \alpha \geq 0.
\end{align}
These are the \emph{second laws of thermodynamics} as stated in the introduction. In fact, the precise statement of the second laws which we give here requires two additional technical assumptions. First, it assumes an arbitrarily small but nonzero error in the transformation. Second, it requires borrowing a qubit in a pure state that is given back with an arbitrarily small, but again, nonzero error \footnote{Without this two technical assumptions the precise form of the second laws should read: $\forall\, \alpha \in \mathbb{R}\quad F_{\alpha}(\rho_{\rm S}, \tau_{\rm S}) > F_{\alpha}(\sigma_{\rm S}, \tau_{\rm S})$}. The functions $F_{\alpha}(\rho_{\rm S}, \tau_{\rm S})$ are defined as:
\begin{align}
    F_{\alpha}(\rho_{\rm S}, \tau_{\rm S}) := \frac{1}{\beta}\left[D_{\alpha}(\rho_{\rm S}||\tau_{\rm S}) - \log Z_{\rm S} \right]  
\end{align}
with $D_{\alpha}(\rho||\tau)$ being the quantum Renyi divergences defined in \cite{Lennert2013}. Importantly, the conditions (\ref{eq:sec_laws}) become sufficient if the states $\rho_{\rm S}$ and $\sigma_{\rm S}$ are block-diagonal in the energy basis determined by $H_{\rm S}$. This means that they commute with the operator $\tau_{\rm S}$ and hence the quantum Renyi divergence $D_{\alpha}(\rho_{\rm S}||\tau_{\rm S})$ for $\alpha \geq 0$ simplifies to:
\begin{align}
    \label{def:rel_ent}
    D_{\alpha} (\bm{p}||\bm{g}) =  \frac{1}{\alpha-1} \log\left[ \sum_{i} p_i^{\alpha} g_i^{1-\alpha}\right].
\end{align}
This also allows the second laws to be written in a much simpler form, and to see more clearly the connection between $F_{\alpha}$ and the non-equilibrium Helmholtz free energy, which is given by $F_{1} = - k_{\rm B}T \log Z_{\rm S}$. It is important to note that the second laws (\ref{eq:sec_laws}) are \emph{strictly} looser than the thermo-majorization criteria. This means that there are transformations which cannot be realized via TO, i.e. without a catalyst, but can be performed when given access to a one. This is precisely due to this realization that catalysis is an important and highly non-trivial phenomenon in the resource theory of thermodynamics.

The notion of catalysis can be naturally generalized to more physical scenarios if we allow for small perturbation in the final state of the catalyst. This relaxation leads to \emph{inexact catalysis}, where the error on the catalyst $\eps{C}$ is defined as:
\begin{align}
     \eps{C} := \norm{\Tr_{\rm S} \mathcal{T}_{\rm SC}[\rho_{\rm S} \ot \omega_{\rm C}] - \omega_{\rm C}}_1, 
\end{align}
where $\norm{M}_1 := \max{\{\Tr[P M]\,|\, \scriptsize 0 \preceq P \preceq \mathbb{1}\}}$ is the trace distance or $1$-norm. The case of exact catalysis can be recovered by  $(i)$ setting the error on the catalyst to zero, i.e. $\eps{C} = 0$ and $(ii)$ allowing no correlations between the system and the catalyst, i.e. demanding that the two subsystems end up in a product form. This is also the regime in which all of the second laws must be satisfied in order to transform one state into another. The case when $\eps{C} = 0$ but arbitrary correlations between $\rm S$ and $\rm C$ are allowed to build up has been thoroughly studied in \cite{Mueller2018}. There it was found that using a finely-tuned catalyst one can transform $\rho_{\rm S}$ into $\sigma_{\rm S}$, as long as the free energy of $\rho_{\rm S}$ is higher than the free energy of $\sigma_{\rm S}$. This leads to the conclusion that only one of the family of second laws remains, namely the non-equilibrium Helmholtz free energy $F_1$. Moreover, the authors of \cite{Brand_o_2015} showed that when the error on the catalyst scales \emph{linearly} with the number of particles (up to a constant factor) $n = \log d_{\rm C}$, that is when $\eps{C} \sim 1 / n$, then the non-equilibrium Helmholtz free energy $F_{ 1}$ again completely describes all possible transformations. Finally, in \cite{Ng_2015} it was found that the second laws completely vanish (meaning that all state transitions are possible) when the error on the catalyst surpasses a certain threshold which is determined by:
\begin{align}
    \label{def:emb_bnd}
    \eps{C}^{\text{bnd}} = \frac{d_{\rm S} - 1}{1 +(d_{\rm S} - 1) \log d_{\rm C}} \sim \frac{1}{n}.
\end{align}
In other words, this is the minimal error which can be achieved under the assumption that all states can be converted between each other. This sets a boundary between the embezzlement and genuine catalysis regime; whenever error on the catalyst scales with its dimension better than (\ref{def:emb_bnd}) then there are state transitions which are not allowed and so the partial order induced by the second laws is recovered to a certain degree. Whenever the scaling of $\eps{C}$ is worse or equal to (\ref{def:emb_bnd}) we will refer to the corresponding regime of catalysis as \emph{embezzlement} regime. On the contrary, we will use the term \emph{genuine catalysis} to indicate that transformations are still governed by the second laws (or some non-empty subset of them). In order to simplify notation in what follows we will indicate the type of scaling using the big-$\mathcal{O}$ notation, i.e. $\mathcal{O}(1/n)$ means that the error scales as $1/n$ up the leading order.

Importantly, even if the second laws are satisfied, there exists no general method of choosing catalysts for a given transformation. Moreover, not much is also known about thermodynamic properties of such catalysts, like their average energy, entropy or dimension. As such, there is still a lot to be understood about catalysis. In what follows we will show that in fact any catalyst composed of sufficiently many copies of an arbitrary state can catalyze any state transformation which is allowed by the transformation laws. We will refer to this phenomena as ``catalytic universality'' and demonstrate its appearance in the case of genuine catalysis and embezzlement. 

\section{Results}
\label{sec3}
In this section we present our main result, i.e. we prove that all multi-copy states can act as catalysts for all allowed transformations. We will prove the result by constructing a general protocol (see Fig. \ref{fig:1}) which will be then adapted to a specific regime of catalysis (embezzlement or genuine) by appropriately choosing the corresponding parameters. 

\begin{figure}
    \centering
    \includegraphics[width=0.5\textwidth]{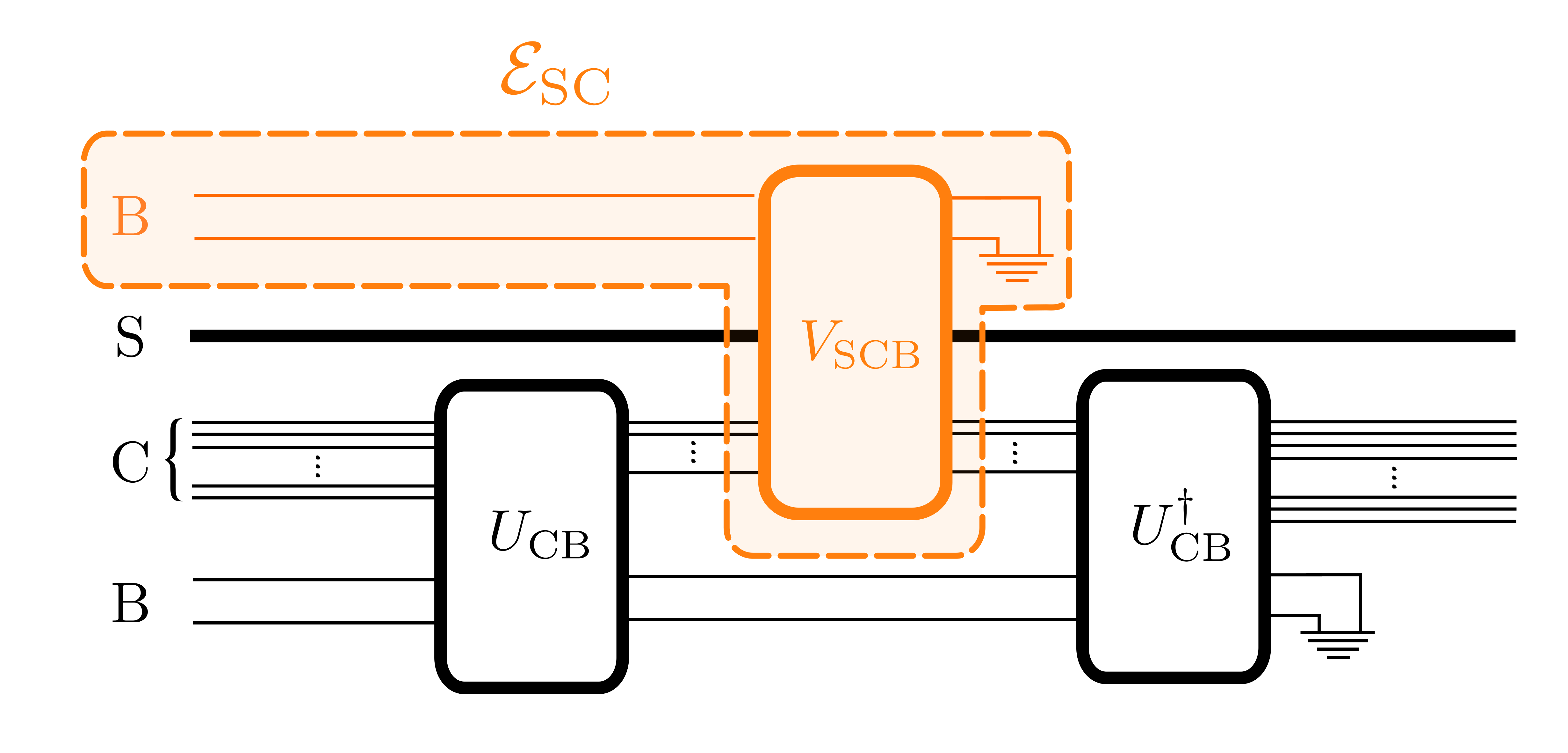}
    \caption{The main protocol. The catalyst $\rm C$ is processed using the unitary $U_{\rm CB}$ and then, along with the system $\rm S$, supplied as an input for the operation $\mathcal{E}_{\rm SC}(\cdot) := \Tr_{\rm B}[V_{\rm SCB}((\cdot)_{\rm SC} \ot \tau_{\rm B}^{\,})V_{\rm SCB}^{\dagger}]$. The resulting state of the catalyst is then transformed back using the unitary $U_{\rm CB}^{\dagger}$. As long as the back-action of the map  $\mathcal{E}_{\rm SC}$ on the catalyst is small, the protocol as a whole leaves the state of the catalyst approximately undisturbed.}
    \label{fig:1}
\end{figure}
\subsection{Intuition}
Before we describe our main protocol in full detail let us first qualitatively argue why multi-copy states can be seen as useful for catalysis. 

Consider the catalyst to be an $n$-copy state $\omega_{\rm C}^{\ot n}$ where the single-copy state $\omega_{\rm C}$ is an arbitrary state diagonal in the energy eigenbasis. Due to the law of large numbers \cite{Cover2006_book,Schumacher1995} in the asymptotic regime ($n \rightarrow \infty$) there exists a subset of eigenvalues of $\omega_{\rm C}^{\ot n}$, the so-called typical set, which carries almost the whole probability weight and is almost uniformly occupied. Such a state can be approximately reversibly converted into other states, at a rate quantified by the relative entropy. Importantly, this conversion can happen with a negligible error which vanishes quickly as the number of copies $n$ increases. 

In this way, when having access to a large number $n$ of copies of the state $\omega_{\rm C}^{\ot n}$ we can convert it almost reversibly into $m$ copies of another state which we can now finely-tune to our desired transformation. Once the catalyst is appropriately `preprocessed', we can apply the actual catalytic transformation and map $\rho_{\rm S}$ into $\sigma_{\rm S}$ with the help of the converted catalyst. Using the fact that for large $n$ the error is negligible and the conversion is almost reversible, we can approximately recover the initial state of the catalyst by applying a suitable reverse map, i.e. by `postprocessing' it. Such a combined transformation consisting of these three steps can be equivalently viewed as a valid thermal operation on the system and the catalyst. 

The surprising fact is that, for states which satisfy respective transformation laws, it is always possible to find an explicit intermediate state which can be used to catalyze a given state transformation with a sufficiently small or even no error. This is far from being obvious and solves one of the big problems of catalysis by explicitly determining the state which can catalyse a given transformation. %In particular, in the case of exact catalysis there are examples of state transformations which can be only realized given access to an infinitely big catalyst. It turns out that allowing for a small error can amend this potentially adverse situation, in the same time keeping the error small enough so that we remain in the desired regime of catalysis. 

We now present the two main theorems of this paper. Their proofs are based on constructing a specific protocol which formalizes the above reasoning and adapting it to the respective regime of catalysis. In the case when the error on the catalyst goes to zero with $n \rightarrow \infty$ slower than (\ref{def:emb_bnd}), the partial order between states vanishes and marks the emergence of the embezzlement regime. In this case any multi-copy state, provided that $n$ is large enough, can act as a catalyst, or more precisely, an embezzler, for any state transformation. This is formalized by the following theorem:  
\begin{theorem}
\label{thm1}
For any two states $\rho_{\rm S}$ and $\sigma_{\rm S}$ and any catalyst state $\omega_{\rm C}$ there exists a thermal operation $\mathcal{T}_{\rm SC}$ such that:
\begin{align}
    \mathcal{T}_{\rm SC}\left[\rho_{\rm S} \ot \omega_{\rm C}^{\ot n} \right] = \sigma_{\rm SC}',
\end{align}
so that the following holds:
\begin{align}
    \eps{C} &:= \norm{\Tr_{\rm S} \left[\sigma_{\rm SC}'\right] - \omega_{\rm C}^{\ot n}}_{1} \leq \mathcal{O}\left(\frac{1}{\sqrt{n}}\right), \\
    \eps{S} &:= \norm{\Tr_{\rm C} \left[\sigma_{\rm SC}'\right] - \sigma_{\rm S}}_{1} \leq \mathcal{O}\left(\frac{1}{n}\right).
\end{align}
The constants can be explicitly computed and are provided in the Appendix.
\end{theorem}
\noindent The next theorem is more interesting as it relates to the regime when the partial order between states does not fully vanish. In this case, as long as the second laws are satisfied, a sufficient number of copies of any state can catalyze any state transformation with an error scaling sub-exponentially with the number of catalyst particles, i.e. genuine catalysis regime: 
\begin{theorem}
\label{thm2}
Let $\rho_{\rm S}$ and $\sigma_{\rm S}$ be two states with corresponding representations $\bm{p} = \emph{diag}[\rho_{\rm S}]$ and $\bm{q} = \emph{diag}[\sigma_{\rm S}]$ which satisfy:
\begin{align}
    \label{eq:seq_laws2}
    F_{\alpha}(\bm{p}, \bm{g}) > F_{\alpha}(\bm{q}, \bm{g}) \qquad \forall\, \alpha \geq 0,
\end{align}
where $\bm{g} = \emph{diag}[\tau_{\rm S}]$. Then, for any catalyst state $\omega_{\rm C}$ with  $\bm{c} = \emph{diag}[\omega_{\rm C}]$ and sufficiently large $n$ there exists a thermal operation $\mathcal{T}_{\rm SC}$ such that:
\begin{align}
    \mathcal{T}_{\rm SC}\left[\rho_{\rm S} \ot \omega_{\rm C}^{\ot n} \right] = \sigma_{\rm SC}',  
\end{align}
and the errors on the the system and the catalyst satisfy:
\begin{align}
    \eps{C} &:= \norm{\Tr_{\rm S} \left[\sigma_{\rm SC}'\right] - \omega_{\rm C}^{\ot n}}_{1} \leq \mathcal{O}\left(e^{- n^{\kappa}}\right), \\
    \label{eq:err_sys_gcat}
    \eps{S} &:= \norm{\Tr_{\rm C} \left[\sigma_{\rm SC}'\right] - \sigma_{\rm S}}_{1} = 0,
\end{align}
where $\kappa \in (0, 1)$ can be chosen arbitrarily. The explicit constants are provided in the Appendix. 
\end{theorem}
\noindent The complete proofs of Theorems \ref{thm1} and \ref{thm2} are provided in the next section, with a few technical steps which we postpone to the Appendix.  
\subsection{The protocol}
\textit{\textbf{Preprocessing.}} Let $\rho_{\rm S}$ denote the initial state of the system. Our goal is to transform it into another state $\sigma_{\rm S}$ using $n$ copies of a catalyst $\omega_{\rm C}$ via thermal operations. The total state of the system, the catalyst and the heat bath is given by the product state:
\begin{align}
    \rho_{\rm SCB}^{(0)} = \rho_{\rm S} \ot \omega_{\rm C}^{\ot n} \ot \tau_{\rm B}^{\,}.
\end{align}
In the first step we transform $n$ copies of the catalyst $\omega_{\rm C}$ into $m$ copies of an arbitrary state $\eta_{\rm C}$. The particular form of this intermediate state will be specified later as it crucially depends on which regime of catalysis we choose. As we stated in the previous section, this conversion step can be accomplished with a sub-exponential error on the catalyst. To see this explicitly let us invoke the following result from Ref. \cite{Chubb_2019}:
\begin{lemma}[Multi-copy state conversion]
\label{lem:non-asymp}
\label{lemma_transf_error}
\label{eq:non-asymp}
There is a thermal operation $\mathcal{T}_{\rm C}$ which performs the transformation:
\begin{align}
    \label{eq:mapT}
    \mathcal{T}_{\rm C}[\omega^{\ot n}] = \widetilde{\eta}_m,  
\end{align}
such that: 
\begin{align}
 \delta(n) := \norm{\widetilde{\eta}_m - \eta^{\ot m}}_{1} \leq e^{-n^{\kappa}},
\end{align}
where $\kappa \in (0, 1)$ can be chosen arbitrarily and $m = n\cdot r_n$ with the conversion rate given by: 
\begin{align}
    \label{eq:rn_def}
    r_n = \frac{D(\omega||\tau)}{D(\eta||\tau)} - \mathcal{O}\left(\frac{1}{\sqrt{n^{1-\kappa}}}\right),
\end{align}
and $D(\rho||\sigma) := \Tr[\rho \log \rho] - \Tr [\rho \log \sigma]$ is the relative entropy.
\end{lemma} \noindent
In what follows we will not be directly  interested in the the thermal operation $\mathcal{T}_{\rm C}$ but in the unitary $U_{\rm CB}$ which generates it as in (\ref{def_to_eq}). In particular, we will compose this unitary alongside unitaries from the other two steps of the protocol to form a single thermal operation (which is different from just composing the thermal operations of the steps).  Moreover, recall that we are working with block-diagonal states; for such states any thermal operation can be realized using a `gentle' unitary which ($i$) permutes energy levels inside the subspaces of equal energy and ($ii$) leads to the same error on $\rm C$ and the joint system $\rm CB$ (see the Appendix for details). In particular, this second observation will allow us to reverse the action of the unitary in order to recover the original catalyst with a sufficiently small error.

It is important to emphasize that Lemma \ref{lem:non-asymp} \emph{does not} require that the number of catalyst particles $n$ goes to infinity. In fact, the lemma is valid in the intermediate regime, i.e. when the number of copies is large but still finite. In that case the conversion rate $r_n$ will generally be smaller than the asymptotic one ($r_{\infty}$). However, this discrepancy can be quantified by looking at the second order corrections (see the Appendix and Ref. \cite{Chubb_2019} for more details). For our purposes this means that in order to achieve the rate $r_{n} \approx r_{\infty}$ \emph{and} the error scaling $\approx e^{-n^{\kappa}}$ it is enough to consider large but not infinitely large $n$. What ``large'' here means depends on the states involved in the transformation and can be found by estimating the second order corrections in (\ref{eq:rn_def}), i.e. we can say that $n$ is large  when $\mathcal{O}(1/\sqrt{n^{1-\kappa}}) \ll 1$. Naturally, including more terms in the expansion would allow for a better estimate of the exact rate, and even smaller numbers of copies comprising the catalyst.

Let us now consider the unitary  $U_{\rm CB}$ that implements the map that performs the transformation (\ref{eq:mapT}). There are many unitaries which can implement this map - for our purposes we use the one which does not propagate the error to the environment, i.e. satisfies the property ($ii$). This in turn implies:
\begin{align}
    \label{eq:err_del}
    \norm{\Tr_{\rm B} \mathcal{U}_{\rm CB}[\omega^{\ot n}_{\rm C} \ot \tau_{\rm B}]  - \eta_{\rm C}^{\ot m}}_1 \leq \delta(n),
\end{align}
where $\mathcal{U}_{\rm CB}[\cdot] := U_{\rm CB}(\cdot)U_{\rm CB}^{\dagger}$. The state of the system, the catalyst and the environment after the first step of the protocol is given by:
\begin{align}
    \rho_{\rm SCB}^{(1)} =  (\mathcal{I}_{\rm S} \ot \mathcal{U}_{\rm CB})[\rho_{\rm SCB}^{(0)}].
\end{align}
If we now trace out the bath, the state of the system and the catalyst for large $n$ will be close to $\rho_{\rm S} \ot \eta_{\rm C}^{\ot m}$, with $m = r_n \cdot n$ being linearly proportional to $n$. That is, the conversion rate  $r_n$ up to the leading order depends only on the states $\omega_{\rm C}$, $\eta_{\rm C}$ and their respective thermal states. The precise form of the state $\eta_{\rm C}$ will be specified later when we focus on specific regimes of catalysis.

\textit{\textbf{Catalytic transformation.}}
In the next step we use the preprocessed catalyst state to facilitate the main thermal operation $\mathcal{E}_{\rm SC}$ which transforms $\rho_{\rm S}$ into $\sigma_{\rm S}$ and (potentially) perturbs the state of the catalyst. Depending on the acceptable size of the backaction on the catalyst we will construct the transformation $\mathcal{E}_{\rm SC}$ differently. With this in mind, the total state of the system $\rm SCB$ after this step reads:
\begin{align}
    \rho_{\rm SCB}^{(2)} = (\mathcal{E}_{\rm SC} \ot \mathcal{I}_{\rm B})[\rho_{\rm SCB}^{(1)}].
\end{align}
As this can differ from our target state $\sigma_{\rm S} \ot \eta_{\rm C}^{\ot m}$, we also define the transformation error induced on the catalyst in this step to be:
\begin{align}
    \label{eq:err_nu}
    \nu(n) := \norm{\Tr_{\rm S}\mathcal{E}_{\rm SC}[\rho_{\rm S} \ot \eta_{\rm C}^{\ot m}] - \eta_{\rm C}^{\ot m}}_{1}.
\end{align}
Let us note that this is also the only time in the protocol in which we modify the state of the system $\rm S$. Hence, the transformation error on the system can be entirely associated with the map $\mathcal{E}_{\rm SC}$. 

\textit{\textbf{Postprocessing.}}
In the last step we apply the inverse of the unitary transformation which we applied in the preprocessing stage. As a result, the initial state of the catalyst should be approximately recovered, given that it was not perturbed too much during the previous step. To do so we simply apply the inverse unitary channel $\mathcal{U}^{\dagger}_{\rm CB}$ to the joint state of the catalyst and the heat bath. Importantly, in order to do so we use the part of the heat bath which had only interacted with the catalyst in the preprocessing stage (see Fig. \ref{fig:1}). This allows the catalyst to be transformed back to its initial state again with a small error. In this way the final state of the three systems becomes:
\begin{align}
    \rho_{\rm SCB}^{(3)} =  (\mathcal{I}_{\rm S} \ot \mathcal{U}^{\dagger}_{\rm CB})[\rho_{\rm SCB}^{(2)}].
\end{align}
Notice that this reversal is possible only if we keep the state of the bath from the preprocessing step. This implies that the correlations with the heat bath which are created during the preprocessing step play an important role in the whole protocol and allow to keep the final error on the catalyst acceptably small. Crucially, this composed protocol is still a thermal operation. 

\textit{\textbf{Analysis of the protocol.}}
We now move on to quantifying the total disturbance induced on the subsystems. The final state of the catalyst can be written as: 
\begin{align}
    \omega_{\rm C}' &= \Tr_{\rm SB}[\rho_{\rm SCB}^{(3)}] \\
    &= \Tr_{\rm B}\left[\mathcal{U}^{\dagger}_{\rm CB} \left(\mathcal{E}^{(\rho)}_{\rm C} \ot \mathcal{I}_{\rm B}\right) \mathcal{U}_{\rm CB}[\omega_{\rm C}^{\ot n} \ot \tau_{\rm B}]\right],
\end{align}
where we labelled the effective channel which acts on the catalyst with $\mathcal{E}_{\rm C}^{(\rho)}[\omega] := \Tr_{\rm S}\mathcal{E}_{\rm SC}[\rho_{\rm S} \ot \omega_{\rm C}]$. In a similar way we find the final state of the system $\rm S$ to be:
\begin{align}
    \rho_{\rm S}' = \Tr_{\rm CB}[\rho_{\rm SCB}^{(3)}] = \Tr_{\rm C} \mathcal{E}_{\rm SC}\left[\rho_{\rm S} \ot \mathcal{T}_{\rm C}[\omega_{\rm C}^{\ot n}]\right],
\end{align}
where $\mathcal{T}_{\rm C}$ is the thermal operation which performs the conversion from Lemma \ref{lem:non-asymp} and is related to the unitary $U_{\rm CB}$ via (\ref{def_to_eq}). With this in mind we can now find the final transformation errors on the system $\rm S$ and the catalyst $\rm C$. Recall that our main goal was to perform the transformation $\rho_{\rm S} \rightarrow \sigma_{\rm S}$ while keeping the $n$-copy catalyst state $\omega_{\rm C}^{\ot n}$ approximately undisturbed. Using the triangle inequality (see Appendix for details) and standard algebraic manipulations it can be shown that the errors on the system and the catalyst satisfy:
\begin{align}
    \label{def:eps_S}
    \eps{S} &:= \norm{\rho_{\rm S}' - \rho_{\rm S}}_1, \\
    \label{def:eps_C}
    \eps{C} &:= \norm{\omega_{\rm C}' - \omega_{\rm C}^{\ot n}}_{1} \leq 2  \delta(n) + \nu(n),
\end{align}
where $\delta(n)$ is related to the unitary $U_{\rm SB}$ used in the pre- and postprocessing steps (\ref{eq:err_del}) and $\nu(n)$ comes from the channel $\mathcal{E}_{\rm SC}$ (\ref{eq:err_nu}). Notice that the two contributions in (\ref{def:eps_C}) account for all possible types of error incurred on the catalyst during the main protocol. In particular, the term $\delta(n)$ quantifies both the error due to the failure in preparing the desired intermediate state $\eta_{\rm C}^{\ot m}$ using unitary $U_{\rm CB}$ in the preprocessing step, as well as the failure in reversing the action of the unitary $U_{\rm CB}$ in the postprocessing step. This can happen for example when the action of the the channel $\mathcal{E}_{\rm SC}$ in the catalytic step significantly disturbs the catalyst. On the other hand, the term $\nu(n)$ is related solely to the disturbance applied to the catalyst in the catalytic step and, as we shall see, changes depending on the specific regime of catalysis one is interested in.

So far we have treated the intermediate state $\eta_{\rm C}$ and the map $\mathcal{E}_{\rm SC}$ as parameters of the main protocol. Depending on the specific choice of these parameters we can now address different regimes of catalysis. In the remaining part of the paper we will specialize the above protocol first to the embezzlement regime and then to the regime of genuine catalysis. 

\subsection{Embezzlement regime}
We begin by applying our protocol to the embezzlement regime, i.e. when the partial order between states completely vanishes and all transformation become possible. 
Even though embezzlement is not a proper form of catalysis, it is still an interesting phenomenon which has found several important applications mostly in the resource theory of pure-state entanglement. The power of embezzling has been exploited in several areas of quantum information, such as coherent state exchange protocols \cite{leung2008coherent} or entangled projection games  \cite{dinur2013parallel}. Moreover, embezzlement can be also viewed as a protocol hiding quantum states from external observers. With this interpretation it was used to prove the quantum version of the reverse Shannon theorem \cite{Berta_2011,Bennett_2014}.

Today embezzlement is still a mysterious concept and its full significance in a thermodynamic context still constitutes an important open problem. Even though its role is not well understood, quantifying which states can be used as embezzlers is a very relevant problem. Recent studies revealed a few families of universal embezzling states, both in the context of the resource theory of entanglement \cite{leung2008coherent} and thermodynamics \cite{Ng_2015}. Such universal embezzlers have the power to ''catalyze'' any state transformation. In the case of the resource theory of entanglement, a further study exposed another family of universal embezzling states and some of their properties where examined in  \cite{leung2014characteristics}. In general, however, very few such families of universal embezzlers are known, and the effects related to the dimension, entropy or energy of the embezzler have been hardly studied. 

The main technical tool we will make use of in this section is the \emph{convex-split lemma}, a recently discovered result from quantum communication theory. This is an important tool which allows us to prove that $n$ copies of \emph{any} state can serve as a universal embezzler, i.e. can help in facilitating any state transformation. As we will see, in this case the error on the catalyst, up to the leading order, is proportional to $1/ \sqrt{n}$ and so approaches zero as $n \rightarrow \infty$. Since this rate is slower than the boundary specified by (\ref{def:emb_bnd}), the transformation should be associated with the embezzlement regime.

We begin by specifying the preprocessed catalyst state $\eta_{\rm C}$. We choose it to be precisely equal to the target state of the system $\sigma_{\rm S}$, that is:
\begin{align}
    \eta_{\rm C}^{\ot m} = \sigma_{\rm C}^{\ot m}.
\end{align}
In order to specify the channel for the catalytic step, $\mathcal{E}_{\rm SC}$, let us consider the following mixing process acting on an $(m+1)-$partite system:
\begin{align}
    \label{eq:mix_perm}
   \mathcal{T}^{(\text{mix})} [\cdot] = \frac{1}{m+1} \sum_{i=0}^{m} S_{(0, i)} (\cdot) S_{(0,i)}^{\dagger},
\end{align}
where $S_{(i, j)}$ is a unitary which swaps subsystems $i$ and $j$ leaving all the remaining subsystems untouched, i.e.:
\begin{align}
    S_{(i, j)} \ket{\ldots a_i, \ldots a_j, \ldots} = \ket{\ldots a_j, \ldots a_i, \ldots}.
\end{align}
This type of operation naturally preserves the thermal state (or any other state which is combined of identically distributed and independent copies), hence it is also a valid thermal operation. We will apply this map to the state of the system $\rho_{\rm S}$ (treated as the zeroth subsystem) and $m$ copies of the preprocessed catalyst state $\eta_{\rm C} = \sigma_{\rm C}$ (treated as the remaining $m$ subsystems). The resulting state is:
\begin{align}
    \label{eq:con_spl_sta}
    \mathcal{T}_{\mathrm{SC}} ^{\text{(mix)}}\left[\rho_{\mathrm{S}} \ot \sigma_{\mathrm{C}}^{\ot m}\right] \!= &\frac{1}{m+1}\Big(\rho_{\mathrm{S}} \ot \sigma_{\mathrm{C}_{1}} \ot \ldots \ot \sigma_{\mathrm{C}_{m}}+ \ldots \nonumber
    \\ 
    & \ldots + \sigma_{\mathrm{S}} \ot \sigma_{\mathrm{C}_{1}} \ot \ldots \ot \rho_{\mathrm{C}_{m}}\Big), 
\end{align}
where now $\mathrm C = \mathrm C_1 \mathrm C_2 \ldots \mathrm C_m$. It can be easily verified (see the Appendix) that choosing this particular transformation in the main protocol, i.e. $\mathcal{E}_{\rm SC} = \mathcal{T}_{\rm SC}^{\text{(mix)}}$, allows to bound the error term $\nu(n)$ as:
\begin{align}
    \label{eq:nu_conv_split_mt}
    \nu(n) & \leq \norm{ \mathcal{T}^{\text{(mix)}}_{\rm SC} [\rho_{\rm S} \ot \sigma_{\rm C}^{\ot m}] - \sigma_{\rm S} \ot \sigma_{\rm C}^{\ot m}}_{1}, 
\end{align}
Let us now present the main technical tool of this section, that is the convex-split lemma adapted from Ref. \cite{Anshu_2017}:
\begin{lemma}[Convex-split]
\label{lemma2}
Let $\rho$ and $\sigma$ be two quantum states satisfying $\emph{supp}(\rho) \subseteq \emph{supp}(\sigma)$. Then for any $m \geq 1$ it holds that: 
\begin{align}
    \norm{\mathcal{T}_{\mathrm{SC}} ^{\emph{(mix)}}\left[\rho_{\mathrm{S}} \ot \sigma_{\mathrm{C}}^{\ot m}\right] - \sigma_{\rm S} \ot \sigma_{\rm C}^{\ot m}}_{1}^2 \leq  \frac{2^{D_{\infty}(\rho||\sigma)}}{m} .
\end{align}
where $D_{{\infty}}(\rho||\sigma)$ is the max-relative entropy corresponding to $\alpha \rightarrow \infty$ in the definition (\ref{def:rel_ent}).
\end{lemma}

\noindent Using the lemma from above we can bound the error term from (\ref{eq:nu_conv_split_mt}) as:
\begin{align}
 \nu(n) &\leq  \sqrt{\frac{2^{D_{\infty}(\rho||\sigma)}}{r_n}}  \cdot \frac{1}{\sqrt{n}} \\
 &= c_n \cdot \frac{1}{\sqrt{n}},
 \end{align}
where $c_n$ for fixed input and output states $\rho$ and $\sigma$ is effectively constant for large $n$. Consequently, the error term $\nu(n)$ scales with $n$ as $\nu(n) \sim n^{-1/2}$ and using (\ref{def:eps_C}) we find that:
\begin{align}
    \eps{C} \leq \mathcal{O}\left( \frac{1}{\sqrt{n}}\right)
\end{align}
Note that this procedure does not assume anything particular about the states $\rho$, $\sigma$ and $\omega$. The only technical assumption is that $\text{supp}(\rho) \subseteq \text{supp}(\sigma)$, which can be achieved for any states $\rho$ and $\sigma$ by an arbitrarily small perturbation. This means that by choosing sufficiently large $n$ we can carry out \emph{any} state transformation on the system, at the same time keeping an arbitrarily small error on the catalyst. In this way we have proven Theorem \ref{thm1} from the beginning of this section.  

\subsection{Genuine catalysis regime}
Let us now consider the other regime of inexact catalysis, that is when the error scales with the catalyst dimension slow enough to maintain the partial order between states. We will again apply our main protocol with a specific choice of the intermediate state $\eta_{\rm C}$ and transformation $\mathcal{E}_{\rm SC}$. 

Recall that we can represent the states $\rho_{\rm S}$ and $\sigma_{\rm S}$ with probability vectors $\bm{p}$ and $\bm{q}$ as described in (\ref{eq:res_rep}). We are going to assume that these two states satisfy the second laws of thermodynamics (\ref{eq:sec_laws}). This means that there exists a catalyst which can be used to transform $\bm{p}$ into $\bm{q}$ without any disturbance on the catalyst. Interestingly, this is not the only implication guaranteed by the second laws. It turns out that these family of conditions also provide a sufficient condition for a multi-copy state transformation between $\rho_{\rm S}$ and $\sigma_{\rm S}$. This is captured by a recent result from majorization theory \cite{jensen2019asymptotic} (Proposition 3.2.7). For our purposes this result can be adapted to the thermodynamic case by using the so-called embedding map introduced in \cite{Brand_o_2015}. Intuitively, the embedding map is is an operation which allows to be translated between the microcanonical and macrocanonical descriptions of a thermodynamic system. Leaving the technical details of the proof to the Appendix % \footnote{In \cite{jensen2019asymptotic} a slightly different version of Lemma \ref{lemma3} was proven in which the inequalities (\ref{lem3:ineq}) are strict. However, using a reasoning similar as in \cite{Brand_o_2015} it can be shown that the Lemma is also true if we replace strict inequalities with non-strict ones.}
, here we just present the lemma which we adapted from \cite{jensen2019asymptotic}:
\begin{lemma}
\label{lemma3}
Let $\rho_{\rm S} = \emph{diag}[\bm{p}]$ and $\sigma_{\rm S} = \emph{diag}[\bm{q}]$ be two quantum states with dimension $d_{\rm S}$, Hamiltonian $H_{\rm S}$ and a corresponding thermal state $\tau_{\rm S}^{\,} = \emph{diag}[\bm{g}]$ such that:
\begin{align}
    \label{lem3:ineq}
    F_{\alpha}(\bm{p}, \bm{g}) \geq F_{\alpha}(\bm{q}, \bm{g}) \qquad \forall \, \alpha \geq 0.
\end{align}
Then, for sufficiently large $k$, the following holds:
\begin{align}
\label{eq:mcopy}
\bm{p}^{\ot k} \succ_{\rm T} \bm{q}^{\ot k}.
\end{align}
\end{lemma}
\noindent The Lemma ensures that a $k$-copy transformation between $\rho_{\rm S}$ and $\sigma_{\rm S}$ is possible for a finite $k$ if the second laws are satisfied for all $\alpha \geq 0$. The second ingredient which we are going to use is a special catalyst state which can ``simulate'' any $k$-copy state transformation. It turns out that if $k$ copies of the state $\rho_{\rm S}$ can be transformed into $k$ copies of another state $\sigma_{\rm S}$, then there exists a special state that can be used as a catalyst when transforming only a single copy of $\rho_{\rm S}$ into a single copy of $\sigma_{\rm S}$. With this in mind, let us introduce the \emph{Duan state} $\omega^{\rm D}_k(\rho, \sigma)$ \cite{PhysRevA.71.062306,PhysRevA.74.042312} to be a state of the form:
\begin{align}
    \label{def:duan}
    \omega^{\rm D}_k(\rho, \sigma) :=\! \frac{1}{k} \sum_{i=1}^k  \sigma^{\ot k-i} \ot \rho^{\ot i-1} \ot \dyad{i}.
\end{align}
Then, as shown in \cite{PhysRevA.71.062306}, the Duan state $\omega^{\rm D}_k(\rho, \sigma)$ can be used as a catalyst which ``simulates'' the $k$-copy transformation between $\rho_{\rm S}$ and $\sigma_{\rm S}$. This transformation is exact, in the sense that if $\rho_{\rm S}^{\ot k} \rightarrow \sigma_{\rm S}^{\ot k}$ then:
\begin{align}
    \rho_{\rm S} \ot \omega^{\rm D}_k(\rho, \sigma) \rightarrow \sigma_{\rm S} \ot \omega^{\rm D}_k(\rho, \sigma).
\end{align}
The above can be easily verified (see the Appendix for the details). The dimension of the Duan state is a function of both $k$ and $d_{\rm S}$, that is $\text{dim}[\omega^{\rm D}_k(\rho, \sigma)] = k d_{\rm S}^{k-1}$ and hence, for a given $\rho_{\rm S}$ and $\sigma_{\rm S}$, is constant and does not scale with $n$.

In this way, once the second laws are satisfied, there always exists a large enough (and finite) integer $k$ such that $k$ copies of $\rho_{\rm S}$ can be converted into $k$ copies of $\sigma_{\rm S}$. This, on the other hand, means that there always exists a special catalyst (the Duan state) which can catalyze the transformation from $\rho_{\rm S}$  to  $\sigma_{\rm S}$  \emph{without} any disturbance on the catalyst. We formalize this in the following theorem:
\begin{theorem}
\label{thm3}
Let $\rho_{\rm S} = \emph{diag}[\bm{p}]$ and $\sigma_{\rm S} = \emph{diag}[\bm{q}]$ be two quantum states with dimension $d_{\rm S}$, Hamiltonian $H_{\rm S}$ and a corresponding thermal state $\tau_{\rm S}^{\,} = \emph{diag}[\bm{g}]$ such that:
\begin{align}
    \label{sec_laws_2}
    \forall \, \alpha \geq 0 \qquad F_{\alpha}(\bm{p}, \bm{g}) \geq F_{\alpha}(\bm{q}, \bm{g}).
\end{align}
then, for sufficiently large $k$, the following holds:
\begin{align}
    \label{eq:duan_transf}
    \rho_{\rm S} \ot \omega^{\rm D}_k(\rho, \sigma) \xrightarrow{\emph{TO}} \sigma_{\rm S} \ot \omega^{\rm D}_k(\rho, \sigma),
\end{align}
where $\omega^{\rm D}_k(\rho, \sigma)$ is the Duan state defined in (\ref{def:duan}).
\end{theorem}

Notice that the above Theorem tells us explicitly how to find a good catalyst when the second laws hold. Indeed, once the conditions in (\ref{sec_laws_2}) are satisfied we now have a method of choosing the catalyst state which can facilitate a given transformation.

Let us now return to our main protocol and choose the intermediate state $\eta_{\rm C}$ to be precisely the Duan state corresponding to $\rho_{\rm S}$ and $\sigma_{\rm S}$, that is:
\begin{align}
    \eta_{\rm C} = \omega^{\rm D}_k(\rho, \sigma).
\end{align}
In fact we only need one copy of the state $\eta_{\rm C}$ to transform $\rho_{\rm S}$ into $\sigma_{\rm S}$ on the system $\rm S$ and so any $m \geq 1$ will lead to a valid transformation. However, now we have to assure that our preprocessing step applied to the initial catalyst state produces at least one copy of the respective Duan state. If we recall that $m = n \cdot r_n$, then $m \geq 1$ can be guaranteed if:
\begin{align}
    % n \geq r_n^{-1} \approx \frac{D\left(\omega^{\rm D}_k(\rho, \sigma)\big|\big|\frac{1}{D}\mathbb{1}_{D}\right)}{D\left(\omega_{\rm C} \big|\big|\frac{1}{d_{\rm C}}\mathbb{1}_{d_{\rm C}}\right)},
    n \geq r_n^{-1} \approx \frac{\log D - H(\omega_k^{\rm D}(\rho, \sigma))}{\log d_{\rm C} - H(\omega_{\rm C})},
\end{align}
where $H(\rho) := -\Tr[\rho \log \rho]$ is the von Neuman entropy and $D = k \cdot d_{\rm S}^{k-1}$ is the dimension of the corresponding Duan state.

Now, from Theorem \ref{thm3} we know that there exists a TO which transforms $\rho_{\rm S}$ into $\sigma_{\rm S}$ using the Duan state as in (\ref{eq:duan_transf}). This transformation is exact, and so it does not introduce any additional error in the main protocol. In this way the only disturbance of the catalyst is applied by the pre- and postprocessing steps. In this way any $n$-copy catalyst state $\omega$ can be used to catalyze any allowed state transformation, i.e. any transformation which obeys the family of the second laws (\ref{eq:seq_laws2}). This concludes the proof of Theorem \ref{thm2}.

From a practical point of view it is interesting to see what is the typical value of $k$ such that catalytic universality holds. Indeed, by Lemma \ref{lemma3} we know that the second laws imply the existence of a finite $k$ and a thermal operation such that $\rho^{\ot k}_{\rm S} \rightarrow \sigma_{\rm S}^{\ot k}$. It turns out that for many pairs of states a multi-copy transformation exists even for small values of $k$, which implies that the dimension of the associated Duan state $\omega_k^{\rm D}(\rho, \sigma)$ is often not too large. In such situations the main protocol does not require having control over an asymptotically large Hilber space, as might be expected \emph{a priori}. This can be seen by looking at Fig. \ref{fig:3} which estimates the fraction of all transitions which can be realised by having only a small number of copies $k$. We will return to the discussion of the problem of the amount of control required by the main protocol in Sec. \ref{sec:control}.

\begin{figure}
    \centering
    \includegraphics[width=\linewidth]{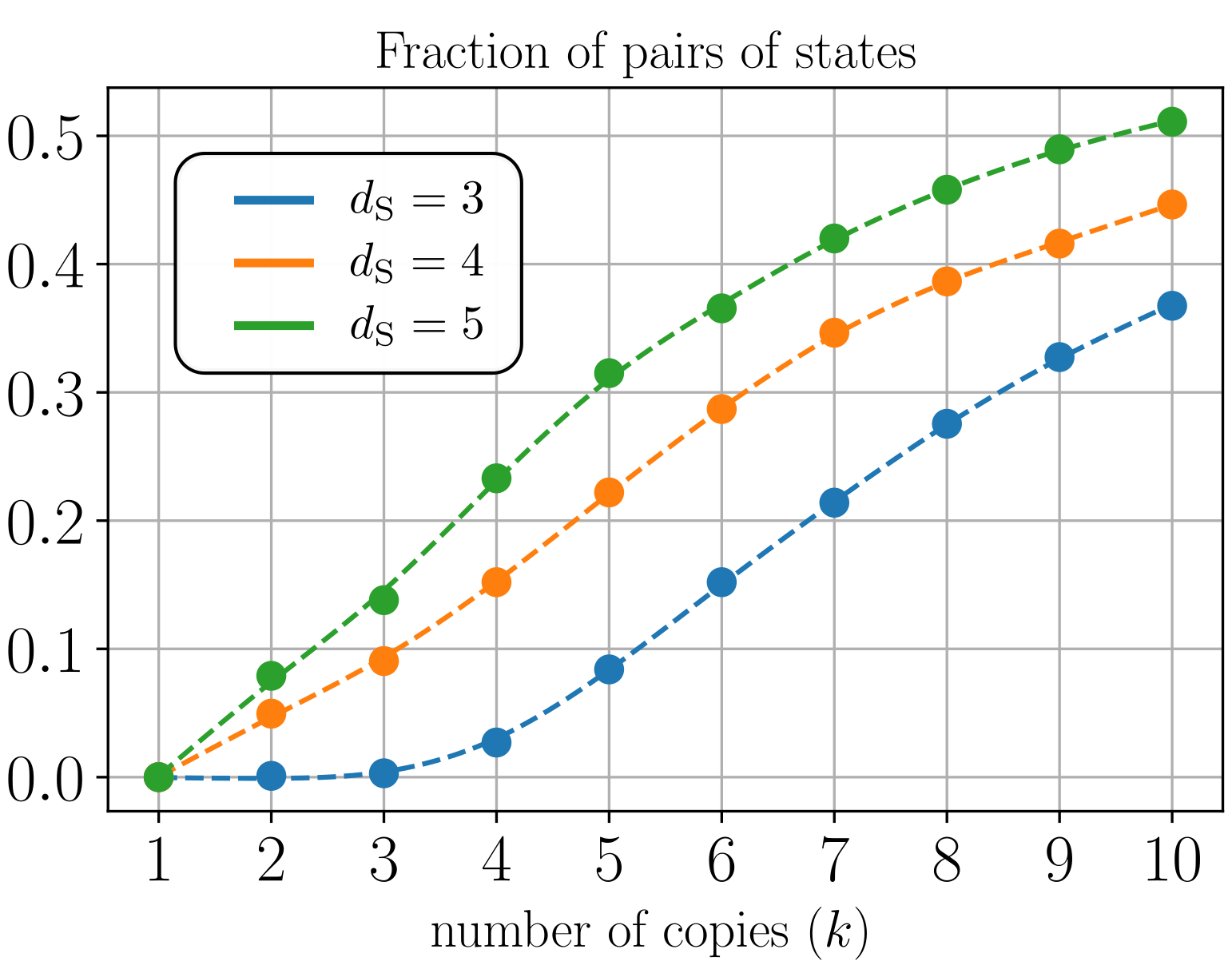}
    \caption{An estimate of the volume of all block-diagonal pairs of states $(\rho_{\rm S}, \sigma_{\rm S})$ with $H_{\rm S} \propto \mathbb{1}_{\rm S}$ for which there is a transformation taking $\rho_{\rm S}^{\ot i}$ into $\sigma_{\rm S}^{\ot i}$ for some $i \leq k$, divided by the volume of all block-diagonal pairs which can be catalytically activated. Catalytically activated pairs are such that they $(i)$ satisfy the respective second laws and $(ii)$ $\rho_{\rm S}$ cannot be directly converted into $\sigma_{\rm S}$.}
    \label{fig:3}
\end{figure}

\subsection{Practical aspects of the main protocol}
\label{sec:control}

We have described a general protocol which allows for the use of multiple copies of any non-equilibrium state to catalyse any state transformation which is allowed by the second laws (\ref{eq:sec_laws}). Still, in order to show that catalytic universality is a relevant feature of single-shot quantum thermodynamics, a similar behavior must also appear outside of the idealised assumptions of our framework. In other words, the catalytic universality should not be merely an artifact of the idealized model but be robust to its reasonable relaxations.

One might raise two natural and fair objections about the methods we used to demonstrate catalytic universality. First, one might be worried that the set of transformations which we consider here (thermal operations) is too large. This is because, at least in principle, it contains all possible energy-preserving unitaries, among them also those which might be impossible to implement in practice. Such unitaries might require the experimenter to accurately manipulate all of the particles contained in a typical heat bath, which is clearly beyond the scope of even the most passionate experimenter. Second, our protocol essentially requires that the number of catalyst particles $n$ is \emph{large enough}, so that the preprocessing step can prepare a sufficiently large Duan state. Again, on might be worried that for practically relevant situations this number will be unrealistically large, making any reasonable implementation of the protocol infeasible. These two problems must be resolved before catalytic universality can be seen as a practical and relevant feature of single-shot quantum thermodynamics. In the remaining part of this section we address these two important issues separately.

\subsubsection{How much control is needed by thermal operatons?}
The framework of thermal operations proved to be a useful approach in deriving \emph{fundamental} limitations for general thermodynamic processes. In parallel, the practical applicability of the framework was criticised for its requirement of accessing a huge number of degrees of freedom involved in the dynamics, even for small quantum systems (see e.g. Ref. \cite{Lostaglio2018} for a discussion). Indeed, the unitaries from Eq. (\ref{def_to_eq}) are required to be energy-conserving, but are otherwise arbitrary. Such unitaries may, in principle, act on a large number of energy levels of the heat bath at once. This is beyond the scope \chg{of} any realistic scenario, as realising such unitaries might require the experimenter to access and manipulate an overwhelming number of different parameters.

Since thermal operations allow for such unitaries, we cannot assure \emph{a priori} that the results derived within this formalism are relevant also beyond the paradigm of infinite control. Thankfully, this important problem was recently \chg{mitigated} in Ref. \cite{Perry2018}, where it was shown that any transformation of the form (\ref{def_to_eq}) can be realised using much simpler class of operations called Crude Thermal Operations.  These thermodynamic transformations require \chg{much less} control and consist of three types of operations:
\begin{enumerate}
    \item \textbf{Partial level thermalisations}, which selectively thermalise any two energy levels of the system. This type of operations can be implemented either by performing a partial SWAP between the system and the bath, or by selectively putting the system's energy levels in contact with the bath, preceded by filtering out irrelevant frequencies.  
    \item \textbf{Level transformations}, which either raise or lower any two energy levels of the system’s Hamiltonian. This is a standard type of transformations considered within thermodynamics, whose implementation generally requires performing work. Interestingly, it is enough to consider only those level transformations which do not require work. 
    \item \textbf{Subspace rotations}, which involve energy conserving unitaries acting upon the system only. For example, when the system has degenerate energy levels one may apply such unitaries within the degenerate subspace to rotate the state into a state diagonal in a specific basis. 
\end{enumerate}
The main result of Ref. \cite{Perry2018} states that, for states block-diagonal in the energy eigenbasis, any transformation of the form (\ref{def_to_eq}) can be implemented using the above three types of operations applied to the system and a single thermal qubit. In other words, even though the framework of thermal operations allows for unitaries with arbitrary complexity and level of control, such unitaries cannot provide any thermodynamic advantage over their much simpler counterparts. \chg{However, let us also remind that although Ref. \cite{Perry2018} shows that thermal operations can be decomposed into a sequence of simpler unitaries, even this might not be easy to implement in practice. Firstly, the actual sequence of crude thermal operations may be very long and thus experimentally infeasible; secondly, the three basic types of crude operations (1-3) might still be difficult to realise in a practical scenario. This is mainly because realising partial level thermalisations (1) and level transformations (2) require, in general, a considerable amount of effort. Finally, the approach presented in Ref. \cite{Perry2018} does not seem to be easily scalable and therefore it might still be necessary to search for alternative solutions that could bring thermal operations closer to their experimental implementations. }

Let us now return to our main problem. Translating the result of Ref. \cite{Perry2018} for our purposes means that the unitaries $U_{\rm CB}$ and $V_{\rm SCB}$ which we apply in the main protocol (see Fig. \ref{fig:1}) require only a coarse-grained control over the joint state of the system and the catalyst ($\rm SC$) and the heat bath ($\rm{B}$).    

\subsubsection{How large the catalyst must be?}
Unfortunately, the argument from the previous section only partially solves the problem of control in the main protocol. Although the transformation can be realized using simple interactions with a small region of the heat bath, it may still be necessary to manipulate a large number of particles composing the multi-copy catalyst. % One way to circumvent this problem is to show that catalytic universality can be demonstrated for a moderate number of particles comprising the catalyst.

Let us study this problem in more detail. Notice that before the unitary $V_{\rm SCB}$ from the main protocol can be applied, the associated Duan state $\omega_{k}^{\rm D}(\rho, \sigma)$ must be determined. This requires finding a number $k$ such that $k$ copies of $\rho_{\rm S}$ can be converted into $k$ copies of $\sigma_{\rm S}$ (see Eq. (\ref{eq:mcopy})) . Due to Lemma \ref{lemma2} we know that, for states satisfying respective second laws, such a $k$ (however large) always exists. This assures that the unitary $V_{\rm SCB}$ can realize the catalytic transformation. However, it might still be difficult to apply this unitary since the required number of copies $k$ can be large, and hence also the dimension of the associated Duan state. Consequently, both unitaries $U_{\rm CB}$ and $V_{\rm SCB}$ might need to act on a large number of energy levels of joint state of the system and the catalyst.

Let us explain the origin of this problem. The second laws of thermodynamics provide a mathematically clean and succinct characterisation of allowed transformations in the resource theory of quantum thermodynamics. However, the generality of this characterisation also leads to certain drawbacks. In particular, there are transformations which are allowed by the second laws, but at the same time the catalyst they require has to be infinitely large (see Proposition $2$ from Ref. \cite{Sanders_2009}). This situation is far from being physical and demonstrates that certain transformations, although technically allowed by the second laws, can never be realised using practical catalysts. This implies that in such instances our main protocol may also require asymptotically large catalysts.

This problem can be circumvented if we sacrifice the elegant mathematical description in terms of the second laws for a more operational characterisation. In particular, instead of using $\alpha-$generalised free energies, we can describe the partial order of states using the concept of $k$-copy transformations. This amounts to considering only those pairs of states $(\rho_{\rm S}, \sigma_{\rm S})$ for which there is a transformation that takes $\rho_{\rm S}^{\ot i}$ into $\sigma_{\rm S}^{\ot i}$ for some $i \leq k$.  Interestingly, the partial orders described by the second laws and $k$-copy transformations (for all $k$) are exactly the same, which can be seen as a direct consequence of Theorem $1$ from Ref. \cite{PhysRevA.71.062306} and Lemma \ref{lemma3}. For a fixed $k$ this characterisation involves checking at most $d_{\rm S}^2 (d_{\rm S}^{k-1}-1) / (d_{\rm S}-1)$ conditions, which can be done e.g. by looking at the elbows of the thermo-majorization curves. We emphasize that the price to pay for considering a \emph{fixed} $k$ is that we limit the set of all possible transformations, as compared to those allowed by the second laws. At the same time, a large fraction of all transformations allowed by the second laws can be $k$-copy transformed for a relatively small $k$ (see Fig. \ref{fig:3}). Using the concept of $k$-copy transformations gives rise to the following alternative version of Theorem \ref{thm2}: 

\newtheorem*{T1}{Theorem~2 (alternative)}
\begin{T1}
Let $\rho_{\rm S}$ and $\sigma_{\rm S}$ be two states for which there exists a $k$-copy transformation such that:
\begin{align}
    \rho_{\rm S}^{\ot k} \rightarrow \sigma_{\rm S}^{\ot k},
\end{align}
Then, for any catalyst state $\omega_{\rm C}$ and for a sufficiently large $n$ satisfying: 
\begin{align}
     n \geq \frac{\log D - H(\omega_k^{\rm D}(\rho, \sigma))}{\log d_{\rm C} - H(\omega_{\rm C})}
\end{align}
there is a thermal operation $\mathcal{T}_{\rm SC}$ such that:
\begin{align}
    \mathcal{T}_{\rm SC}\left[\rho_{\rm S} \ot \omega_{\rm C}^{\ot n} \right] = \sigma_{\rm SC}',  
\end{align}
and $\Tr_{\rm C} \left[\sigma_{\rm SC}'\right] = \sigma_{\rm S}$ with the following disturbance of the catalyst:
\begin{align}
    \eps{C} &:= \norm{\Tr_{\rm S} \left[\sigma_{\rm SC}'\right] - \omega_{\rm C}^{\ot n}}_{1} \leq \mathcal{O}\left(e^{- n^{\kappa}}\right),
\end{align}
where $\kappa \in (0, 1)$ can be chosen arbitrarily. The explicit constants are provided in the Appendix.
\end{T1}

% summary (CTO, moderate deviation regime, replaced partial order)
Let us summarise the reasoning presented in this section. First, by replacing the partial order of second laws with a mathematically equivalent partial order characterised by $k$-copy transformations, one can formulate an alternative version of Theorem \ref{thm2}. This allows the phenomenon of catalytic universality to be demonstrated using multi-copy catalysts composed of a moderate number of particles. Although in that case the multi-copy catalyst will no longer be useful for all transformations (unless more copies are supplied), it will be useful for all of these transformations in which the participating states are $k$-copy convertible. 

%Finally, as a concrete example, let us consider the following: \plb{Example}.

% by estimating the second order corrections to the rate of transformation we can find minimal number of catalyst particles $n$ which allow for a catalytic transformation. Furthermore, using the fact that  
% might be an exclusively theoretical statement valid only in the limit of infinitely many particles comprising the catalyst.  
% describe the problem of control for Thermal Operations
% However, the above only partially solves the problem of requiring an unphysical amount of control. 

% describe the problem of control for 

% 

%We emphasize that we can also formulate our result without using the 

%The final problem might be in that the catalysts which we might need are often too big in order to be used in any real experiment.

\subsection{Extension to general majorization-based resource theories}
As a final note we emphasize that although we have focused our presentation on quantum thermodynamics, catalysis can be easily introduced into any QRT. Arguably the most well-studied class of resource theories are those which are governed by majorization. In such theories transitions between states are governed by pure majorization or its generalized version, $\bm{d}$-majorization \cite{Ruch1976}. In this work we focused mainly on quantum thermodynamics which can be thought of as a particular resource theory whose transitions are governed by a variant of majorization called thermo-majorization. In fact, each theory which can be described using majorization has its own collection of ``second laws'' which can be expressed using $\alpha$-Renyi entropies. Although the physical settings and interpretation of such laws differ among resource theories, the mathematical framework governing state transitions in such theories remains the same. Thanks to this wide applicability of majorization, the results presented in this work can be interpreted as well in the context of other majorization-based QRTs like the theory of pure-state entanglement or pure-state coherence. In this way the phenomenon of catalytic universality naturally trancends into other physical settings and can be viewed as a general feature of all majorization-based QRTs. 

\section{Generic catalysts}
\label{sec4}
We have shown that multi-copy states are universal catalysts. It is interesting to ask if, and to what extent, these results may be able to be generalized even further. As a first step, one might ask whether any high-dimension state is a catalyst, i.e. whether the crucial property is just high dimensionality, or whether the multi-copy form is essential, since there we have much stronger typicality properties emerging. This extension is relevant from an experimental point of view, where experimenters ability to manipulate the system are often limited by the number of degrees of freedom that can be effectively controlled. In this section we present numerical evidence which indicates that the catalytic universality phenomenon is more general and, in fact, concerns almost all large dimensional catalysts. We leave the proof of this conjecture for future research. The code used to run the numerics presented in this section was developed using MATLAB and is freely available at \cite{git_code}.

Throughout this section the Hamiltonians of the system $\rm S$ and the catalyst $\rm C$ are fully degenerate, meaning that $H_{\rm S} \propto H_{\rm C} \propto \mathbb{1}$. This will allow us to simplify both the presentation and numerical computation, since checking majorization is computationally easier than checking thermo-majorization. It should be noted that this does not reduce generality of our findings, as thermo-majorization criteria can always be expressed in terms of standard majorization using the embedding map \cite{Brand_o_2015}. Moreover, for the purpose of visualisation we focus here on the case when $d_{\rm S} = 3$. This will allow us to describe the numerical findings in a more natural and visually appealing way. In the Appendix we report further numerical evidence which indicate that these conclusions naturally extend to larger dimensional systems. 

\subsection{Fixed initial and final state}
Let us consider two states $\rho_{\rm S}$ and $\sigma_{\rm S}$ such that $\bm{p} = \text{diag}[\rho_{\rm S}]$ and $\bm{q} = \text{diag}[\sigma_{\rm S}]$ and chosen such that $(i)$ they satisfy the corresponding second laws (\ref{eq:sec_laws}), meaning that $H_{\alpha}(\bm{p}) \leq H_{\alpha}(\bm{q})$ for all $\alpha \geq 0$ and such that $(ii)$ the probability vector $\bm{p}$ does \emph{not} majorize $\bm{q}$ and vice versa. In this way we know that neither of $\rho_{\rm S}$ and $\sigma_{\rm S}$ can be transformed into each other, but there exists a catalyst $\omega_{\rm C}$ which can be used to facilitate the transformation from $\rho_{\rm S}$ to $\sigma_{\rm S}$. For illustrative purposes let us choose the following two representative states:
 \begin{align}
    \label{eq:special_states}
     \bm{p^{\star}} = (0.65,\, 0.2,\, 0.15), \qquad \bm{q}^{\star} = (0.5,\, 0.4,\, 0.1).
 \end{align}
 It is easy to check that $\bm{p}^{\star}$ and $\bm{q}^{\star}$ are incomparable using e.g. the concept of Lorenz curves \cite{Horodecki2013}. In this part of the section we  focus exclusively on these particular states.
 
 Suppose now that we choose a probability distribution $\mathcal{P}_{\text{dist}}$ and draw $d_{\rm C}$ positive numbers according to this distribution. We organize them in a vector and then normalize, obtaining a valid probability vector. In this way we have a simple method of drawing random catalysts, which for a large dimension $d_{\rm C}$ well approximates drawing from the probability simplex. Let us denote with $\bm{c}^{(i)} = (c_1^{(i)}, c_2^{(i)}, \ldots, c_{d_{\rm C}}^{(i)})$ an $i$-th probability vector obtained via this method. Each $\bm{c}^{(i)}$ will model a random catalyst drawn according to a respective probability distribution. In this way the elements of each such random catalyst are given by:
 \begin{align}
     c_{k}^{(i)} = \frac{X_{\text{dist}}^k}{\sum_{k = 1}^{d_{\rm C}} X_{\text{dist}}^k},
 \end{align}
  where $X_{\text{dist}}^k$ is a random variable drawn according to the probability distribution $\mathcal{P}_{\text{dist}}$. In what follows we will consider three different distributions:
  \begin{align}
      &\mathcal{P}_{\text{ray}} \rightarrow \mathsf{Prob}(X_{\text{ray}}^k = x) \sim x e^{-x^2/2}, \\
      & \mathcal{P}_{\text{uni}} \rightarrow \mathsf{Prob} (X_{\text{uni}}^k = x) \sim const, \\
      &\mathcal{P}_{\text{exp}} \rightarrow \mathsf{Prob}(X_{\text{exp}}^k = x) \sim e^{-x}.
      \label{eq:exp_dist}
  \end{align}
 Next we fix the error which we can tolerate on the catalyst $\eps{C}$ and repeat the process of sampling catalysts many times. Having done so we can now ask: how frequently does a randomly chosen catalyst can catalyze a given state transformation for a fixed error? Counting the frequency of cases in which the following transformation is possible:
 \begin{align}
     \bm{p}^{\star} \ot \bm{c}^{(i)} \rightarrow \bm{q}^{\star} \ot \widetilde{\bm{c}}^{\,(i)},
 \end{align}
 with $\widetilde{\bm{c}}^{\,(i)}$ chosen such that $\norm{\widetilde{\bm{c}}^{\,(i)} - \bm{c}^{(i)}}_1 \leq \eps{C}$, leads to the success probability $p_{\text{succ}}(\bm{p}^{\star}, \bm{q}^{\star}, \eps{C})$. This is an estimate of the probability that a randomly chosen catalyst can help in facilitating a given state transformation, with the disturbance on the catalyst at most $\eps{C}$. Moreover, using the results of Ref. \cite{Horodecki_2018} we can readily determine the final state of the catalyst $\widetilde{\bm{c}}^{\,(i)}$ to be the so-called $\epsilon$-flattest state (see Ref. \cite{Horodecki_2018} for the method of constructing these states). To summarise, the success probability $p_{\text{succ}}(\bm{p}^{\star}, \bm{q}^{\star}, \eps{C})$ is computed using the following set of steps:
 
 \vspace{10pt}
 \begin{algorithm}[H]
    \label{alg1}
    \DontPrintSemicolon
    \KwInput{$\bm{p}, \, \bm{q}, \, d_{\rm C},\, \eps{C}, \, \mathcal{P}_{\rm dist}$}
  \KwOutput{Estimate of $p_{\text{succ}}(\bm{p}, \bm{q}, \eps{C})$}
  \KwParams{$N_{\rm C}$ \tcp*{precision of estimation}} 
  $pos = 0$\\
  %\tcc{Now this is an if...else conditional loop}
  \ForEach{$i \in \{1, \ldots, N_{\rm C}\}$}
   {
   		$\bm{c}^{(i)}$ $\gets$ random catalyst sampled using $\mathcal{P}_{\rm dist}$ \; 
   		% \If{$\bm{p} \ot \bm{c}^{(i)} \rightarrow \bm{q} \ot \bm{\widetilde{c}}^{(i)}$ \emph{s.t.} $\norm{\bm{\widetilde{c}}^{(i)} - \bm{c}^{(i)}}_1 \leq\eps{C}$ }
   		\If{$\,\emph{there exists} \,\, \widetilde{\bm{c}}^{\,(i)} \,\,\emph{s.t.}\,\, \bm{p} \ot \bm{c}^{(i)} \rightarrow \bm{q} \ot \bm{\widetilde{c}}^{\,(i)}$ \emph{and} $\norm{\bm{\widetilde{c}}^{(i)} - \bm{c}^{(i)}}_1 \leq\eps{C}$ }
   		{
   		$pos \gets pos + 1$\;
   		}
   }
   $p_{\text{succ}}(\bm{p}, \bm{q}, \eps{C}) \gets pos / N_{\rm C}$
\caption{{Estimating $p_{\text{succ}}(\bm{p}, \bm{q}, \eps{C})$ by sampling}}
\end{algorithm}
\vspace{10pt}

The results of this numerical experiment are summarized in Fig. \ref{fig:2}. As we can see, when we increase the dimension $d_{\rm C}$, the probability that a randomly chosen catalyst can catalyze a given transformation increases and very rapidly approaches a fixed value. This value, as well as the rate at which it is approached, depends on the specific distribution $\mathcal{P}_{\text{dist}}$ we choose. This indicates that the success probability $p_{\text{succ}}(\mathcal{P}_{\text{dist}}, d_{\rm C})$ depends both on the dimension of the catalyst and the distribution of its eigenvalues. This numerical experiment allows us to conclude that there are state transformations for which random states act as catalysts with high probability.

\begin{figure}
    \centering
    \includegraphics[width=0.9\linewidth]{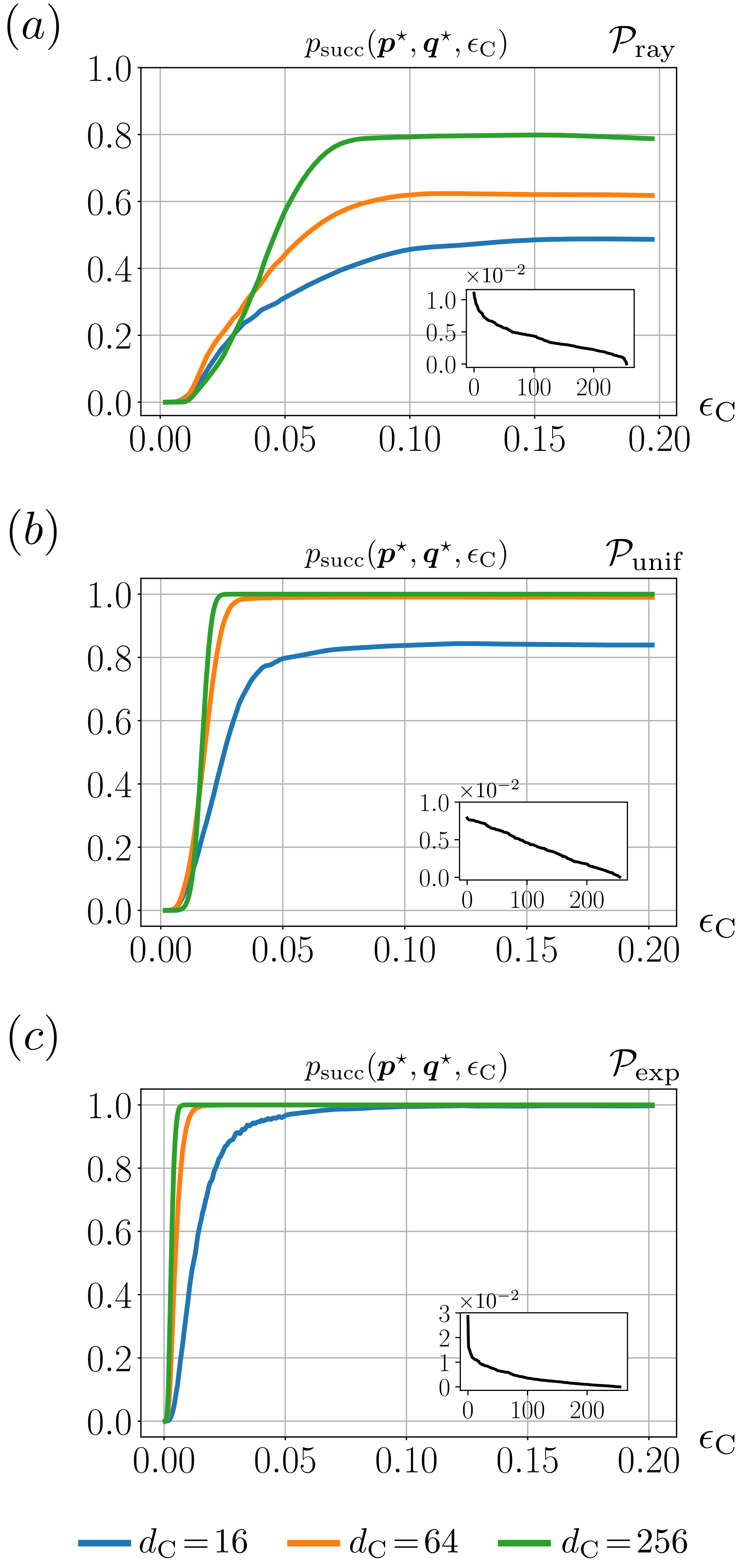}
    \caption{The probability $p_{\text{succ}}(\bm{p}^{\star}, \bm{q}^{\star}, \eps{C})$ that a random state of dimension $d_{\rm C}$ drawn from probability distribution $\mathcal{P}_{\text{dist}}$ can catalyze a fixed state transformation $\bm{p}^{\star} \rightarrow \bm{q}^{\star}$. Each panel corresponds to a different distribution: ($a$) Rayleigh $\mathcal{P}_{\text{ray}}$, ($b$) uniform $\mathcal{P}_{\text{unif}}$ and ($c$) exponential $\mathcal{P}_{\text{exp}}$. Inset plots illustrate an exemplary distribution of eigenvalues $c_{k}^{(i)}$ of a random catalyst $\bm{c}^{(i)}$ drawn according to a respective distribution and then normalized.}
    \label{fig:2}
\end{figure}

\subsection{{Fixed initial state and arbitrary final state}}
In the previous section we studied how useful random catalysts are for a fixed state transformation. We now go one step further and generalise this investigation to arbitrary final states, while still keeping the initial state fixed. Let us then consider again the state $\bm{p}^{\star}$ given by Eq. (\ref{eq:special_states}) as the input state and let $\bm{q}$ be an arbitrary state. Since $\bm{p}^{\star}$ and $\bm{q}$ are $d_{\rm S}$-dimensional probability vectors, it is useful to think of them as points in the space of all $d_{\rm S}$-dimensional probability vectors, the so-called probability simplex $\Delta_{d_{\rm S}}$ defined as:
\begin{align}
    \Delta_{N} :=\! \left\{ \bm{x} = (x_1, \ldots, x_{N})\,|\, x_i \geq 0 \,\,\text{and}\,\, \sum_{i=1}^{N} x_i = 1\!\right\}
\end{align} 
Let us now define two sets of states inside $\Delta_{d_{\rm S}}$:
\begin{align}
    \label{set_S}
    \mathsf{S}(\bm{p}) &= \{\bm{q}\, |\, \bm{p} \rightarrow \bm{q} \,\,\text{and}\,\, \bm{q} \in \Delta_{d_{\rm S}}\}, \\
    \label{set_T}
    \mathsf{T}(\bm{p}) &= \{\bm{q}\, |\, \bm{p} \ot \bm{c} \rightarrow \bm{q} \ot \bm{c}, \,\, \bm{q} \in \Delta_{d_{\rm S}} \,\,\text{and}\,\, \bm{c} \in \Delta\},
\end{align}
where $\Delta$ is the set of all $N$-dimensional probability simplices $\Delta_{N}$ for all natural $N$. The set $\mathsf{S}(\bm{p})$ contains all states $\bm{q}$ which are (thermo-) majorised by $\bm{p}$, that is all states which can be reached via thermal operations when starting from a state described by $\bm{p}$. The set $\mathsf{T}(\bm{p})$ contains all states $\bm{q}$ which can be reached via thermal operations with the help of some (unspecified) catalyst. We will refer to these sets as the thermal and the catalytic-thermal set, respectively. It can be readily verified that $\mathsf{S}(\bm{p}) \subseteq \mathsf{T}(\bm{p})$ for all $\bm{p}$. In this language, the main result of Ref. \cite{Jonathan_1999} shows that $\mathsf{S}(\bm{p}) \subset \mathsf{T}(\bm{p})$ for some $\bm{p}$. Moreover, due to the results of Refs. \cite{PhysRevA.64.42314,Turgut_2007,Brand_o_2015} a complete characterisation of the set $\mathsf{T}(\bm{p})$ is known and whether $\bm{q} \in \mathsf{T}(\bm{q})$ is determined by the second laws (\ref{eq:sec_laws}). 

% one (very) big plot and two smaller ones
% reduce size of the set lines
\begin{figure*}
    \centering
    \includegraphics[width=\linewidth]{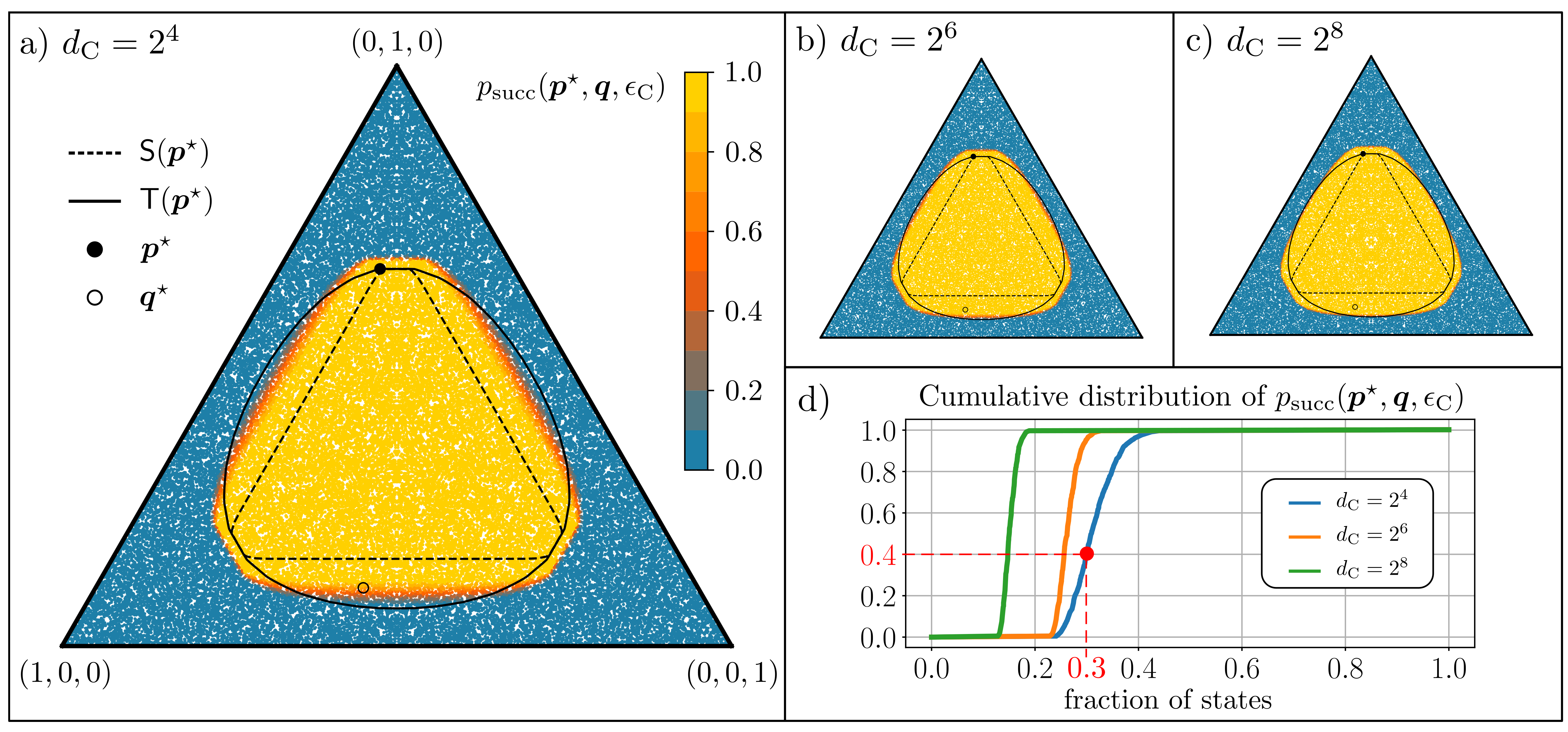}
    \caption{The probability $p_{\rm succ}(\bm{p}^{\star}, \bm{q}, \eps{C})$ that a random state can be used to enable the transformation from $\bm{p}^{\star}$ to $\bm{q}$ approximately catalytically, i.e. with an error $\eps{C} = \mu \eps{C}^{\rm bnd}$ and $\mu = 0.1$. Plots $(a)-(c)$ correspond to different dimension of the random catalyst ($d_{\rm C} = 2^4$, $2^6$ and $2^8$ respectively). States inside the region bounded by dashed lines define the thermal set $\mathsf{S}(\bm{p}^{\star})$, consisting of all states that can be reached from $\bm{p}^{\star}$ using thermal operations. The solid line corresponds to the catalytic-thermal set $\mathsf{T}(\bm{p}^{\star})$, consisting of all states which can be reached from $\bm{p}^{\star}$ with the help of some (potentially finely-tuned) catalyst. Note that the probability of success $p_{\rm succ}(\bm{p}^{\star}, \bm{q}, \eps{C})$ is generally very close to $1$ for most final states $\bm{q}$ inside $\mathsf{D}(\bm{p}^{\star}) := \mathsf{T}(\bm{p}^{\star}) \setminus \mathsf{S}(\bm{p}^{\star})$, even when the dimension of the catalyst is relatively small. Plot $(d)$ illustrates the cumulative distribution of $p_{\rm succ}(\bm{p}^{\star}, \bm{q}, \eps{C})$. To simplify interpretation, an exemplary point is drawn in red. It corresponds to the case when $30\%$ of all states $\bm{q} \in \mathsf{D}(\bm{p}^{\star})$ can be reached using random catalysts of dimension $d_{\rm C} = 2^4$, with probability less than $0.4$. In other words, $70\%$ of all possible catalytic transformations can be realised using random catalysts with probability $0.4$ or higher. }
    \label{fig:4}
\end{figure*}

Let us now perform our second numerical experiment. In the previous section we saw that for a fixed transformation, catalysts sampled from the exponential distribution (\ref{eq:exp_dist}) achieve a high probability of success $p_{\rm succ}(\bm{p}^{\star}, \bm{q}^{\star}, \eps{C})$, even for moderate dimensions of catalysts. Let us now use the exponential distribution to sample random catalysts and compute the associated success probability. Furthermore, to assure that we do not work in the embezzlement regime, we also fix the allowable error on the catalyst to be $\eps{C} = \mu \eps{C}^{\rm bnd}$, where $\eps{C}^{\rm bnd}$ is the embezzlement bound from Eq. (\ref{def:emb_bnd}) and $0 \leq \mu < 1$. For arbitrary points $\bm{q} \in \Delta_{d_{\rm S}}$ we then estimate $p_{\rm succ}(\bm{p}^{\star}, \bm{q}, \eps{C})$, the probability that a random catalyst can be used to transform $\bm{p}^{\star}$ into $\bm{q}$ using the method described in Alg. \ref{alg1}. 

The results of this numerical experiment are summarized in Fig. \ref{fig:4}. For the purpose of illustration we also draw the sets $\mathsf{S}(\bm{p}^{\star})$ and $\mathsf{T}(\bm{p}^{\star})$. The numerics demonstrate that the probability $p_{\rm succ}(\bm{p}^{\star}, \bm{q}, \eps{C})$ is large for most $\bm{q}$ in the region $\mathsf{D}(\bm{p}^{\star}) := \mathsf{T}(\bm{p}^{\star}) \setminus \mathsf{S}(\bm{p}^{\star})$ even for small dimension of the catalyst (e.g. when $d_{\rm C} = 2^4$). Interestingly, the success probability increases significantly with the dimension of the catalyst, so that for $d_{\rm C} = 2^8$ the value of $p_{\rm succ}(\bm{p}^{\star}, \bm{q}, \eps{\rm C}) \approx 1$ for almost all points $\bm{q}$ inside $\mathsf{D}(\bm{p}^{\star})$. As a consequence, we can infer that there are (input) states for which random states act as catalysts with high probability \emph{for all} possible output states. 

\subsection{{Arbitrary initial and final states}}
In the previous section we saw that a random catalyst can catalyze most of the possible state transformations for a fixed state $\bm{p}$, even for catalysts with a moderate dimension. In this section we extend our analysis and show that this behavior is a generic feature valid for arbitrary initial states. 

Before going into the details, let us emphasize that not all initial states $\bm{p}$ lead to an interesting catalytic advantage. For example, when the system starts in a thermal state, $\bm{p} = \bm{g}$, there is no catalyst which can enhance the system's transformation potential. In that case the thermal and catalytic-thermal sets coincide, i.e. $\mathsf{S}(\bm{g}) = \mathsf{T}(\bm{g})$. Similar situation happens when the initial state of the system is a pure state. One natural way to quantify the potential for a catalytic improvement is to estimate the volume of set defined as the difference between the thermal and the catalytic-thermal set. In what follows we will refer to such a set of states $\mathsf{D}(\bm{p}) := \mathsf{T}(\bm{p}) \setminus \mathsf{S}(\bm{p})$ as the \emph{catalytic activation set} (CAS). Naturally, the volume of this region in the space of distributions largely varies between different initial states $\bm{p}$. 

The aim of the next numerical experiment is to extend the results from the previous section to arbitrary initial states. We again fix a small error on the catalyst $\eps{C} = \mu \eps{C}^{\rm bnd}$ to assure that we do not work in the embezzlement regime. We then sample uniformly the initial state $\bm{p}$ and compute the associated CAS, denoted $\mathsf{D}(\bm{p})$. For each point $\bm{q} \in \mathsf{D}(\bm{p})$ we estimate the probability that a random catalyst can be used to transform $\bm{p}$ into $\bm{q}$ using the methods described in Alg. \ref{alg1}. Finally, we calculate the number of states inside $\mathsf{D}(\bm{p})$ for which the probability of catalysing using random catalysts, $p_{\rm succ}(\eps{C})$, is larger than a fixed threshold value $\gamma_{\rm thd}$ \footnote{It's worth to remind that the transformations which we consider here are always deterministic and the probability $p_{\rm succ}(\eps{C})$ refers to the sampled catalysts, i.e. how probable it is to \emph{deterministically} catalyze a given transformation using a catalyst drawn at random.}. This allows us to estimate, for each $\bm{p}$, the fraction $f(\bm{p})$ of all possible transformations which can be catalyzed using random catalysts with probability at least $\gamma_{\rm thd}$, i.e.:
\begin{align}
    \label{f_quant}
    f(\bm{p}) := \frac{| \widetilde{\mathsf{D}}(\bm{p})|}{|\mathsf{D}(\bm{p})|},
\end{align}
where $\widetilde{\mathsf{D}}(\bm{p}) \subseteq \mathsf{D}(\bm{p})$ and $\bm{q} \in \widetilde{\mathsf{D}}(\bm{p})$ if and only if $p_{\rm succ}(\bm{p}, \bm{q}, \eps{C}) \geq \gamma_{\rm thd}$. The quantity $f(\bm{p})$ can be estimated using the following simple algorithm:

\vspace{10pt}
 \begin{algorithm}[H]
    \label{alg2}
    \DontPrintSemicolon
    \KwInput{$\bm{p}, d_{\rm C},\, \eps{C}, \gamma_{\rm{thd}},\, \mathcal{P}_{\rm dist}$}
  \KwOutput{Estimate of $f(\bm{p})$}
  \KwParams{$N_{\rm S}$} %\tcp*{precision of estimation}} 
  \vspace{5pt}
  $\mathsf{A}(\bm{p})$ $\gets$ initialise uniformly $N_{\rm S}$ states of dimension $ d_{\rm \rm S}$ \\
  $\mathsf{S}(\bm{p})$ $\gets$ all states in $\mathsf{A}(\bm{p})$ satisfying (\ref{set_S}) \\ 
  $\mathsf{T}(\bm{p})$ $\gets$ all states in $ \mathsf{A}(\bm{p})$ satisfying (\ref{set_T})\\
  $\mathsf{D}(\bm{p})$ $\gets$ $\mathsf{T}(\bm{p})\setminus\mathsf{S}(\bm{p})$ \\
  
  \vspace{5pt}
  %$\text{vol}\, \mathsf{S}(\bm{p}) \gets \norm{\mathsf{S}(\bm{p})} / N_{\rm S}$ \\
  %$\text{vol}\, \mathsf{T}(\bm{p}) \gets \norm{\mathsf{T}(\bm{p})} / N_{\rm S}$ \\
  %\vspace{5pt}
  $pos = 0$ \\
  \ForEach {$\bm{q} \in \mathsf{D}(\bm{p})$}{
   $p_{\rm succ}(\bm{p}, \bm{q}, \eps{C})$ $\gets$ compute using Alg. \ref{alg1} \\
   \If{$p_{\rm succ}(\bm{p}, \bm{q}, \eps{C}) \geq \gamma_{\rm thd}$}{$pos \gets pos + 1$}
   }
   $f(\bm{p}) \gets pos /|D(\bm{p})|$
\caption{Estimating $f(\bm{p})$ by sampling}
\end{algorithm}
\vspace{10pt}

The results of this numerical experiment are summarized in Fig. \ref{fig:5}. Interestingly, even for relatively small catalyst dimensions (e.g. $d_{\rm C} = 16$) there is a modest fraction of possible transformations that with high probability can be catalysed with a random catalyst. Furthermore, by increasing dimension of the catalyst this fraction improves significantly, so that already for moderate sized catalysts ($d_{\rm C} = 256$) most of the possible transformations can be catalysed using random catalysts with a large probability. As a consequence, this numerical investigation allows us to infer that most of all possible transformations can be catalytically activated with high probability using random states as catalysts. In the Appendix we give analogous plots for a different choice of the relative error $\mu$, the threshold value $\gamma_{\rm thld}$ and system dimension $d_{\rm S}$, to demonstrate that this behavior is generic, i.e. it does not depend on our particular choice of parameters. 

\begin{figure*}
    \centering
    \includegraphics[width=\linewidth]{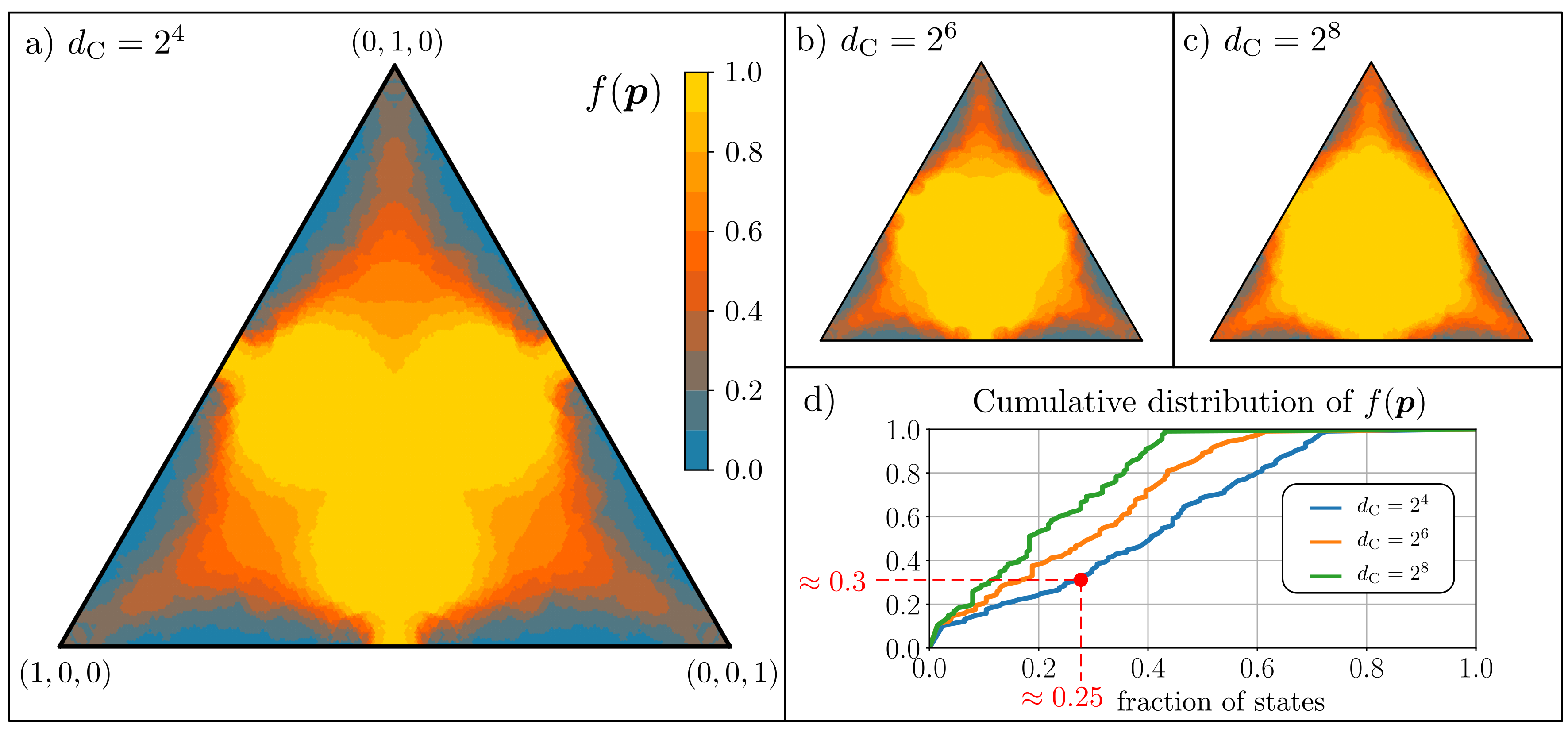}
    \caption{The fraction $f(\bm{p})$ of all states inside the catalytic activation set $\mathsf{D}(\bm{p})$ (CAS) for which the probability $p_{\rm succ}(\bm{p}, \bm{q}, \eps{C})$ that a random state can be used as a catalyst, is larger than the threshold value $\gamma_{\rm thd} = 0.9$. Plots $(a-c)$ correspond to random catalysts of dimensions $2^4$, $2^6$ and $2^8$ respectively. Plot $(d)$ illustrates the cumulative distribution of $f(\bm{p})$. As an example, the red point corresponds to the situation where random states of dimension $d_{\rm C} = 2^4$  are used to catalyze possible transformations. These random catalysts are not useful for roughly $25\%$ of all possible initial states $\bm{p}$, i.e. for each such state they allow at most $\approx 30 \%$ of all output states $\bm{q}$ in CAS to be reached with probability equal to or greater than $\gamma_{\rm thld}$. In other words, such random catalysts can be used to reach more than $30\%$ of all possible output states $\bm{q} \in \mathsf{D}(\bm{p})$, with probability at least $\gamma_{\rm thld}$, for at least $75\%$ of all initial states $\bm{p}$. }
    \label{fig:5}
\end{figure*}

\subsection{{Comparison with multi-copy states}}
In the previous sections we studied how useful random states are in catalysing thermodynamic transformations. In the final numerical experiment we compare these insights with the analytical results presented in this paper. In particular, we compute the quantity $f(\bm{p})$ defined in Eq. (\ref{f_quant}) using as a catalyst multiple copies of a fixed state. This will allow us to examine some of the practical aspects of our construction, i.e. when the catalysts consists of a moderate number of particles. 

We again fix the error on the catalyst to be $\eps{C} = \mu \eps{C}^{\rm bnd}$ and compute the quantity $f(\bm{p})$ using Alg. \ref{alg2}, with the only difference that now $p_{\rm succ}(\bm{p}, \bm{q}, \eps{C})$ is computed using a fixed multi-copy catalyst. Hence, it can be either $0$ (the multi-copy catalyst does not allow to transform $\bm{p}$ into $\bm{q}$ within the allowed error on the catalyst) or $1$ (when $\bm{p}$ can be transformed into $\bm{q}$ within the allowed error). The single-copy catalyst is chosen to be a qubit in a state $\omega_{\rm C} = \diag(1-r, r)$ with $0 \leq r < 1/2$. The results of this numerical experiment are summarised in Fig. \ref{fig:6}. Comparing these with Fig. \ref{fig:5} illustrates that multi-copy catalysts generally achieve a larger fraction $f(\bm{p})$ for the same dimension of catalyst. However, in the multi-copy case the improvement in $f(\bm{p})$ obtained by increasing the catalyst's dimension is generally smaller than in the case of random catalysts. This indicates that catalytic universality using only few copies might be generally possible when the catalyst is sufficiently mixed. Finally, Fig. \ref{fig:6} demonstrates that some marks of catalytic universality can be observed for a relatively modest number of particles comprising the catalyst.

\begin{figure*}
    \centering
    \includegraphics[width=\linewidth]{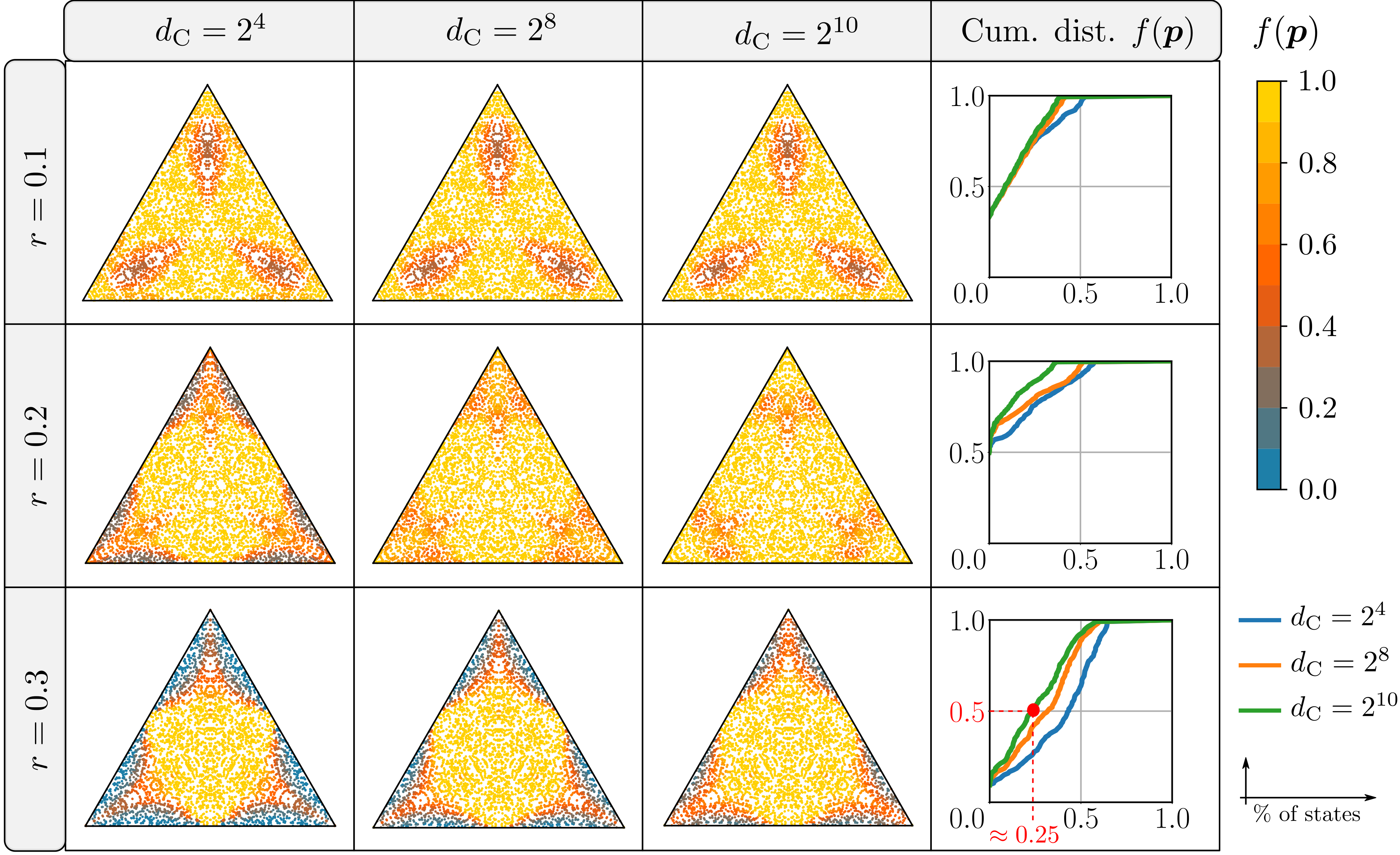}
    \caption{The fraction $f(\bm{p})$ of all states inside the catalytic activation set of $\bm{p}$ which can be catalysed using a multi-copy catalyst composed of $n \in \{4, 8, 10\}$ qubits $\omega_{\rm C} = \diag(1-r, r)$, for different values of parameter $r$. Column $4$ shows a cumulative distribution of $f(\bm{p})$, i.e. it illustrates the fraction of catalytic transformations that can be activated using a multi-copy catalyst. For example, the red exemplary point indicates that for around $0.25$ of all possible input states $\bm{p}$, a quantum system composed of $n = 10$ qubits in a state with $r = 0.3$ acts as a catalyst for less than $0.5$ of all possible output states $\bm{q}$ in CAS. Equivalently, for this choice of parameters, the multi-copy catalyst can activate more than $0.5$ of all possible transformations for approximately $0.75$ of all possible input states.}
    \label{fig:6}
\end{figure*}

\subsection{{Further directions}}
Although still somewhat preliminary in nature, these numerical findings strongly suggest that high-dimensional states, with high probability, will act as catalysts. This indicates that the universality phenomenon we uncovered here may be even more general than we can yet prove analytically. It seems reasonable to expect that a potential route to further analytic results will be to look for statements which hold with high probability. We leave this tantalising extension of our results for future work.  

Finally, in the Appendix we provide further numerical evidence that the catalytic universality phenomenon is a generic feature of sufficiently high dimensional catalysts. In particular, we $(i)$ perform numerical calculations for larger dimensional systems $\rm S$, $(ii)$ change the sampling of catalysts to other distributions and $(iii)$ choose different thresholds for the catalyst error $\eps{C}$ and $\delta$.

\section{Summary}
\label{sec5}
In this work we have studied the problem of catalysis in quantum thermodynamics. We have shown that any state can act as a (universal) catalyst for all transformations allowed by the laws of thermodynamics, provided that enough copies of the catalyst are available. In particular, in the case of genuine catalysis we have shown that all states can be used as catalysts in transforming $\rho$ into $\sigma$, as long as the two states satisfy the second laws of thermodynamics. In this case the error on the catalyst decreases sub-exponentially with the number of particles. Furthermore, in the embezzlement regime, which appears when the second laws completely vanish, we found an analogous behavior. In that case the error on the catalyst scales with the number of particles $n$ much worse than in the regime of genuine catalysis, however still approaches zero as $n \rightarrow \infty$.   

We have also emphasized that this surprising property of catalysis is a genuine feature arising in any majorization-based resource theory. In this way the results can be applied in a wide range of other contexts, ranging from the theory of pure entanglement to the theory of purity or coherence.

Finally, we conjectured that this phenomena is a genuine feature of large-dimensional catalysts and provided simple numerical evidence to support this conjecture.

\section{Discussion and open problems}
\label{sec6}
In this work we have presented and proved a surprising property of multi-copy catalyst states: that every state given sufficiently many copies can act as a universal catalyst. We believe that this new realisation is a substantial step forward in our understanding of catalysis and provides new insights both in the field of quantum thermodynamics and resource theories. % Finding the right catalyst is one of the main difficulties when employing catalysis in any resource theory.
What is more, our work opens the door for new avenues of exploration which will be of independent interest. In the following subsections we briefly sketch the most promising, in our opinion, directions of extending the results presented in this work.

\subsection{The mechanism of catalysis}
Since the seminal paper of Jonathan and Plenio \cite{Jonathan_1999} our understanding of catalysis has grown significantly. However, we still do not fully understand the real mechanism behind catalysis and how it allows for lifting some of the restrictions imposed by allowable operations.

Here we made a step forward in explaining this mechanism. However, before a satisfactory understanding can be reached, several important challenges still need to be tackled.
In particular, a long-standing open problem is determining which physical properties of states are important for catalysis? Moreover, we do not know how is the set of states reachable via catalytic transformations modified when additional constraints on the catalyst are made - e.g. in terms of energy, entropy or the distribution of its eigenvalues. What is the main property or ''resource'' relevant for catalysis? Our analytic results and preliminary numerics strongly suggest that dimension of the catalyst and distribution of its eigenvalues are both important properties which make a good catalyst. This also indicates a trade-off relation between catalyst dimension and its ability to catalyze transformations. Quantifying and understanding this potential trade-off will significantly advance our understanding of catalysis.

\subsection{Catalytic universality for generic states}
In Sec. \ref{sec4} we presented a simple numerical evidence which indicates that the catalytic universality might appear for arbitrary large-dimensional catalysts. We believe that solving this problem will shed more light on the fundamental problem of what does the catalyst really do to facilitate the transformation. In particular, should we expect to find only specific catalysts if we modify some of our initial assumptions? 

Another interesting way to proceed would be to study how important for the catalytic universality are correlations between the subsystems which form the catalyst. Looking more closely at the proofs presented here we can see that both in the embezzlement and genuine catalysis regime  the main catalytic transformation $\mathcal{E}_{\rm SC}$ does not build such correlations. In this respect, the only time where correlations can increase is during the pre and post-processing steps. However, since this potential increase in correlations is only due to the transformation error, we conjecture that it is not a necessary requirement for our results to hold.

In this respect, it would be also interesting to revisit the results from Ref. \cite{Mueller2018} and check whether in the regime where only correlations are allowed to build up (that is when the reduced state of the catalyst subsystems remain undisturbed), multi-copy catalysts can be still viewed as universal catalysts. 

\subsection{Improving the error scaling}
Another interesting direction of extending the results presented in this paper would be to reduce the disturbance induced on the catalyst during the main protocol.

In particular, in the embezzlement regime this would involve choosing a different intermediate state $\eta_{\rm C}$ and choosing a more specialized  transformation $\mathcal{E}_{\rm SC}$. It seems plausible that using a more elaborate mixing transformation would allow one to further reduce the error scaling and, potentially, approach the threshold specified by (\ref{def:emb_bnd}).

On the other hand, in the regime of genuine catalysis we expect that the scaling of the catalyst disturbance can be further improved using the tools of large deviation theory. The primary source of error in this regime are the pre- and postprocessing steps which both incur a subexponential disturbance on the catalyst. Indeed, the results of \cite{Chubb_2019} were obtained in the regime when 
the error term $\delta(n)$ scales no better than $\mathcal{O}(e^{-n^{\kappa}})$, where $\kappa \in (0, 1)$. In the large deviation regime the error is exponentially vanishing at the cost of a constant gap between the effective ($r_n$) and asymptotic ($r_{\infty}$) conversion rates. Since the presence of such gap would not alter our main protocol, we believe that working in the large deviation regime would effectively allow the error scaling to be improved to an exponential one. 

Finally, in our work we assumed that the catalyst system has a fully degenerate energetic spectrum. Due to the result from \cite{Brand_o_2015} we know that it is sufficient to consider only such catalysts. However, what we do not know is whether these catalysts are optimal with respect to either the dimension or entropy. It would be interesting to study in what way considering a nontrivial Hamiltonian of the catalyst allows the error scaling or the minimal number of necessary copies to be improved. In particular, one can ask if it is possible to achieve a faster error scaling or equivalently, require smaller catalysts for the same error, by using more energetic catalysts.

 \subsection{Quantification of catalysis regimes}
 In this work we divided catalysis in two different regimes. We explored the embezzlement regime in which the partial order between states completely vanishes and the genuine catalysis regime in which we are guaranteed that at least some of the second laws remain.
 
 It would be interesting to pursue this idea more carefully and examine which of the second laws remain as valid monotones when allowing for a certain type of error scaling on the catalyst. We know from \cite{Brand_o_2015} that when one allows for an error scaling which is linear in the number of particles, that is $\eps{C} \sim \mathcal{O}(n^{-1})$, then the non-equilibrium Helmholtz free energy remains as the only necessary and sufficient second law. What we do not know is how many and which laws remain when we allow for other type of errors. Solving this problem would certainly increase our understanding of catalysis and its physical importance for quantum thermodynamics and other resource theories. 
  
  \subsection{Catalytic universality and second laws for coherence}  
  In our work we have not explored catalysis in the regime where states $\rho_{\rm S}$ and $\sigma_{\rm S}$ contain coherences between energy levels. It is known that in this case the second laws (\ref{eq:sec_laws}) provide only necessary but not sufficient conditions for state transformations. When considering fully general states with coherences, one has to additionally satisfy a completely new set of conditions resulting from the time symmetry constraints \cite{lostaglio2015description}. Loosely speaking, these new laws tell us that coherences between energy levels must decrease during thermodynamic transformation.  In that case it would be interesting to see if the catalytic universality phenomenon can appear also for fully general coherent states.

 \subsection{Other potential directions}
 
 \textbf{Consequences of the resonance phenomenon.}  One of the main tools which we used in our protocol was the multi-copy state conversion in the moderate deviation regime \cite{Chubb_2019}. In a recent article \cite{Korzekwa_2019} it was shown that moderate deviation analysis exhibits an interesting phenomenon of resource resonance which arises during non-asymptotic state conversions in the resource theories of entanglement, coherence and thermodynamics. This resource resonance implies that certain pairs of resource states can be interconverted at the asymptotically optimal rate with negligible error, even in the regime of finite $n$. In the context of our results this means that for certain states catalytic transformations can be achieved by using many fewer copies of the catalyst state. We believe that understanding the role of the resonance behavior in catalysis can lead to novel insights not only for quantum thermodynamics, but also for the resource theories of entanglement, purity and coherence.
 
 \textbf{Extending catalytic universality to arbitrary QRTs.} A natural question which arises when studying catalysis in the context of majorization-based QRT's is whether the catalysis phenomenon can be properly defined and studied for general resource theories as well. Interestingly, there are examples of QRTs for which catalysis does not enlarge the set of states which can be reached using free operations \cite{schmid2020standard}. This leads to an interesting question: what are the necessary properties of a general QRT which allow it to have a nontrivial catalysis?  Consequently, one can further ask if the catalytic universality phenomenon can also emerge for such theories. If not, then it would mean that catalytic universality is a unique feature of majorization-based QRTs and it would be interesting to see which special aspects of such theories allow for the catalytic universality? 
   
\section{acknowledgements}
We would like to thank Michał Horodecki for helpful discussions, especially for suggesting using the ``flattest state'' in the numerical part of this work. PLB acknowledges support from the UK EPSRC (grant no. EP/R00644X/1). PS acknowledges support from a Royal Society URF (UHQT). This study did not involve any underlying data.

\bibliography{citations}

\begin{thebibliography}{86}%
\makeatletter
\providecommand \@ifxundefined [1]{%
 \@ifx{#1\undefined}
}%
\providecommand \@ifnum [1]{%
 \ifnum #1\expandafter \@firstoftwo
 \else \expandafter \@secondoftwo
 \fi
}%
\providecommand \@ifx [1]{%
 \ifx #1\expandafter \@firstoftwo
 \else \expandafter \@secondoftwo
 \fi
}%
\providecommand \natexlab [1]{#1}%
\providecommand \enquote  [1]{``#1''}%
\providecommand \bibnamefont  [1]{#1}%
\providecommand \bibfnamefont [1]{#1}%
\providecommand \citenamefont [1]{#1}%
\providecommand \href@noop [0]{\@secondoftwo}%
\providecommand \href [0]{\begingroup \@sanitize@url \@href}%
\providecommand \@href[1]{\@@startlink{#1}\@@href}%
\providecommand \@@href[1]{\endgroup#1\@@endlink}%
\providecommand \@sanitize@url [0]{\catcode `\\12\catcode `\$12\catcode
  `\&12\catcode `\#12\catcode `\^12\catcode `\_12\catcode `\%12\relax}%
\providecommand \@@startlink[1]{}%
\providecommand \@@endlink[0]{}%
\providecommand \url  [0]{\begingroup\@sanitize@url \@url }%
\providecommand \@url [1]{\endgroup\@href {#1}{\urlprefix }}%
\providecommand \urlprefix  [0]{URL }%
\providecommand \Eprint [0]{\href }%
\providecommand \doibase [0]{https://doi.org/}%
\providecommand \selectlanguage [0]{\@gobble}%
\providecommand \bibinfo  [0]{\@secondoftwo}%
\providecommand \bibfield  [0]{\@secondoftwo}%
\providecommand \translation [1]{[#1]}%
\providecommand \BibitemOpen [0]{}%
\providecommand \bibitemStop [0]{}%
\providecommand \bibitemNoStop [0]{.\EOS\space}%
\providecommand \EOS [0]{\spacefactor3000\relax}%
\providecommand \BibitemShut  [1]{\csname bibitem#1\endcsname}%
\let\auto@bib@innerbib\@empty
%</preamble>
\bibitem [{\citenamefont {Koski}\ \emph {et~al.}(2014)\citenamefont {Koski},
  \citenamefont {Maisi}, \citenamefont {Sagawa},\ and\ \citenamefont
  {Pekola}}]{Koski2014}%
  \BibitemOpen
  \bibfield  {author} {\bibinfo {author} {\bibfnamefont {J.~V.}\ \bibnamefont
  {Koski}}, \bibinfo {author} {\bibfnamefont {V.~F.}\ \bibnamefont {Maisi}},
  \bibinfo {author} {\bibfnamefont {T.}~\bibnamefont {Sagawa}},\ and\ \bibinfo
  {author} {\bibfnamefont {J.~P.}\ \bibnamefont {Pekola}},\ }\bibfield  {title}
  {\bibinfo {title} {Experimental observation of the role of mutual information
  in the nonequilibrium dynamics of a maxwell demon},\ }\href
  {https://doi.org/10.1103/PhysRevLett.113.030601} {\bibfield  {journal}
  {\bibinfo  {journal} {Phys. Rev. Lett.}\ }\textbf {\bibinfo {volume} {113}},\
  \bibinfo {pages} {030601} (\bibinfo {year} {2014})}\BibitemShut {NoStop}%
\bibitem [{\citenamefont {Koski}\ \emph {et~al.}(2015)\citenamefont {Koski},
  \citenamefont {Kutvonen}, \citenamefont {Khaymovich}, \citenamefont
  {Ala-Nissila},\ and\ \citenamefont {Pekola}}]{Koski2015}%
  \BibitemOpen
  \bibfield  {author} {\bibinfo {author} {\bibfnamefont {J.~V.}\ \bibnamefont
  {Koski}}, \bibinfo {author} {\bibfnamefont {A.}~\bibnamefont {Kutvonen}},
  \bibinfo {author} {\bibfnamefont {I.~M.}\ \bibnamefont {Khaymovich}},
  \bibinfo {author} {\bibfnamefont {T.}~\bibnamefont {Ala-Nissila}},\ and\
  \bibinfo {author} {\bibfnamefont {J.~P.}\ \bibnamefont {Pekola}},\ }\bibfield
   {title} {\bibinfo {title} {On-chip maxwell's demon as an information-powered
  refrigerator},\ }\href {https://doi.org/10.1103/PhysRevLett.115.260602}
  {\bibfield  {journal} {\bibinfo  {journal} {Phys. Rev. Lett.}\ }\textbf
  {\bibinfo {volume} {115}},\ \bibinfo {pages} {260602} (\bibinfo {year}
  {2015})}\BibitemShut {NoStop}%
\bibitem [{\citenamefont {Chida}\ \emph {et~al.}(2017)\citenamefont {Chida},
  \citenamefont {Desai}, \citenamefont {Nishiguchi},\ and\ \citenamefont
  {Fujiwara}}]{Chida2017}%
  \BibitemOpen
  \bibfield  {author} {\bibinfo {author} {\bibfnamefont {K.}~\bibnamefont
  {Chida}}, \bibinfo {author} {\bibfnamefont {S.}~\bibnamefont {Desai}},
  \bibinfo {author} {\bibfnamefont {K.}~\bibnamefont {Nishiguchi}},\ and\
  \bibinfo {author} {\bibfnamefont {A.}~\bibnamefont {Fujiwara}},\ }\bibfield
  {title} {\bibinfo {title} {Power generator driven by maxwell's demon},\
  }\href {https://doi.org/10.1038/ncomms15301} {\bibfield  {journal} {\bibinfo
  {journal} {Nat. Commun.}\ }\textbf {\bibinfo {volume} {8}},\ \bibinfo {pages}
  {15310} (\bibinfo {year} {2017})}\BibitemShut {NoStop}%
\bibitem [{\citenamefont {Camati}\ \emph {et~al.}(2016)\citenamefont {Camati},
  \citenamefont {Peterson}, \citenamefont {Batalh\~ao}, \citenamefont
  {Micadei}, \citenamefont {Souza}, \citenamefont {Sarthour}, \citenamefont
  {Oliveira},\ and\ \citenamefont {Serra}}]{Camati2016}%
  \BibitemOpen
  \bibfield  {author} {\bibinfo {author} {\bibfnamefont {P.~A.}\ \bibnamefont
  {Camati}}, \bibinfo {author} {\bibfnamefont {J.~P.~S.}\ \bibnamefont
  {Peterson}}, \bibinfo {author} {\bibfnamefont {T.~B.}\ \bibnamefont
  {Batalh\~ao}}, \bibinfo {author} {\bibfnamefont {K.}~\bibnamefont {Micadei}},
  \bibinfo {author} {\bibfnamefont {A.~M.}\ \bibnamefont {Souza}}, \bibinfo
  {author} {\bibfnamefont {R.~S.}\ \bibnamefont {Sarthour}}, \bibinfo {author}
  {\bibfnamefont {I.~S.}\ \bibnamefont {Oliveira}},\ and\ \bibinfo {author}
  {\bibfnamefont {R.~M.}\ \bibnamefont {Serra}},\ }\bibfield  {title} {\bibinfo
  {title} {Experimental rectification of entropy production by maxwell's demon
  in a quantum system},\ }\href
  {https://doi.org/10.1103/PhysRevLett.117.240502} {\bibfield  {journal}
  {\bibinfo  {journal} {Phys. Rev. Lett.}\ }\textbf {\bibinfo {volume} {117}},\
  \bibinfo {pages} {240502} (\bibinfo {year} {2016})}\BibitemShut {NoStop}%
\bibitem [{\citenamefont {S.}\ \emph {et~al.}(2016)\citenamefont {S.},
  \citenamefont {S.}, \citenamefont {M.}, \citenamefont {S.}, \citenamefont
  {J.}, \citenamefont {K.}, \citenamefont {O.},\ and\ \citenamefont
  {C.}}]{Peterson2016}%
  \BibitemOpen
  \bibfield  {author} {\bibinfo {author} {\bibfnamefont {P.~J.~P.}\
  \bibnamefont {S.}}, \bibinfo {author} {\bibfnamefont {S.~R.}\ \bibnamefont
  {S.}}, \bibinfo {author} {\bibfnamefont {S.~A.}\ \bibnamefont {M.}}, \bibinfo
  {author} {\bibfnamefont {O.~I.}\ \bibnamefont {S.}}, \bibinfo {author}
  {\bibfnamefont {G.}~\bibnamefont {J.}}, \bibinfo {author} {\bibfnamefont
  {M.}~\bibnamefont {K.}}, \bibinfo {author} {\bibfnamefont {S.-P.~D.}\
  \bibnamefont {O.}},\ and\ \bibinfo {author} {\bibfnamefont {C.~L.}\
  \bibnamefont {C.}},\ }\bibfield  {title} {\bibinfo {title} {Experimental
  demonstration of information to energy conversion in a quantum system at the
  landauer limit},\ }\href {https://doi.org/10.1098/rspa.2015.0813} {\bibfield
  {journal} {\bibinfo  {journal} {Proc. Royal Soc. A}\ }\textbf {\bibinfo
  {volume} {472}},\ \bibinfo {pages} {20150813} (\bibinfo {year}
  {2016})}\BibitemShut {NoStop}%
\bibitem [{\citenamefont {{Zanin}}\ \emph {et~al.}(2019)\citenamefont
  {{Zanin}}, \citenamefont {{H{\"a}ffner}}, \citenamefont {{Talarico}},
  \citenamefont {{Duzzioni}}, \citenamefont {{Souto Ribeiro}}, \citenamefont
  {{Landi}},\ and\ \citenamefont {{C{\'e}leri}}}]{Zanin2019}%
  \BibitemOpen
  \bibfield  {author} {\bibinfo {author} {\bibfnamefont {G.~L.}\ \bibnamefont
  {{Zanin}}}, \bibinfo {author} {\bibfnamefont {T.}~\bibnamefont
  {{H{\"a}ffner}}}, \bibinfo {author} {\bibfnamefont {M.~A.~A.}\ \bibnamefont
  {{Talarico}}}, \bibinfo {author} {\bibfnamefont {E.~I.}\ \bibnamefont
  {{Duzzioni}}}, \bibinfo {author} {\bibfnamefont {P.~H.}\ \bibnamefont {{Souto
  Ribeiro}}}, \bibinfo {author} {\bibfnamefont {G.~T.}\ \bibnamefont
  {{Landi}}},\ and\ \bibinfo {author} {\bibfnamefont {L.~C.}\ \bibnamefont
  {{C{\'e}leri}}},\ }\bibfield  {title} {\bibinfo {title} {{Experimental
  quantum thermodynamics with linear optics}},\ }\href@noop {} {\bibfield
  {journal} {\bibinfo  {journal} {arXiv e-prints}\ } (\bibinfo {year}
  {2019})},\ \Eprint {https://arxiv.org/abs/1905.02829} {arXiv:1905.02829}
  \BibitemShut {NoStop}%
\bibitem [{\citenamefont {Shannon}(1948)}]{shannon1948mathematical}%
  \BibitemOpen
  \bibfield  {author} {\bibinfo {author} {\bibfnamefont {C.~E.}\ \bibnamefont
  {Shannon}},\ }\bibfield  {title} {\bibinfo {title} {A mathematical theory of
  communication},\ }\href@noop {} {\bibfield  {journal} {\bibinfo  {journal}
  {Bell Labs Tech. J.}\ }\textbf {\bibinfo {volume} {27}},\ \bibinfo {pages}
  {379} (\bibinfo {year} {1948})}\BibitemShut {NoStop}%
\bibitem [{\citenamefont {Ruch}\ and\ \citenamefont {Mead}(1976)}]{Ruch1976}%
  \BibitemOpen
  \bibfield  {author} {\bibinfo {author} {\bibfnamefont {E.}~\bibnamefont
  {Ruch}}\ and\ \bibinfo {author} {\bibfnamefont {A.}~\bibnamefont {Mead}},\
  }\bibfield  {title} {\bibinfo {title} {The principle of increasing mixing
  character and some of its consequences},\ }\href
  {https://doi.org/10.1007/BF01178071} {\bibfield  {journal} {\bibinfo
  {journal} {Theor. Chim. Acta}\ }\textbf {\bibinfo {volume} {41}},\ \bibinfo
  {pages} {95} (\bibinfo {year} {1976})}\BibitemShut {NoStop}%
\bibitem [{\citenamefont {Marshall}\ \emph {et~al.}(2011)\citenamefont
  {Marshall}, \citenamefont {Olkin},\ and\ \citenamefont
  {Arnold}}]{Marshall2011}%
  \BibitemOpen
  \bibfield  {author} {\bibinfo {author} {\bibfnamefont {A.~W.}\ \bibnamefont
  {Marshall}}, \bibinfo {author} {\bibfnamefont {I.}~\bibnamefont {Olkin}},\
  and\ \bibinfo {author} {\bibfnamefont {B.~C.}\ \bibnamefont {Arnold}},\
  }\href {https://doi.org/10.1007/978-0-387-68276-1} {\emph {\bibinfo {title}
  {Inequalities: Theory of Majorization and its Applications}}},\ \bibinfo
  {edition} {2nd}\ ed.,\ Vol.\ \bibinfo {volume} {143}\ (\bibinfo  {publisher}
  {Springer},\ \bibinfo {year} {2011})\BibitemShut {NoStop}%
\bibitem [{\citenamefont {Janzing}\ \emph {et~al.}(2000)\citenamefont
  {Janzing}, \citenamefont {Wocjan}, \citenamefont {Zeier}, \citenamefont
  {Geiss},\ and\ \citenamefont {Beth}}]{Janzing2000}%
  \BibitemOpen
  \bibfield  {author} {\bibinfo {author} {\bibfnamefont {D.}~\bibnamefont
  {Janzing}}, \bibinfo {author} {\bibfnamefont {P.}~\bibnamefont {Wocjan}},
  \bibinfo {author} {\bibfnamefont {R.}~\bibnamefont {Zeier}}, \bibinfo
  {author} {\bibfnamefont {R.}~\bibnamefont {Geiss}},\ and\ \bibinfo {author}
  {\bibfnamefont {T.}~\bibnamefont {Beth}},\ }\bibfield  {title} {\bibinfo
  {title} {Thermodynamic cost of reliability and low temperatures: Tightening
  landauer's principle and the second law},\ }\href
  {https://doi.org/10.1023/A:1026422630734} {\bibfield  {journal} {\bibinfo
  {journal} {Int. J. Theor. Phys.}\ }\textbf {\bibinfo {volume} {39}},\
  \bibinfo {pages} {2717} (\bibinfo {year} {2000})}\BibitemShut {NoStop}%
\bibitem [{\citenamefont {Brandão}\ \emph {et~al.}(2015)\citenamefont
  {Brandão}, \citenamefont {Horodecki}, \citenamefont {Ng}, \citenamefont
  {Oppenheim},\ and\ \citenamefont {Wehner}}]{Brand_o_2015}%
  \BibitemOpen
  \bibfield  {author} {\bibinfo {author} {\bibfnamefont {F.}~\bibnamefont
  {Brandão}}, \bibinfo {author} {\bibfnamefont {M.}~\bibnamefont {Horodecki}},
  \bibinfo {author} {\bibfnamefont {N.}~\bibnamefont {Ng}}, \bibinfo {author}
  {\bibfnamefont {J.}~\bibnamefont {Oppenheim}},\ and\ \bibinfo {author}
  {\bibfnamefont {S.}~\bibnamefont {Wehner}},\ }\bibfield  {title} {\bibinfo
  {title} {The second laws of quantum thermodynamics},\ }\href
  {https://doi.org/10.1073/pnas.1411728112} {\bibfield  {journal} {\bibinfo
  {journal} {PNAS}\ }\textbf {\bibinfo {volume} {112}},\ \bibinfo {pages}
  {3275–3279} (\bibinfo {year} {2015})}\BibitemShut {NoStop}%
\bibitem [{\citenamefont {Rényi}(1961)}]{renyi1961}%
  \BibitemOpen
  \bibfield  {author} {\bibinfo {author} {\bibfnamefont {A.}~\bibnamefont
  {Rényi}},\ }\bibfield  {title} {\bibinfo {title} {On measures of entropy and
  information},\ }in\ \href {https://projecteuclid.org/euclid.bsmsp/1200512181}
  {\emph {\bibinfo {booktitle} {Proceedings of the Fourth Berkeley Symposium on
  Mathematical Statistics and Probability, Volume 1: Contributions to the
  Theory of Statistics}}}\ (\bibinfo  {publisher} {University of California
  Press},\ \bibinfo {address} {Berkeley, Calif.},\ \bibinfo {year} {1961})\
  pp.\ \bibinfo {pages} {547--561}\BibitemShut {NoStop}%
\bibitem [{\citenamefont {Jonathan}\ and\ \citenamefont
  {Plenio}(1999)}]{Jonathan_1999}%
  \BibitemOpen
  \bibfield  {author} {\bibinfo {author} {\bibfnamefont {D.}~\bibnamefont
  {Jonathan}}\ and\ \bibinfo {author} {\bibfnamefont {M.~B.}\ \bibnamefont
  {Plenio}},\ }\bibfield  {title} {\bibinfo {title} {Entanglement-assisted
  local manipulation of pure quantum states},\ }\href
  {https://doi.org/10.1103/physrevlett.83.3566} {\bibfield  {journal} {\bibinfo
   {journal} {Phys. Rev. Lett.}\ }\textbf {\bibinfo {volume} {83}},\ \bibinfo
  {pages} {3566–3569} (\bibinfo {year} {1999})}\BibitemShut {NoStop}%
\bibitem [{\citenamefont {Horodecki}\ and\ \citenamefont
  {Oppenheim}(2012)}]{HORODECKI2012}%
  \BibitemOpen
  \bibfield  {author} {\bibinfo {author} {\bibfnamefont {M.}~\bibnamefont
  {Horodecki}}\ and\ \bibinfo {author} {\bibfnamefont {J.}~\bibnamefont
  {Oppenheim}},\ }\bibfield  {title} {\bibinfo {title} {(quantumness in the
  context of) resource theories},\ }\href
  {https://doi.org/10.1142/s0217979213450197} {\bibfield  {journal} {\bibinfo
  {journal} {Int. J. Mod. Phys. B}\ }\textbf {\bibinfo {volume} {27}},\
  \bibinfo {pages} {1345019} (\bibinfo {year} {2012})}\BibitemShut {NoStop}%
\bibitem [{\citenamefont {Chitambar}\ and\ \citenamefont
  {Gour}(2019)}]{Chitambar2019}%
  \BibitemOpen
  \bibfield  {author} {\bibinfo {author} {\bibfnamefont {E.}~\bibnamefont
  {Chitambar}}\ and\ \bibinfo {author} {\bibfnamefont {G.}~\bibnamefont
  {Gour}},\ }\bibfield  {title} {\bibinfo {title} {Quantum resource theories},\
  }\href {https://doi.org/10.1103/RevModPhys.91.025001} {\bibfield  {journal}
  {\bibinfo  {journal} {Rev. Mod. Phys.}\ }\textbf {\bibinfo {volume} {91}},\
  \bibinfo {pages} {025001} (\bibinfo {year} {2019})}\BibitemShut {NoStop}%
\bibitem [{\citenamefont {Nielsen}(1999)}]{Nielsen1999}%
  \BibitemOpen
  \bibfield  {author} {\bibinfo {author} {\bibfnamefont {M.~A.}\ \bibnamefont
  {Nielsen}},\ }\bibfield  {title} {\bibinfo {title} {Conditions for a class of
  entanglement transformations},\ }\href
  {https://doi.org/10.1103/PhysRevLett.83.436} {\bibfield  {journal} {\bibinfo
  {journal} {Phys. Rev. Lett.}\ }\textbf {\bibinfo {volume} {83}},\ \bibinfo
  {pages} {436} (\bibinfo {year} {1999})}\BibitemShut {NoStop}%
\bibitem [{\citenamefont {Vidal}\ and\ \citenamefont
  {Tarrach}(1999)}]{Vidal1999}%
  \BibitemOpen
  \bibfield  {author} {\bibinfo {author} {\bibfnamefont {G.}~\bibnamefont
  {Vidal}}\ and\ \bibinfo {author} {\bibfnamefont {R.}~\bibnamefont
  {Tarrach}},\ }\bibfield  {title} {\bibinfo {title} {Robustness of
  entanglement},\ }\href@noop {} {\bibfield  {journal} {\bibinfo  {journal}
  {Phys. Rev. A}\ }\textbf {\bibinfo {volume} {59}},\ \bibinfo {pages} {141}
  (\bibinfo {year} {1999})}\BibitemShut {NoStop}%
\bibitem [{\citenamefont {Du}\ \emph {et~al.}(2015)\citenamefont {Du},
  \citenamefont {Bai},\ and\ \citenamefont {Guo}}]{PhysRevA.91.052120}%
  \BibitemOpen
  \bibfield  {author} {\bibinfo {author} {\bibfnamefont {S.}~\bibnamefont
  {Du}}, \bibinfo {author} {\bibfnamefont {Z.}~\bibnamefont {Bai}},\ and\
  \bibinfo {author} {\bibfnamefont {Y.}~\bibnamefont {Guo}},\ }\bibfield
  {title} {\bibinfo {title} {Conditions for coherence transformations under
  incoherent operations},\ }\href {https://doi.org/10.1103/PhysRevA.91.052120}
  {\bibfield  {journal} {\bibinfo  {journal} {Phys. Rev. A}\ }\textbf {\bibinfo
  {volume} {91}},\ \bibinfo {pages} {052120} (\bibinfo {year}
  {2015})}\BibitemShut {NoStop}%
\bibitem [{\citenamefont {Horodecki}\ \emph {et~al.}(2003)\citenamefont
  {Horodecki}, \citenamefont {Shor},\ and\ \citenamefont
  {Ruskai}}]{Horodecki2003}%
  \BibitemOpen
  \bibfield  {author} {\bibinfo {author} {\bibfnamefont {M.}~\bibnamefont
  {Horodecki}}, \bibinfo {author} {\bibfnamefont {P.~W.}\ \bibnamefont
  {Shor}},\ and\ \bibinfo {author} {\bibfnamefont {M.~B.}\ \bibnamefont
  {Ruskai}},\ }\bibfield  {title} {\bibinfo {title} {Entanglement breaking
  channels},\ }\href {https://doi.org/10.1142/S0129055X03001709} {\bibfield
  {journal} {\bibinfo  {journal} {Rev. Math. Phys.}\ }\textbf {\bibinfo
  {volume} {15}},\ \bibinfo {pages} {629} (\bibinfo {year} {2003})}\BibitemShut
  {NoStop}%
\bibitem [{\citenamefont {Streltsov}\ \emph
  {et~al.}(2018{\natexlab{a}})\citenamefont {Streltsov}, \citenamefont
  {Kampermann}, \citenamefont {Wölk}, \citenamefont {Gessner},\ and\
  \citenamefont {Bru{\ss}}}]{Streltsov2018}%
  \BibitemOpen
  \bibfield  {author} {\bibinfo {author} {\bibfnamefont {A.}~\bibnamefont
  {Streltsov}}, \bibinfo {author} {\bibfnamefont {H.}~\bibnamefont
  {Kampermann}}, \bibinfo {author} {\bibfnamefont {S.}~\bibnamefont {Wölk}},
  \bibinfo {author} {\bibfnamefont {M.}~\bibnamefont {Gessner}},\ and\ \bibinfo
  {author} {\bibfnamefont {D.}~\bibnamefont {Bru{\ss}}},\ }\bibfield  {title}
  {\bibinfo {title} {Maximal coherence and the resource theory of purity},\
  }\href {https://doi.org/10.1088/1367-2630/aac484} {\bibfield  {journal}
  {\bibinfo  {journal} {New J. Phys.}\ }\textbf {\bibinfo {volume} {20}},\
  \bibinfo {pages} {053058} (\bibinfo {year} {2018}{\natexlab{a}})}\BibitemShut
  {NoStop}%
\bibitem [{\citenamefont {Brand\~ao}\ \emph
  {et~al.}(2013{\natexlab{a}})\citenamefont {Brand\~ao}, \citenamefont
  {Horodecki}, \citenamefont {Oppenheim}, \citenamefont {Renes},\ and\
  \citenamefont {Spekkens}}]{Brandao2013}%
  \BibitemOpen
  \bibfield  {author} {\bibinfo {author} {\bibfnamefont {F.~G. S.~L.}\
  \bibnamefont {Brand\~ao}}, \bibinfo {author} {\bibfnamefont {M.}~\bibnamefont
  {Horodecki}}, \bibinfo {author} {\bibfnamefont {J.}~\bibnamefont
  {Oppenheim}}, \bibinfo {author} {\bibfnamefont {J.~M.}\ \bibnamefont
  {Renes}},\ and\ \bibinfo {author} {\bibfnamefont {R.~W.}\ \bibnamefont
  {Spekkens}},\ }\bibfield  {title} {\bibinfo {title} {Resource theory of
  quantum states out of thermal equilibrium},\ }\href
  {https://doi.org/10.1103/PhysRevLett.111.250404} {\bibfield  {journal}
  {\bibinfo  {journal} {Phys. Rev. Lett.}\ }\textbf {\bibinfo {volume} {111}},\
  \bibinfo {pages} {250404} (\bibinfo {year} {2013}{\natexlab{a}})}\BibitemShut
  {NoStop}%
\bibitem [{\citenamefont {Horodecki}\ and\ \citenamefont
  {Oppenheim}(2013)}]{Horodecki2013}%
  \BibitemOpen
  \bibfield  {author} {\bibinfo {author} {\bibfnamefont {M.}~\bibnamefont
  {Horodecki}}\ and\ \bibinfo {author} {\bibfnamefont {J.}~\bibnamefont
  {Oppenheim}},\ }\bibfield  {title} {\bibinfo {title} {Fundamental limitations
  for quantum and nanoscale thermodynamics},\ }\href
  {https://doi.org/10.1038/ncomms3059} {\bibfield  {journal} {\bibinfo
  {journal} {Nat. Commun.}\ }\textbf {\bibinfo {volume} {4}},\ \bibinfo {pages}
  {2059} (\bibinfo {year} {2013})}\BibitemShut {NoStop}%
\bibitem [{\citenamefont {Ng}\ and\ \citenamefont {Woods}(2018)}]{Ng2018}%
  \BibitemOpen
  \bibfield  {author} {\bibinfo {author} {\bibfnamefont {N.~H.~Y.}\
  \bibnamefont {Ng}}\ and\ \bibinfo {author} {\bibfnamefont {M.~P.}\
  \bibnamefont {Woods}},\ }\bibfield  {title} {\bibinfo {title} {Resource
  theory of quantum thermodynamics: Thermal operations and second laws},\
  }\href {https://doi.org/10.1007/978-3-319-99046-0_26} {\bibfield  {journal}
  {\bibinfo  {journal} {Thermodynamics in the Quantum Regime}\ ,\ \bibinfo
  {pages} {625–650}} (\bibinfo {year} {2018})}\BibitemShut {NoStop}%
\bibitem [{\citenamefont {Horodecki}\ \emph {et~al.}(2015)\citenamefont
  {Horodecki}, \citenamefont {Grudka}, \citenamefont {Joshi}, \citenamefont
  {K\l{}obus},\ and\ \citenamefont {\L{}odyga}}]{Horodecki2015}%
  \BibitemOpen
  \bibfield  {author} {\bibinfo {author} {\bibfnamefont {K.}~\bibnamefont
  {Horodecki}}, \bibinfo {author} {\bibfnamefont {A.}~\bibnamefont {Grudka}},
  \bibinfo {author} {\bibfnamefont {P.}~\bibnamefont {Joshi}}, \bibinfo
  {author} {\bibfnamefont {W.}~\bibnamefont {K\l{}obus}},\ and\ \bibinfo
  {author} {\bibfnamefont {J.}~\bibnamefont {\L{}odyga}},\ }\bibfield  {title}
  {\bibinfo {title} {Axiomatic approach to contextuality and nonlocality},\
  }\href {https://doi.org/10.1103/PhysRevA.92.032104} {\bibfield  {journal}
  {\bibinfo  {journal} {Phys. Rev. A}\ }\textbf {\bibinfo {volume} {92}},\
  \bibinfo {pages} {032104} (\bibinfo {year} {2015})}\BibitemShut {NoStop}%
\bibitem [{\citenamefont {Cavalcanti}\ and\ \citenamefont
  {Skrzypczyk}(2016)}]{Cavalcanti2016}%
  \BibitemOpen
  \bibfield  {author} {\bibinfo {author} {\bibfnamefont {D.}~\bibnamefont
  {Cavalcanti}}\ and\ \bibinfo {author} {\bibfnamefont {P.}~\bibnamefont
  {Skrzypczyk}},\ }\bibfield  {title} {\bibinfo {title} {Quantitative relations
  between measurement incompatibility, quantum steering, and nonlocality},\
  }\href {https://doi.org/10.1103/PhysRevA.93.052112} {\bibfield  {journal}
  {\bibinfo  {journal} {Phys. Rev. A}\ }\textbf {\bibinfo {volume} {93}},\
  \bibinfo {pages} {052112} (\bibinfo {year} {2016})}\BibitemShut {NoStop}%
\bibitem [{\citenamefont {Piani}\ \emph {et~al.}(2016)\citenamefont {Piani},
  \citenamefont {Cianciaruso}, \citenamefont {Bromley}, \citenamefont {Napoli},
  \citenamefont {Johnston},\ and\ \citenamefont {Adesso}}]{Piani2016}%
  \BibitemOpen
  \bibfield  {author} {\bibinfo {author} {\bibfnamefont {M.}~\bibnamefont
  {Piani}}, \bibinfo {author} {\bibfnamefont {M.}~\bibnamefont {Cianciaruso}},
  \bibinfo {author} {\bibfnamefont {T.~R.}\ \bibnamefont {Bromley}}, \bibinfo
  {author} {\bibfnamefont {C.}~\bibnamefont {Napoli}}, \bibinfo {author}
  {\bibfnamefont {N.}~\bibnamefont {Johnston}},\ and\ \bibinfo {author}
  {\bibfnamefont {G.}~\bibnamefont {Adesso}},\ }\bibfield  {title} {\bibinfo
  {title} {Robustness of asymmetry and coherence of quantum states},\ }\href
  {https://doi.org/10.1103/PhysRevA.93.042107} {\bibfield  {journal} {\bibinfo
  {journal} {Phys. Rev. A}\ }\textbf {\bibinfo {volume} {93}},\ \bibinfo
  {pages} {042107} (\bibinfo {year} {2016})}\BibitemShut {NoStop}%
\bibitem [{\citenamefont {Buscemi}\ \emph {et~al.}(2020)\citenamefont
  {Buscemi}, \citenamefont {Chitambar},\ and\ \citenamefont
  {Zhou}}]{Buscemi2020}%
  \BibitemOpen
  \bibfield  {author} {\bibinfo {author} {\bibfnamefont {F.}~\bibnamefont
  {Buscemi}}, \bibinfo {author} {\bibfnamefont {E.}~\bibnamefont {Chitambar}},\
  and\ \bibinfo {author} {\bibfnamefont {W.}~\bibnamefont {Zhou}},\ }\bibfield
  {title} {\bibinfo {title} {Complete resource theory of quantum
  incompatibility as quantum programmability},\ }\href
  {https://doi.org/10.1103/PhysRevLett.124.120401} {\bibfield  {journal}
  {\bibinfo  {journal} {Phys. Rev. Lett.}\ }\textbf {\bibinfo {volume} {124}},\
  \bibinfo {pages} {120401} (\bibinfo {year} {2020})}\BibitemShut {NoStop}%
\bibitem [{\citenamefont {\ifmmode \check{S}\else
  \v{S}\fi{}upi\ifmmode~\acute{c}\else \'{c}\fi{}}\ \emph
  {et~al.}(2019)\citenamefont {\ifmmode \check{S}\else
  \v{S}\fi{}upi\ifmmode~\acute{c}\else \'{c}\fi{}}, \citenamefont
  {Skrzypczyk},\ and\ \citenamefont {Cavalcanti}}]{Supic2019}%
  \BibitemOpen
  \bibfield  {author} {\bibinfo {author} {\bibfnamefont {I.}~\bibnamefont
  {\ifmmode \check{S}\else \v{S}\fi{}upi\ifmmode~\acute{c}\else \'{c}\fi{}}},
  \bibinfo {author} {\bibfnamefont {P.}~\bibnamefont {Skrzypczyk}},\ and\
  \bibinfo {author} {\bibfnamefont {D.}~\bibnamefont {Cavalcanti}},\ }\bibfield
   {title} {\bibinfo {title} {Methods to estimate entanglement in teleportation
  experiments},\ }\href {https://doi.org/10.1103/PhysRevA.99.032334} {\bibfield
   {journal} {\bibinfo  {journal} {Phys. Rev. A}\ }\textbf {\bibinfo {volume}
  {99}},\ \bibinfo {pages} {032334} (\bibinfo {year} {2019})}\BibitemShut
  {NoStop}%
\bibitem [{\citenamefont {Cavalcanti}\ \emph {et~al.}(2017)\citenamefont
  {Cavalcanti}, \citenamefont {Skrzypczyk},\ and\ \citenamefont {\ifmmode
  \check{S}\else \v{S}\fi{}upi\ifmmode~\acute{c}\else
  \'{c}\fi{}}}]{Cavalcanti2017}%
  \BibitemOpen
  \bibfield  {author} {\bibinfo {author} {\bibfnamefont {D.}~\bibnamefont
  {Cavalcanti}}, \bibinfo {author} {\bibfnamefont {P.}~\bibnamefont
  {Skrzypczyk}},\ and\ \bibinfo {author} {\bibfnamefont {I.}~\bibnamefont
  {\ifmmode \check{S}\else \v{S}\fi{}upi\ifmmode~\acute{c}\else \'{c}\fi{}}},\
  }\bibfield  {title} {\bibinfo {title} {All entangled states can demonstrate
  nonclassical teleportation},\ }\href
  {https://doi.org/10.1103/PhysRevLett.119.110501} {\bibfield  {journal}
  {\bibinfo  {journal} {Phys. Rev. Lett.}\ }\textbf {\bibinfo {volume} {119}},\
  \bibinfo {pages} {110501} (\bibinfo {year} {2017})}\BibitemShut {NoStop}%
\bibitem [{\citenamefont {Howard}\ and\ \citenamefont
  {Campbell}(2017)}]{magic2017}%
  \BibitemOpen
  \bibfield  {author} {\bibinfo {author} {\bibfnamefont {M.}~\bibnamefont
  {Howard}}\ and\ \bibinfo {author} {\bibfnamefont {E.}~\bibnamefont
  {Campbell}},\ }\bibfield  {title} {\bibinfo {title} {Application of a
  resource theory for magic states to fault-tolerant quantum computing},\
  }\href {https://doi.org/10.1103/PhysRevLett.118.090501} {\bibfield  {journal}
  {\bibinfo  {journal} {Phys. Rev. Lett.}\ }\textbf {\bibinfo {volume} {118}},\
  \bibinfo {pages} {090501} (\bibinfo {year} {2017})}\BibitemShut {NoStop}%
\bibitem [{\citenamefont {Bhattacharya}\ \emph {et~al.}(2018)\citenamefont
  {Bhattacharya}, \citenamefont {Bhattacharya},\ and\ \citenamefont
  {Majumdar}}]{bhattacharya2018convex}%
  \BibitemOpen
  \bibfield  {author} {\bibinfo {author} {\bibfnamefont {S.}~\bibnamefont
  {Bhattacharya}}, \bibinfo {author} {\bibfnamefont {B.}~\bibnamefont
  {Bhattacharya}},\ and\ \bibinfo {author} {\bibfnamefont {A.~S.}\ \bibnamefont
  {Majumdar}},\ }\href@noop {} {\bibinfo {title} {Convex resource theory of
  non-markovianity}} (\bibinfo {year} {2018}),\ \Eprint
  {https://arxiv.org/abs/1803.06881} {arXiv:1803.06881 [quant-ph]} \BibitemShut
  {NoStop}%
\bibitem [{\citenamefont {Wakakuwa}(2017)}]{wakakuwa2017operational}%
  \BibitemOpen
  \bibfield  {author} {\bibinfo {author} {\bibfnamefont {E.}~\bibnamefont
  {Wakakuwa}},\ }\href@noop {} {\bibinfo {title} {Operational resource theory
  of non-markovianity}} (\bibinfo {year} {2017}),\ \Eprint
  {https://arxiv.org/abs/1709.07248} {arXiv:1709.07248 [quant-ph]} \BibitemShut
  {NoStop}%
\bibitem [{\citenamefont {Bennett}\ \emph {et~al.}(1996)\citenamefont
  {Bennett}, \citenamefont {Bernstein}, \citenamefont {Popescu},\ and\
  \citenamefont {Schumacher}}]{PhysRevA.53.2046}%
  \BibitemOpen
  \bibfield  {author} {\bibinfo {author} {\bibfnamefont {C.~H.}\ \bibnamefont
  {Bennett}}, \bibinfo {author} {\bibfnamefont {H.~J.}\ \bibnamefont
  {Bernstein}}, \bibinfo {author} {\bibfnamefont {S.}~\bibnamefont {Popescu}},\
  and\ \bibinfo {author} {\bibfnamefont {B.}~\bibnamefont {Schumacher}},\
  }\bibfield  {title} {\bibinfo {title} {Concentrating partial entanglement by
  local operations},\ }\href {https://doi.org/10.1103/PhysRevA.53.2046}
  {\bibfield  {journal} {\bibinfo  {journal} {Phys. Rev. A}\ }\textbf {\bibinfo
  {volume} {53}},\ \bibinfo {pages} {2046} (\bibinfo {year}
  {1996})}\BibitemShut {NoStop}%
\bibitem [{\citenamefont {Kumagai}\ and\ \citenamefont
  {Hayashi}(2016)}]{kumagai2016second}%
  \BibitemOpen
  \bibfield  {author} {\bibinfo {author} {\bibfnamefont {W.}~\bibnamefont
  {Kumagai}}\ and\ \bibinfo {author} {\bibfnamefont {M.}~\bibnamefont
  {Hayashi}},\ }\bibfield  {title} {\bibinfo {title} {Second-order asymptotics
  of conversions of distributions and entangled states based on rayleigh-normal
  probability distributions},\ }\href@noop {} {\bibfield  {journal} {\bibinfo
  {journal} {IEEE Transactions on Information Theory}\ }\textbf {\bibinfo
  {volume} {63}},\ \bibinfo {pages} {1829} (\bibinfo {year}
  {2016})}\BibitemShut {NoStop}%
\bibitem [{\citenamefont {Napoli}\ \emph {et~al.}(2016)\citenamefont {Napoli},
  \citenamefont {Bromley}, \citenamefont {Cianciaruso}, \citenamefont {Piani},
  \citenamefont {Johnston},\ and\ \citenamefont {Adesso}}]{Napoli2016}%
  \BibitemOpen
  \bibfield  {author} {\bibinfo {author} {\bibfnamefont {C.}~\bibnamefont
  {Napoli}}, \bibinfo {author} {\bibfnamefont {T.~R.}\ \bibnamefont {Bromley}},
  \bibinfo {author} {\bibfnamefont {M.}~\bibnamefont {Cianciaruso}}, \bibinfo
  {author} {\bibfnamefont {M.}~\bibnamefont {Piani}}, \bibinfo {author}
  {\bibfnamefont {N.}~\bibnamefont {Johnston}},\ and\ \bibinfo {author}
  {\bibfnamefont {G.}~\bibnamefont {Adesso}},\ }\bibfield  {title} {\bibinfo
  {title} {Robustness of coherence: An operational and observable measure of
  quantum coherence},\ }\href {https://doi.org/10.1103/PhysRevLett.116.150502}
  {\bibfield  {journal} {\bibinfo  {journal} {Phys. Rev. Lett.}\ }\textbf
  {\bibinfo {volume} {116}},\ \bibinfo {pages} {150502} (\bibinfo {year}
  {2016})}\BibitemShut {NoStop}%
\bibitem [{\citenamefont {Zhu}\ \emph {et~al.}(2017)\citenamefont {Zhu},
  \citenamefont {Ma}, \citenamefont {Cao}, \citenamefont {Fei},\ and\
  \citenamefont {Vedral}}]{Zhu_2017}%
  \BibitemOpen
  \bibfield  {author} {\bibinfo {author} {\bibfnamefont {H.}~\bibnamefont
  {Zhu}}, \bibinfo {author} {\bibfnamefont {Z.}~\bibnamefont {Ma}}, \bibinfo
  {author} {\bibfnamefont {Z.}~\bibnamefont {Cao}}, \bibinfo {author}
  {\bibfnamefont {S.-M.}\ \bibnamefont {Fei}},\ and\ \bibinfo {author}
  {\bibfnamefont {V.}~\bibnamefont {Vedral}},\ }\bibfield  {title} {\bibinfo
  {title} {Operational one-to-one mapping between coherence and entanglement
  measures},\ }\href {http://dx.doi.org/10.1103/PhysRevA.96.032316} {\bibfield
  {journal} {\bibinfo  {journal} {Phys. Rev. A}\ }\textbf {\bibinfo {volume}
  {96}} (\bibinfo {year} {2017})}\BibitemShut {NoStop}%
\bibitem [{\citenamefont {Chitambar}\ and\ \citenamefont
  {Gour}(2016)}]{chitambar2016comparison}%
  \BibitemOpen
  \bibfield  {author} {\bibinfo {author} {\bibfnamefont {E.}~\bibnamefont
  {Chitambar}}\ and\ \bibinfo {author} {\bibfnamefont {G.}~\bibnamefont
  {Gour}},\ }\bibfield  {title} {\bibinfo {title} {Comparison of incoherent
  operations and measures of coherence},\ }\href@noop {} {\bibfield  {journal}
  {\bibinfo  {journal} {Phys. Rev. A}\ }\textbf {\bibinfo {volume} {94}},\
  \bibinfo {pages} {052336} (\bibinfo {year} {2016})}\BibitemShut {NoStop}%
\bibitem [{\citenamefont {Gour}\ \emph {et~al.}(2015)\citenamefont {Gour},
  \citenamefont {Müller}, \citenamefont {Narasimhachar}, \citenamefont
  {Spekkens},\ and\ \citenamefont {Yunger~Halpern}}]{Gour_2015}%
  \BibitemOpen
  \bibfield  {author} {\bibinfo {author} {\bibfnamefont {G.}~\bibnamefont
  {Gour}}, \bibinfo {author} {\bibfnamefont {M.~P.}\ \bibnamefont {Müller}},
  \bibinfo {author} {\bibfnamefont {V.}~\bibnamefont {Narasimhachar}}, \bibinfo
  {author} {\bibfnamefont {R.~W.}\ \bibnamefont {Spekkens}},\ and\ \bibinfo
  {author} {\bibfnamefont {N.}~\bibnamefont {Yunger~Halpern}},\ }\bibfield
  {title} {\bibinfo {title} {The resource theory of informational
  nonequilibrium in thermodynamics},\ }\href
  {https://doi.org/10.1016/j.physrep.2015.04.003} {\bibfield  {journal}
  {\bibinfo  {journal} {Phys. Rep.}\ }\textbf {\bibinfo {volume} {583}},\
  \bibinfo {pages} {1–58} (\bibinfo {year} {2015})}\BibitemShut {NoStop}%
\bibitem [{\citenamefont {Streltsov}\ \emph
  {et~al.}(2018{\natexlab{b}})\citenamefont {Streltsov}, \citenamefont
  {Kampermann}, \citenamefont {W{\"o}lk}, \citenamefont {Gessner},\ and\
  \citenamefont {Bru{\ss}}}]{streltsov2018maximal}%
  \BibitemOpen
  \bibfield  {author} {\bibinfo {author} {\bibfnamefont {A.}~\bibnamefont
  {Streltsov}}, \bibinfo {author} {\bibfnamefont {H.}~\bibnamefont
  {Kampermann}}, \bibinfo {author} {\bibfnamefont {S.}~\bibnamefont
  {W{\"o}lk}}, \bibinfo {author} {\bibfnamefont {M.}~\bibnamefont {Gessner}},\
  and\ \bibinfo {author} {\bibfnamefont {D.}~\bibnamefont {Bru{\ss}}},\
  }\bibfield  {title} {\bibinfo {title} {Maximal coherence and the resource
  theory of purity},\ }\href@noop {} {\bibfield  {journal} {\bibinfo  {journal}
  {New J. Phys.}\ }\textbf {\bibinfo {volume} {20}},\ \bibinfo {pages} {053058}
  (\bibinfo {year} {2018}{\natexlab{b}})}\BibitemShut {NoStop}%
\bibitem [{\citenamefont {Ding}\ \emph {et~al.}(2020)\citenamefont {Ding},
  \citenamefont {Hu},\ and\ \citenamefont {Fan}}]{ding2020amplifying}%
  \BibitemOpen
  \bibfield  {author} {\bibinfo {author} {\bibfnamefont {F.}~\bibnamefont
  {Ding}}, \bibinfo {author} {\bibfnamefont {X.}~\bibnamefont {Hu}},\ and\
  \bibinfo {author} {\bibfnamefont {H.}~\bibnamefont {Fan}},\ }\href@noop {}
  {\bibinfo {title} {Amplifying asymmetry with correlated catalysts}} (\bibinfo
  {year} {2020}),\ \Eprint {https://arxiv.org/abs/2007.06247} {arXiv:2007.06247
  [quant-ph]} \BibitemShut {NoStop}%
\bibitem [{\citenamefont {Brand\~ao}\ \emph
  {et~al.}(2013{\natexlab{b}})\citenamefont {Brand\~ao}, \citenamefont
  {Horodecki}, \citenamefont {Oppenheim}, \citenamefont {Renes},\ and\
  \citenamefont {Spekkens}}]{PhysRevLett.111.250404}%
  \BibitemOpen
  \bibfield  {author} {\bibinfo {author} {\bibfnamefont {F.~G. S.~L.}\
  \bibnamefont {Brand\~ao}}, \bibinfo {author} {\bibfnamefont {M.}~\bibnamefont
  {Horodecki}}, \bibinfo {author} {\bibfnamefont {J.}~\bibnamefont
  {Oppenheim}}, \bibinfo {author} {\bibfnamefont {J.~M.}\ \bibnamefont
  {Renes}},\ and\ \bibinfo {author} {\bibfnamefont {R.~W.}\ \bibnamefont
  {Spekkens}},\ }\bibfield  {title} {\bibinfo {title} {Resource theory of
  quantum states out of thermal equilibrium},\ }\href
  {https://doi.org/10.1103/PhysRevLett.111.250404} {\bibfield  {journal}
  {\bibinfo  {journal} {Phys. Rev. Lett.}\ }\textbf {\bibinfo {volume} {111}},\
  \bibinfo {pages} {250404} (\bibinfo {year} {2013}{\natexlab{b}})}\BibitemShut
  {NoStop}%
\bibitem [{\citenamefont {van Dam}\ and\ \citenamefont
  {Hayden}(2003)}]{vanDam2003}%
  \BibitemOpen
  \bibfield  {author} {\bibinfo {author} {\bibfnamefont {W.}~\bibnamefont {van
  Dam}}\ and\ \bibinfo {author} {\bibfnamefont {P.}~\bibnamefont {Hayden}},\
  }\bibfield  {title} {\bibinfo {title} {Universal entanglement transformations
  without communication},\ }\bibfield  {journal} {\bibinfo  {journal} {Phys.
  Rev. A}\ }\textbf {\bibinfo {volume} {67}},\ \href
  {https://doi.org/10.1103/physreva.67.060302} {10.1103/physreva.67.060302}
  (\bibinfo {year} {2003})\BibitemShut {NoStop}%
\bibitem [{\citenamefont {Ng}\ \emph {et~al.}(2015)\citenamefont {Ng},
  \citenamefont {Mančinska}, \citenamefont {Cirstoiu}, \citenamefont
  {Eisert},\ and\ \citenamefont {Wehner}}]{Ng_2015}%
  \BibitemOpen
  \bibfield  {author} {\bibinfo {author} {\bibfnamefont {N.~H.~Y.}\
  \bibnamefont {Ng}}, \bibinfo {author} {\bibfnamefont {L.}~\bibnamefont
  {Mančinska}}, \bibinfo {author} {\bibfnamefont {C.}~\bibnamefont
  {Cirstoiu}}, \bibinfo {author} {\bibfnamefont {J.}~\bibnamefont {Eisert}},\
  and\ \bibinfo {author} {\bibfnamefont {S.}~\bibnamefont {Wehner}},\
  }\bibfield  {title} {\bibinfo {title} {Limits to catalysis in quantum
  thermodynamics},\ }\href {https://doi.org/10.1088/1367-2630/17/8/085004}
  {\bibfield  {journal} {\bibinfo  {journal} {New J. Phys.}\ }\textbf {\bibinfo
  {volume} {17}},\ \bibinfo {pages} {085004} (\bibinfo {year}
  {2015})}\BibitemShut {NoStop}%
\bibitem [{\citenamefont {Nielsen}\ and\ \citenamefont
  {Chuang}(2011)}]{nielsen_chuang}%
  \BibitemOpen
  \bibfield  {author} {\bibinfo {author} {\bibfnamefont {M.~A.}\ \bibnamefont
  {Nielsen}}\ and\ \bibinfo {author} {\bibfnamefont {I.~L.}\ \bibnamefont
  {Chuang}},\ }\href@noop {} {\emph {\bibinfo {title} {Quantum Computation and
  Quantum Information: 10th Anniversary Edition}}},\ \bibinfo {edition} {10th}\
  ed.\ (\bibinfo  {publisher} {Cambridge University Press},\ \bibinfo {address}
  {USA},\ \bibinfo {year} {2011})\BibitemShut {NoStop}%
\bibitem [{\citenamefont {Brand\~ao}\ and\ \citenamefont
  {Gour}(2015)}]{Brandao2015}%
  \BibitemOpen
  \bibfield  {author} {\bibinfo {author} {\bibfnamefont {F.~G. S.~L.}\
  \bibnamefont {Brand\~ao}}\ and\ \bibinfo {author} {\bibfnamefont
  {G.}~\bibnamefont {Gour}},\ }\bibfield  {title} {\bibinfo {title} {Reversible
  framework for quantum resource theories},\ }\href
  {https://doi.org/10.1103/PhysRevLett.115.070503} {\bibfield  {journal}
  {\bibinfo  {journal} {Phys. Rev. Lett.}\ }\textbf {\bibinfo {volume} {115}},\
  \bibinfo {pages} {070503} (\bibinfo {year} {2015})}\BibitemShut {NoStop}%
\bibitem [{\citenamefont {Perry}\ \emph
  {et~al.}(2018{\natexlab{a}})\citenamefont {Perry}, \citenamefont {\ifmmode
  \acute{C}\else \'{C}\fi{}wikli\ifmmode~\acute{n}\else \'{n}\fi{}ski},
  \citenamefont {Anders}, \citenamefont {Horodecki},\ and\ \citenamefont
  {Oppenheim}}]{Perry2015}%
  \BibitemOpen
  \bibfield  {author} {\bibinfo {author} {\bibfnamefont {C.}~\bibnamefont
  {Perry}}, \bibinfo {author} {\bibfnamefont {P.}~\bibnamefont {\ifmmode
  \acute{C}\else \'{C}\fi{}wikli\ifmmode~\acute{n}\else \'{n}\fi{}ski}},
  \bibinfo {author} {\bibfnamefont {J.}~\bibnamefont {Anders}}, \bibinfo
  {author} {\bibfnamefont {M.}~\bibnamefont {Horodecki}},\ and\ \bibinfo
  {author} {\bibfnamefont {J.}~\bibnamefont {Oppenheim}},\ }\bibfield  {title}
  {\bibinfo {title} {A sufficient set of experimentally implementable thermal
  operations for small systems},\ }\href
  {https://doi.org/10.1103/PhysRevX.8.041049} {\bibfield  {journal} {\bibinfo
  {journal} {Phys. Rev. X}\ }\textbf {\bibinfo {volume} {8}},\ \bibinfo {pages}
  {041049} (\bibinfo {year} {2018}{\natexlab{a}})}\BibitemShut {NoStop}%
\bibitem [{\citenamefont {Richens}\ \emph {et~al.}(2018)\citenamefont
  {Richens}, \citenamefont {Alhambra},\ and\ \citenamefont
  {Masanes}}]{Richens2017}%
  \BibitemOpen
  \bibfield  {author} {\bibinfo {author} {\bibfnamefont {J.~G.}\ \bibnamefont
  {Richens}}, \bibinfo {author} {\bibfnamefont {A.~M.}\ \bibnamefont
  {Alhambra}},\ and\ \bibinfo {author} {\bibfnamefont {L.}~\bibnamefont
  {Masanes}},\ }\bibfield  {title} {\bibinfo {title} {Finite-bath corrections
  to the second law of thermodynamics},\ }\href
  {https://doi.org/10.1103/PhysRevE.97.062132} {\bibfield  {journal} {\bibinfo
  {journal} {Phys. Rev. E}\ }\textbf {\bibinfo {volume} {97}},\ \bibinfo
  {pages} {062132} (\bibinfo {year} {2018})}\BibitemShut {NoStop}%
\bibitem [{\citenamefont {Lostaglio}(2016)}]{Lostaglio2016}%
  \BibitemOpen
  \bibfield  {author} {\bibinfo {author} {\bibfnamefont {M.}~\bibnamefont
  {Lostaglio}},\ }\href@noop {} {\emph {\bibinfo {title} {The resource theory
  of quantum thermodynamics}}}\ (\bibinfo  {publisher} {Imperial College
  London},\ \bibinfo {year} {2016})\BibitemShut {NoStop}%
\bibitem [{\citenamefont {Perry}\ \emph
  {et~al.}(2018{\natexlab{b}})\citenamefont {Perry}, \citenamefont {\ifmmode
  \acute{C}\else \'{C}\fi{}wikli\ifmmode~\acute{n}\else \'{n}\fi{}ski},
  \citenamefont {Anders}, \citenamefont {Horodecki},\ and\ \citenamefont
  {Oppenheim}}]{Perry2018}%
  \BibitemOpen
  \bibfield  {author} {\bibinfo {author} {\bibfnamefont {C.}~\bibnamefont
  {Perry}}, \bibinfo {author} {\bibfnamefont {P.}~\bibnamefont {\ifmmode
  \acute{C}\else \'{C}\fi{}wikli\ifmmode~\acute{n}\else \'{n}\fi{}ski}},
  \bibinfo {author} {\bibfnamefont {J.}~\bibnamefont {Anders}}, \bibinfo
  {author} {\bibfnamefont {M.}~\bibnamefont {Horodecki}},\ and\ \bibinfo
  {author} {\bibfnamefont {J.}~\bibnamefont {Oppenheim}},\ }\bibfield  {title}
  {\bibinfo {title} {A sufficient set of experimentally implementable thermal
  operations for small systems},\ }\href
  {https://doi.org/10.1103/PhysRevX.8.041049} {\bibfield  {journal} {\bibinfo
  {journal} {Phys. Rev. X}\ }\textbf {\bibinfo {volume} {8}},\ \bibinfo {pages}
  {041049} (\bibinfo {year} {2018}{\natexlab{b}})}\BibitemShut {NoStop}%
\bibitem [{\citenamefont {Faist}\ \emph {et~al.}(2015)\citenamefont {Faist},
  \citenamefont {Oppenheim},\ and\ \citenamefont {Renner}}]{Faist2015gibbs}%
  \BibitemOpen
  \bibfield  {author} {\bibinfo {author} {\bibfnamefont {P.}~\bibnamefont
  {Faist}}, \bibinfo {author} {\bibfnamefont {J.}~\bibnamefont {Oppenheim}},\
  and\ \bibinfo {author} {\bibfnamefont {R.}~\bibnamefont {Renner}},\
  }\bibfield  {title} {\bibinfo {title} {{Gibbs-preserving maps outperform
  thermal operations in the quantum regime}},\ }\href
  {https://doi.org/10.1088/1367-2630/17/4/043003} {\bibfield  {journal}
  {\bibinfo  {journal} {New J. Phys.}\ }\textbf {\bibinfo {volume} {17}},\
  \bibinfo {pages} {43003} (\bibinfo {year} {2015})}\BibitemShut {NoStop}%
\bibitem [{\citenamefont {Masanes}\ and\ \citenamefont
  {Oppenheim}(2017)}]{Masanes2017}%
  \BibitemOpen
  \bibfield  {author} {\bibinfo {author} {\bibfnamefont {L.}~\bibnamefont
  {Masanes}}\ and\ \bibinfo {author} {\bibfnamefont {J.}~\bibnamefont
  {Oppenheim}},\ }\bibfield  {title} {\bibinfo {title} {A general derivation
  and quantification of the third law of thermodynamics},\ }\href@noop {}
  {\bibfield  {journal} {\bibinfo  {journal} {Nat. Commun.}\ }\textbf {\bibinfo
  {volume} {8}},\ \bibinfo {pages} {14538} (\bibinfo {year}
  {2017})}\BibitemShut {NoStop}%
\bibitem [{\citenamefont {Lostaglio}\ \emph
  {et~al.}(2015{\natexlab{a}})\citenamefont {Lostaglio}, \citenamefont
  {Korzekwa}, \citenamefont {Jennings},\ and\ \citenamefont
  {Rudolph}}]{Lostaglio2015}%
  \BibitemOpen
  \bibfield  {author} {\bibinfo {author} {\bibfnamefont {M.}~\bibnamefont
  {Lostaglio}}, \bibinfo {author} {\bibfnamefont {K.}~\bibnamefont {Korzekwa}},
  \bibinfo {author} {\bibfnamefont {D.}~\bibnamefont {Jennings}},\ and\
  \bibinfo {author} {\bibfnamefont {T.}~\bibnamefont {Rudolph}},\ }\bibfield
  {title} {\bibinfo {title} {Quantum coherence, time-translation symmetry, and
  thermodynamics},\ }\href {https://doi.org/10.1103/PhysRevX.5.021001}
  {\bibfield  {journal} {\bibinfo  {journal} {Phys. Rev. X}\ }\textbf {\bibinfo
  {volume} {5}},\ \bibinfo {pages} {021001} (\bibinfo {year}
  {2015}{\natexlab{a}})}\BibitemShut {NoStop}%
\bibitem [{\citenamefont {Korzekwa}\ \emph {et~al.}(2016)\citenamefont
  {Korzekwa}, \citenamefont {Lostaglio}, \citenamefont {Oppenheim},\ and\
  \citenamefont {Jennings}}]{Korzekwa2016}%
  \BibitemOpen
  \bibfield  {author} {\bibinfo {author} {\bibfnamefont {K.}~\bibnamefont
  {Korzekwa}}, \bibinfo {author} {\bibfnamefont {M.}~\bibnamefont {Lostaglio}},
  \bibinfo {author} {\bibfnamefont {J.}~\bibnamefont {Oppenheim}},\ and\
  \bibinfo {author} {\bibfnamefont {D.}~\bibnamefont {Jennings}},\ }\bibfield
  {title} {\bibinfo {title} {The extraction of work from quantum coherence},\
  }\href {https://doi.org/10.1088/1367-2630/18/2/023045} {\bibfield  {journal}
  {\bibinfo  {journal} {New J. of Phys.}\ }\textbf {\bibinfo {volume} {18}},\
  \bibinfo {pages} {023045} (\bibinfo {year} {2016})}\BibitemShut {NoStop}%
\bibitem [{\citenamefont {Brandao}\ \emph {et~al.}(2015)\citenamefont
  {Brandao}, \citenamefont {Horodecki}, \citenamefont {Ng}, \citenamefont
  {Oppenheim},\ and\ \citenamefont {Wehner}}]{Horodecki2014}%
  \BibitemOpen
  \bibfield  {author} {\bibinfo {author} {\bibfnamefont {F.}~\bibnamefont
  {Brandao}}, \bibinfo {author} {\bibfnamefont {M.}~\bibnamefont {Horodecki}},
  \bibinfo {author} {\bibfnamefont {N.}~\bibnamefont {Ng}}, \bibinfo {author}
  {\bibfnamefont {J.}~\bibnamefont {Oppenheim}},\ and\ \bibinfo {author}
  {\bibfnamefont {S.}~\bibnamefont {Wehner}},\ }\bibfield  {title} {\bibinfo
  {title} {The second laws of quantum thermodynamics},\ }\href@noop {}
  {\bibfield  {journal} {\bibinfo  {journal} {Proc. Natl. Acad. Sci. U.S.A.}\
  }\textbf {\bibinfo {volume} {112}},\ \bibinfo {pages} {3275} (\bibinfo {year}
  {2015})}\BibitemShut {NoStop}%
\bibitem [{\citenamefont {Kwon}\ \emph {et~al.}(2018)\citenamefont {Kwon},
  \citenamefont {Jeong}, \citenamefont {Jennings}, \citenamefont {Yadin},\ and\
  \citenamefont {Kim}}]{Kwon2018clock}%
  \BibitemOpen
  \bibfield  {author} {\bibinfo {author} {\bibfnamefont {H.}~\bibnamefont
  {Kwon}}, \bibinfo {author} {\bibfnamefont {H.}~\bibnamefont {Jeong}},
  \bibinfo {author} {\bibfnamefont {D.}~\bibnamefont {Jennings}}, \bibinfo
  {author} {\bibfnamefont {B.}~\bibnamefont {Yadin}},\ and\ \bibinfo {author}
  {\bibfnamefont {M.~S.}\ \bibnamefont {Kim}},\ }\bibfield  {title} {\bibinfo
  {title} {Clock--work trade-off relation for coherence in quantum
  thermodynamics},\ }\href {https://doi.org/10.1103/PhysRevLett.120.150602}
  {\bibfield  {journal} {\bibinfo  {journal} {Phys. Rev. Lett.}\ }\textbf
  {\bibinfo {volume} {120}},\ \bibinfo {pages} {150602} (\bibinfo {year}
  {2018})}\BibitemShut {NoStop}%
\bibitem [{\citenamefont {M\"uller}(2018)}]{Mueller2018}%
  \BibitemOpen
  \bibfield  {author} {\bibinfo {author} {\bibfnamefont {M.~P.}\ \bibnamefont
  {M\"uller}},\ }\bibfield  {title} {\bibinfo {title} {Correlating thermal
  machines and the second law at the nanoscale},\ }\href
  {https://doi.org/10.1103/PhysRevX.8.041051} {\bibfield  {journal} {\bibinfo
  {journal} {Phys. Rev. X}\ }\textbf {\bibinfo {volume} {8}},\ \bibinfo {pages}
  {041051} (\bibinfo {year} {2018})}\BibitemShut {NoStop}%
\bibitem [{\citenamefont {Łobejko}\ \emph {et~al.}(2020)\citenamefont
  {Łobejko}, \citenamefont {Mazurek},\ and\ \citenamefont
  {Horodecki}}]{obejko2020thermodynamics}%
  \BibitemOpen
  \bibfield  {author} {\bibinfo {author} {\bibfnamefont {M.}~\bibnamefont
  {Łobejko}}, \bibinfo {author} {\bibfnamefont {P.}~\bibnamefont {Mazurek}},\
  and\ \bibinfo {author} {\bibfnamefont {M.}~\bibnamefont {Horodecki}},\
  }\href@noop {} {\bibinfo {title} {Thermodynamics of minimal coupling quantum
  heat engines}} (\bibinfo {year} {2020}),\ \Eprint
  {https://arxiv.org/abs/2003.05788} {arXiv:2003.05788 [quant-ph]} \BibitemShut
  {NoStop}%
\bibitem [{\citenamefont {Binder}\ \emph {et~al.}(2019)\citenamefont {Binder},
  \citenamefont {Correa}, \citenamefont {Gogolin}, \citenamefont {Anders},\
  and\ \citenamefont {Adesso}}]{binder2019thermodynamics}%
  \BibitemOpen
  \bibfield  {author} {\bibinfo {author} {\bibfnamefont {F.}~\bibnamefont
  {Binder}}, \bibinfo {author} {\bibfnamefont {L.~A.}\ \bibnamefont {Correa}},
  \bibinfo {author} {\bibfnamefont {C.}~\bibnamefont {Gogolin}}, \bibinfo
  {author} {\bibfnamefont {J.}~\bibnamefont {Anders}},\ and\ \bibinfo {author}
  {\bibfnamefont {G.}~\bibnamefont {Adesso}},\ }\bibfield  {title} {\bibinfo
  {title} {Thermodynamics in the quantum regime},\ }\href@noop {} {\bibfield
  {journal} {\bibinfo  {journal} {Fundamental Theories of Physics (Springer,
  2018)}\ } (\bibinfo {year} {2019})}\BibitemShut {NoStop}%
\bibitem [{\citenamefont {Goold}\ \emph {et~al.}(2016)\citenamefont {Goold},
  \citenamefont {Huber}, \citenamefont {Riera}, \citenamefont {del Rio},\ and\
  \citenamefont {Skrzypczyk}}]{Goold2015}%
  \BibitemOpen
  \bibfield  {author} {\bibinfo {author} {\bibfnamefont {J.}~\bibnamefont
  {Goold}}, \bibinfo {author} {\bibfnamefont {M.}~\bibnamefont {Huber}},
  \bibinfo {author} {\bibfnamefont {A.}~\bibnamefont {Riera}}, \bibinfo
  {author} {\bibfnamefont {L.}~\bibnamefont {del Rio}},\ and\ \bibinfo {author}
  {\bibfnamefont {P.}~\bibnamefont {Skrzypczyk}},\ }\bibfield  {title}
  {\bibinfo {title} {The role of quantum information in thermodynamics—a
  topical review},\ }\href@noop {} {\bibfield  {journal} {\bibinfo  {journal}
  {J. Phys. A}\ }\textbf {\bibinfo {volume} {49}},\ \bibinfo {pages} {143001}
  (\bibinfo {year} {2016})}\BibitemShut {NoStop}%
\bibitem [{\citenamefont {Vinjanampathy}\ and\ \citenamefont
  {Anders}(2016)}]{Vinjanampathy_2016}%
  \BibitemOpen
  \bibfield  {author} {\bibinfo {author} {\bibfnamefont {S.}~\bibnamefont
  {Vinjanampathy}}\ and\ \bibinfo {author} {\bibfnamefont {J.}~\bibnamefont
  {Anders}},\ }\bibfield  {title} {\bibinfo {title} {Quantum thermodynamics},\
  }\href {https://doi.org/10.1080/00107514.2016.1201896} {\bibfield  {journal}
  {\bibinfo  {journal} {Contemp. Phys.}\ }\textbf {\bibinfo {volume} {57}},\
  \bibinfo {pages} {545–579} (\bibinfo {year} {2016})}\BibitemShut {NoStop}%
\bibitem [{\citenamefont {Lostaglio}(2019)}]{Lostaglio2019}%
  \BibitemOpen
  \bibfield  {author} {\bibinfo {author} {\bibfnamefont {M.}~\bibnamefont
  {Lostaglio}},\ }\bibfield  {title} {\bibinfo {title} {An introductory review
  of the resource theory approach to thermodynamics},\ }\href
  {https://doi.org/10.1088/1361-6633/ab46e5} {\bibfield  {journal} {\bibinfo
  {journal} {Rep. Prog. Phys.}\ }\textbf {\bibinfo {volume} {82}},\ \bibinfo
  {pages} {114001} (\bibinfo {year} {2019})}\BibitemShut {NoStop}%
\bibitem [{Note1()}]{Note1}%
  \BibitemOpen
  \bibinfo {note} {Without this two technical assumptions the precise form of
  the second laws should read: $\forall \protect \tmspace +\thinmuskip
  {.1667em} \alpha \in \protect \mathbb {R}\hskip 1em\relax F_{\alpha }(\rho
  _{\protect \rm S}, \tau _{\protect \rm S}) > F_{\alpha }(\sigma _{\protect
  \rm S}, \tau _{\protect \rm S})$}\BibitemShut {NoStop}%
\bibitem [{\citenamefont {Müller-Lennert}\ \emph {et~al.}(2013)\citenamefont
  {Müller-Lennert}, \citenamefont {Dupuis}, \citenamefont {Szehr},
  \citenamefont {Fehr},\ and\ \citenamefont {Tomamichel}}]{Lennert2013}%
  \BibitemOpen
  \bibfield  {author} {\bibinfo {author} {\bibfnamefont {M.}~\bibnamefont
  {Müller-Lennert}}, \bibinfo {author} {\bibfnamefont {F.}~\bibnamefont
  {Dupuis}}, \bibinfo {author} {\bibfnamefont {O.}~\bibnamefont {Szehr}},
  \bibinfo {author} {\bibfnamefont {S.}~\bibnamefont {Fehr}},\ and\ \bibinfo
  {author} {\bibfnamefont {M.}~\bibnamefont {Tomamichel}},\ }\bibfield  {title}
  {\bibinfo {title} {On quantum rényi entropies: A new generalization and some
  properties},\ }\href {https://doi.org/10.1063/1.4838856} {\bibfield
  {journal} {\bibinfo  {journal} {J. Math. Phys.}\ }\textbf {\bibinfo {volume}
  {54}},\ \bibinfo {pages} {122203} (\bibinfo {year} {2013})},\ \Eprint
  {https://arxiv.org/abs/https://doi.org/10.1063/1.4838856}
  {https://doi.org/10.1063/1.4838856} \BibitemShut {NoStop}%
\bibitem [{\citenamefont {Cover}\ and\ \citenamefont
  {Thomas}(2006)}]{Cover2006_book}%
  \BibitemOpen
  \bibfield  {author} {\bibinfo {author} {\bibfnamefont {T.~M.}\ \bibnamefont
  {Cover}}\ and\ \bibinfo {author} {\bibfnamefont {J.~A.}\ \bibnamefont
  {Thomas}},\ }\href@noop {} {\emph {\bibinfo {title} {Elements of Information
  Theory (Wiley Series in Telecommunications and Signal Processing)}}}\
  (\bibinfo  {publisher} {Wiley-Interscience},\ \bibinfo {address} {USA},\
  \bibinfo {year} {2006})\BibitemShut {NoStop}%
\bibitem [{\citenamefont {Schumacher}(1995)}]{Schumacher1995}%
  \BibitemOpen
  \bibfield  {author} {\bibinfo {author} {\bibfnamefont {B.}~\bibnamefont
  {Schumacher}},\ }\bibfield  {title} {\bibinfo {title} {Quantum coding},\
  }\href {https://doi.org/10.1103/PhysRevA.51.2738} {\bibfield  {journal}
  {\bibinfo  {journal} {Phys. Rev. A}\ }\textbf {\bibinfo {volume} {51}},\
  \bibinfo {pages} {2738} (\bibinfo {year} {1995})}\BibitemShut {NoStop}%
\bibitem [{\citenamefont {Chubb}\ \emph {et~al.}(2019)\citenamefont {Chubb},
  \citenamefont {Tomamichel},\ and\ \citenamefont {Korzekwa}}]{Chubb_2019}%
  \BibitemOpen
  \bibfield  {author} {\bibinfo {author} {\bibfnamefont {C.~T.}\ \bibnamefont
  {Chubb}}, \bibinfo {author} {\bibfnamefont {M.}~\bibnamefont {Tomamichel}},\
  and\ \bibinfo {author} {\bibfnamefont {K.}~\bibnamefont {Korzekwa}},\
  }\bibfield  {title} {\bibinfo {title} {Moderate deviation analysis of
  majorization-based resource interconversion},\ }\href
  {http://dx.doi.org/10.1103/PhysRevA.99.032332} {\bibfield  {journal}
  {\bibinfo  {journal} {Phys. Rev. A}\ }\textbf {\bibinfo {volume} {99}}
  (\bibinfo {year} {2019})}\BibitemShut {NoStop}%
\bibitem [{\citenamefont {Leung}\ \emph {et~al.}(2008)\citenamefont {Leung},
  \citenamefont {Toner},\ and\ \citenamefont {Watrous}}]{leung2008coherent}%
  \BibitemOpen
  \bibfield  {author} {\bibinfo {author} {\bibfnamefont {D.}~\bibnamefont
  {Leung}}, \bibinfo {author} {\bibfnamefont {B.}~\bibnamefont {Toner}},\ and\
  \bibinfo {author} {\bibfnamefont {J.}~\bibnamefont {Watrous}},\ }\href@noop
  {} {\bibinfo {title} {Coherent state exchange in multi-prover quantum
  interactive proof systems}} (\bibinfo {year} {2008}),\ \Eprint
  {https://arxiv.org/abs/0804.4118} {arXiv:0804.4118 [quant-ph]} \BibitemShut
  {NoStop}%
\bibitem [{\citenamefont {Dinur}\ \emph {et~al.}(2013)\citenamefont {Dinur},
  \citenamefont {Steurer},\ and\ \citenamefont {Vidick}}]{dinur2013parallel}%
  \BibitemOpen
  \bibfield  {author} {\bibinfo {author} {\bibfnamefont {I.}~\bibnamefont
  {Dinur}}, \bibinfo {author} {\bibfnamefont {D.}~\bibnamefont {Steurer}},\
  and\ \bibinfo {author} {\bibfnamefont {T.}~\bibnamefont {Vidick}},\
  }\href@noop {} {\bibinfo {title} {A parallel repetition theorem for entangled
  projection games}} (\bibinfo {year} {2013}),\ \Eprint
  {https://arxiv.org/abs/1310.4113} {arXiv:1310.4113 [quant-ph]} \BibitemShut
  {NoStop}%
\bibitem [{\citenamefont {Berta}\ \emph {et~al.}(2011)\citenamefont {Berta},
  \citenamefont {Christandl},\ and\ \citenamefont {Renner}}]{Berta_2011}%
  \BibitemOpen
  \bibfield  {author} {\bibinfo {author} {\bibfnamefont {M.}~\bibnamefont
  {Berta}}, \bibinfo {author} {\bibfnamefont {M.}~\bibnamefont {Christandl}},\
  and\ \bibinfo {author} {\bibfnamefont {R.}~\bibnamefont {Renner}},\
  }\bibfield  {title} {\bibinfo {title} {The quantum reverse shannon theorem
  based on one-shot information theory},\ }\href
  {https://doi.org/10.1007/s00220-011-1309-7} {\bibfield  {journal} {\bibinfo
  {journal} {Comm. Math. Phys.}\ }\textbf {\bibinfo {volume} {306}},\ \bibinfo
  {pages} {579–615} (\bibinfo {year} {2011})}\BibitemShut {NoStop}%
\bibitem [{\citenamefont {Bennett}\ \emph {et~al.}(2014)\citenamefont
  {Bennett}, \citenamefont {Devetak}, \citenamefont {Harrow}, \citenamefont
  {Shor},\ and\ \citenamefont {Winter}}]{Bennett_2014}%
  \BibitemOpen
  \bibfield  {author} {\bibinfo {author} {\bibfnamefont {C.~H.}\ \bibnamefont
  {Bennett}}, \bibinfo {author} {\bibfnamefont {I.}~\bibnamefont {Devetak}},
  \bibinfo {author} {\bibfnamefont {A.~W.}\ \bibnamefont {Harrow}}, \bibinfo
  {author} {\bibfnamefont {P.~W.}\ \bibnamefont {Shor}},\ and\ \bibinfo
  {author} {\bibfnamefont {A.}~\bibnamefont {Winter}},\ }\bibfield  {title}
  {\bibinfo {title} {The quantum reverse shannon theorem and resource tradeoffs
  for simulating quantum channels},\ }\href
  {https://doi.org/10.1109/tit.2014.2309968} {\bibfield  {journal} {\bibinfo
  {journal} {IEEE Trans. Inf.}\ }\textbf {\bibinfo {volume} {60}},\ \bibinfo
  {pages} {2926–2959} (\bibinfo {year} {2014})}\BibitemShut {NoStop}%
\bibitem [{\citenamefont {Leung}\ and\ \citenamefont
  {Wang}(2014)}]{leung2014characteristics}%
  \BibitemOpen
  \bibfield  {author} {\bibinfo {author} {\bibfnamefont {D.}~\bibnamefont
  {Leung}}\ and\ \bibinfo {author} {\bibfnamefont {B.}~\bibnamefont {Wang}},\
  }\bibfield  {title} {\bibinfo {title} {Characteristics of universal
  embezzling families},\ }\href@noop {} {\bibfield  {journal} {\bibinfo
  {journal} {Phys. Rev. A}\ }\textbf {\bibinfo {volume} {90}},\ \bibinfo
  {pages} {042331} (\bibinfo {year} {2014})}\BibitemShut {NoStop}%
\bibitem [{\citenamefont {Anshu}\ \emph {et~al.}(2017)\citenamefont {Anshu},
  \citenamefont {Devabathini},\ and\ \citenamefont {Jain}}]{Anshu_2017}%
  \BibitemOpen
  \bibfield  {author} {\bibinfo {author} {\bibfnamefont {A.}~\bibnamefont
  {Anshu}}, \bibinfo {author} {\bibfnamefont {V.~K.}\ \bibnamefont
  {Devabathini}},\ and\ \bibinfo {author} {\bibfnamefont {R.}~\bibnamefont
  {Jain}},\ }\bibfield  {title} {\bibinfo {title} {Quantum communication using
  coherent rejection sampling},\ }\bibfield  {journal} {\bibinfo  {journal}
  {Phys. Rev. Lett.}\ }\textbf {\bibinfo {volume} {119}},\ \href
  {https://doi.org/10.1103/physrevlett.119.120506}
  {10.1103/physrevlett.119.120506} (\bibinfo {year} {2017})\BibitemShut
  {NoStop}%
\bibitem [{\citenamefont {Jensen}(2019)}]{jensen2019asymptotic}%
  \BibitemOpen
  \bibfield  {author} {\bibinfo {author} {\bibfnamefont {A.~K.}\ \bibnamefont
  {Jensen}},\ }\bibfield  {title} {\bibinfo {title} {Asymptotic majorization of
  finite probability distributions},\ }\href@noop {} {\bibfield  {journal}
  {\bibinfo  {journal} {IEEE Trans. Inf.}\ }\textbf {\bibinfo {volume} {65}},\
  \bibinfo {pages} {8131} (\bibinfo {year} {2019})}\BibitemShut {NoStop}%
\bibitem [{\citenamefont {Duan}\ \emph {et~al.}(2005)\citenamefont {Duan},
  \citenamefont {Feng}, \citenamefont {Li},\ and\ \citenamefont
  {Ying}}]{PhysRevA.71.062306}%
  \BibitemOpen
  \bibfield  {author} {\bibinfo {author} {\bibfnamefont {R.}~\bibnamefont
  {Duan}}, \bibinfo {author} {\bibfnamefont {Y.}~\bibnamefont {Feng}}, \bibinfo
  {author} {\bibfnamefont {X.}~\bibnamefont {Li}},\ and\ \bibinfo {author}
  {\bibfnamefont {M.}~\bibnamefont {Ying}},\ }\bibfield  {title} {\bibinfo
  {title} {Trade-off between multiple-copy transformation and entanglement
  catalysis},\ }\href {https://doi.org/10.1103/PhysRevA.71.062306} {\bibfield
  {journal} {\bibinfo  {journal} {Phys. Rev. A}\ }\textbf {\bibinfo {volume}
  {71}},\ \bibinfo {pages} {062306} (\bibinfo {year} {2005})}\BibitemShut
  {NoStop}%
\bibitem [{\citenamefont {Feng}\ \emph {et~al.}(2006)\citenamefont {Feng},
  \citenamefont {Duan},\ and\ \citenamefont {Ying}}]{PhysRevA.74.042312}%
  \BibitemOpen
  \bibfield  {author} {\bibinfo {author} {\bibfnamefont {Y.}~\bibnamefont
  {Feng}}, \bibinfo {author} {\bibfnamefont {R.}~\bibnamefont {Duan}},\ and\
  \bibinfo {author} {\bibfnamefont {M.}~\bibnamefont {Ying}},\ }\bibfield
  {title} {\bibinfo {title} {Relation between catalyst-assisted transformation
  and multiple-copy transformation for bipartite pure states},\ }\href
  {https://doi.org/10.1103/PhysRevA.74.042312} {\bibfield  {journal} {\bibinfo
  {journal} {Phys. Rev. A}\ }\textbf {\bibinfo {volume} {74}},\ \bibinfo
  {pages} {042312} (\bibinfo {year} {2006})}\BibitemShut {NoStop}%
\bibitem [{\citenamefont {Lostaglio}\ and\ \citenamefont
  {Mueller}(2018)}]{Lostaglio2018}%
  \BibitemOpen
  \bibfield  {author} {\bibinfo {author} {\bibfnamefont {M.}~\bibnamefont
  {Lostaglio}}\ and\ \bibinfo {author} {\bibfnamefont {M.~P.}\ \bibnamefont
  {Mueller}},\ }\bibfield  {title} {\bibinfo {title} {Coherence and asymmetry
  cannot be broadcast},\ }\href@noop {} {\bibfield  {journal} {\bibinfo
  {journal} {arXiv e-prints}\ } (\bibinfo {year} {2018})},\ \Eprint
  {https://arxiv.org/abs/1812.08214} {arXiv:1812.08214} \BibitemShut {NoStop}%
\bibitem [{\citenamefont {Sanders}\ and\ \citenamefont
  {Gour}(2009)}]{Sanders_2009}%
  \BibitemOpen
  \bibfield  {author} {\bibinfo {author} {\bibfnamefont {Y.~R.}\ \bibnamefont
  {Sanders}}\ and\ \bibinfo {author} {\bibfnamefont {G.}~\bibnamefont {Gour}},\
  }\bibfield  {title} {\bibinfo {title} {Necessary conditions for entanglement
  catalysts},\ }\bibfield  {journal} {\bibinfo  {journal} {Physical Review A}\
  }\textbf {\bibinfo {volume} {79}},\ \href
  {https://doi.org/10.1103/physreva.79.054302} {10.1103/physreva.79.054302}
  (\bibinfo {year} {2009})\BibitemShut {NoStop}%
\bibitem [{git(2020)}]{git_code}%
  \BibitemOpen
  \href {https://github.com/patryk-lb/all-states-are-universal-catalysts}
  {\bibinfo {title} {Matlab code at github.com/patryk-lb}} (\bibinfo {year}
  {2020})\BibitemShut {NoStop}%
\bibitem [{\citenamefont {Horodecki}\ \emph {et~al.}(2018)\citenamefont
  {Horodecki}, \citenamefont {Oppenheim},\ and\ \citenamefont
  {Sparaciari}}]{Horodecki_2018}%
  \BibitemOpen
  \bibfield  {author} {\bibinfo {author} {\bibfnamefont {M.}~\bibnamefont
  {Horodecki}}, \bibinfo {author} {\bibfnamefont {J.}~\bibnamefont
  {Oppenheim}},\ and\ \bibinfo {author} {\bibfnamefont {C.}~\bibnamefont
  {Sparaciari}},\ }\bibfield  {title} {\bibinfo {title} {Extremal distributions
  under approximate majorization},\ }\href
  {https://doi.org/10.1088/1751-8121/aac87c} {\bibfield  {journal} {\bibinfo
  {journal} {J. Phys. A Math. Theor.}\ }\textbf {\bibinfo {volume} {51}},\
  \bibinfo {pages} {305301} (\bibinfo {year} {2018})}\BibitemShut {NoStop}%
\bibitem [{\citenamefont {Turgut}(2007)}]{Turgut_2007}%
  \BibitemOpen
  \bibfield  {author} {\bibinfo {author} {\bibfnamefont {S.}~\bibnamefont
  {Turgut}},\ }\bibfield  {title} {\bibinfo {title} {Catalytic transformations
  for bipartite pure states},\ }\href
  {https://doi.org/10.1088/1751-8113/40/40/012} {\bibfield  {journal} {\bibinfo
   {journal} {Journal of Physics A: Mathematical and Theoretical}\ }\textbf
  {\bibinfo {volume} {40}},\ \bibinfo {pages} {12185} (\bibinfo {year}
  {2007})}\BibitemShut {NoStop}%
\bibitem [{Note2()}]{Note2}%
  \BibitemOpen
  \bibinfo {note} {It's worth to remind that the transformations which we
  consider here are always deterministic and the probability $p_{\protect \rm
  succ}(\epsilon _{\protect \rm \relax \protect \fontsize {8}{9.5pt}\protect
  \selectfont \abovedisplayskip 6\p@ plus2\p@ minus4\p@ \belowdisplayskip
  \abovedisplayskip \abovedisplayshortskip \z@ plus\p@ \belowdisplayshortskip
  3\p@ plus\p@ minus2\p@ \def \leftmargin \leftmargini \parsep 4\p@ plus2\p@
  minus\p@ \topsep 8\p@ plus2\p@ minus4\p@ \itemsep 4\p@ plus2\p@ minus\p@
  {\leftmargin \leftmargini \topsep 3\p@ plus\p@ minus\p@ \parsep 2\p@ plus\p@
  minus\p@ \itemsep \parsep }C})$ refers to the sampled catalysts, i.e. how
  probable it is to \protect \emph {deterministically} catalyze a given
  transformation using a catalyst drawn at random.}\BibitemShut {Stop}%
\bibitem [{\citenamefont {Lostaglio}\ \emph
  {et~al.}(2015{\natexlab{b}})\citenamefont {Lostaglio}, \citenamefont
  {Jennings},\ and\ \citenamefont {Rudolph}}]{lostaglio2015description}%
  \BibitemOpen
  \bibfield  {author} {\bibinfo {author} {\bibfnamefont {M.}~\bibnamefont
  {Lostaglio}}, \bibinfo {author} {\bibfnamefont {D.}~\bibnamefont
  {Jennings}},\ and\ \bibinfo {author} {\bibfnamefont {T.}~\bibnamefont
  {Rudolph}},\ }\bibfield  {title} {\bibinfo {title} {Description of quantum
  coherence in thermodynamic processes requires constraints beyond free
  energy},\ }\href@noop {} {\bibfield  {journal} {\bibinfo  {journal} {Nat.
  Commun.}\ }\textbf {\bibinfo {volume} {6}},\ \bibinfo {pages} {1} (\bibinfo
  {year} {2015}{\natexlab{b}})}\BibitemShut {NoStop}%
\bibitem [{\citenamefont {Korzekwa}\ \emph {et~al.}(2019)\citenamefont
  {Korzekwa}, \citenamefont {Chubb},\ and\ \citenamefont
  {Tomamichel}}]{Korzekwa_2019}%
  \BibitemOpen
  \bibfield  {author} {\bibinfo {author} {\bibfnamefont {K.}~\bibnamefont
  {Korzekwa}}, \bibinfo {author} {\bibfnamefont {C.~T.}\ \bibnamefont
  {Chubb}},\ and\ \bibinfo {author} {\bibfnamefont {M.}~\bibnamefont
  {Tomamichel}},\ }\bibfield  {title} {\bibinfo {title} {Avoiding
  irreversibility: Engineering resonant conversions of quantum resources},\
  }\href {http://dx.doi.org/10.1103/PhysRevLett.122.110403} {\bibfield
  {journal} {\bibinfo  {journal} {Phys. Rev. Lett.}\ }\textbf {\bibinfo
  {volume} {122}} (\bibinfo {year} {2019})}\BibitemShut {NoStop}%
\bibitem [{\citenamefont {Schmid}\ \emph {et~al.}(2020)\citenamefont {Schmid},
  \citenamefont {Fraser}, \citenamefont {Kunjwal}, \citenamefont {Sainz},
  \citenamefont {Wolfe},\ and\ \citenamefont {Spekkens}}]{schmid2020standard}%
  \BibitemOpen
  \bibfield  {author} {\bibinfo {author} {\bibfnamefont {D.}~\bibnamefont
  {Schmid}}, \bibinfo {author} {\bibfnamefont {T.~C.}\ \bibnamefont {Fraser}},
  \bibinfo {author} {\bibfnamefont {R.}~\bibnamefont {Kunjwal}}, \bibinfo
  {author} {\bibfnamefont {A.~B.}\ \bibnamefont {Sainz}}, \bibinfo {author}
  {\bibfnamefont {E.}~\bibnamefont {Wolfe}},\ and\ \bibinfo {author}
  {\bibfnamefont {R.~W.}\ \bibnamefont {Spekkens}},\ }\href@noop {} {\bibinfo
  {title} {Why standard entanglement theory is inappropriate for the study of
  bell scenarios}} (\bibinfo {year} {2020}),\ \Eprint
  {https://arxiv.org/abs/2004.09194} {arXiv:2004.09194 [quant-ph]} \BibitemShut
  {NoStop}%
\bibitem [{\citenamefont {Egloff}\ \emph {et~al.}(2015)\citenamefont {Egloff},
  \citenamefont {Dahlsten}, \citenamefont {Renner},\ and\ \citenamefont
  {Vedral}}]{Egloff2015}%
  \BibitemOpen
  \bibfield  {author} {\bibinfo {author} {\bibfnamefont {D.}~\bibnamefont
  {Egloff}}, \bibinfo {author} {\bibfnamefont {O.~C.~O.}\ \bibnamefont
  {Dahlsten}}, \bibinfo {author} {\bibfnamefont {R.}~\bibnamefont {Renner}},\
  and\ \bibinfo {author} {\bibfnamefont {V.}~\bibnamefont {Vedral}},\
  }\bibfield  {title} {\bibinfo {title} {A measure of majorization emerging
  from single-shot statistical mechanics},\ }\href
  {https://doi.org/10.1088/1367-2630/17/7/073001} {\bibfield  {journal}
  {\bibinfo  {journal} {New J. Phys.}\ }\textbf {\bibinfo {volume} {17}},\
  \bibinfo {pages} {073001} (\bibinfo {year} {2015})}\BibitemShut {NoStop}%
\bibitem [{\citenamefont {Renes}(2016)}]{Renes2016}%
  \BibitemOpen
  \bibfield  {author} {\bibinfo {author} {\bibfnamefont {J.~M.}\ \bibnamefont
  {Renes}},\ }\bibfield  {title} {\bibinfo {title} {Relative submajorization
  and its use in quantum resource theories},\ }\href@noop {} {\bibfield
  {journal} {\bibinfo  {journal} {J. Math. Phys.}\ }\textbf {\bibinfo {volume}
  {57}},\ \bibinfo {pages} {122202} (\bibinfo {year} {2016})}\BibitemShut
  {NoStop}%
\end{thebibliography}%

\onecolumngrid
\appendix

\section*{Appendix}
\section{Notation}
\noindent For a discrete probability distribution $\bm{p} = (p_1, p_2, \ldots, p_{d_{\rm S}})$ the Shannon entropy $H(\bm{p})$ and entropy variance $V(\bm{p})$ are defined as:
\begin{align}
    H(\bm{p}) &:= - \sum_{i} p_i \log p_i, \qquad
    V(\bm{p}) := \sum_{i} p_i \left[-\log p_i - H(\bm{p}) \right]^2 
\end{align}
Similarly we define the relative entropy $D(\bm{p}||\bm{q})$ and the relative entropy variance $V(\bm{p}||\bm{q})$ as:
\begin{align}
    D(\bm{p}||\bm{q}) := \sum_{i} p_i \log\left(\frac{p_i}{q_i}\right), \qquad V(\bm{p}||\bm{q}) := \sum_{i} p_i \left[\log\left(\frac{p_i}{q_i}\right) - D(\bm{p}||\bm{q})\right]^2.
\end{align}
Renyi entropies $H_{\alpha}$ and the corresponding Renyi relative entropies $D_{\alpha}(\bm{p}||\bm{q})$ for all $\alpha \in \mathbb{R}$ are defined by:
\begin{align}
    H_{\alpha}(\bm{p}) := \frac{\text{sgn}(\alpha)}{1 -\alpha}\log\left[ \sum_{i} p_i^{\alpha} \right], \qquad D_{\alpha}(\bm{p}||\bm{q}) := \frac{\text{sgn}(\alpha)}{1 - \alpha} \log \left[\sum_i p_i^{\alpha} q_i^{1-\alpha}\right]
\end{align}
The generalized free energies for a state block-diagonal in the energy eigenbasis $\rho$, such that $\bm{p} = \text{diag}[\rho]$, with Hamiltonian $H$ and in contact with environment at inverse temperature $\beta$, are defined as:
\begin{align}
    \label{def:gen_fre_ene}
    F_{\alpha}(\bm{p}) = D_{\alpha}(\bm{p}||\bm{g}) - \log Z,
\end{align}
where $Z = \Tr[e^{-\beta H}]$ is the partition function, $\bm{g} = \diag[{\tau}]$ and $\tau = e^{-\beta H} / Z$ is the thermal state.

\section{Preliminary lemmas}
Here we present two lemmas which are essential for proving our results. We start with a simple realization that for any thermal operation which maps the system's state into another state with some error we can always find a ``gentle'' unitary which achieves the same error, as measured by the trace distance, on the system and its environment. 
\begin{lemma}
\label{lemma:gentleu}
Let $\mathcal{T}$ be an thermal operation satisfying:
\begin{align}
    \label{eq:lemma1_1}
    \norm{\mathcal{T}[\rho_{\rm S}] - \sigma_{\rm S}}_1 \leq \epsilon.
\end{align}
for block-diagonal states $\rho_{\rm S}$ and $\sigma_{\rm S}$. Then there exists a unitary $U_{\rm SB}$ and the environment system $\rm B$ with Hamiltonian $H_{\rm B}$ such that:
\begin{align}
\norm{U_{\rm SB}\left(\rho_{\rm S} \ot \tau_{\rm B}\right)U_{\rm SB}^{\dagger} - \sigma_{\rm S} \ot \tau_{\rm B}}_1 \leq \epsilon,
\end{align}
where $\tau_{\rm B} := e^{-\beta H_{\rm B}} / Z_{\rm B}$ with $Z_{\rm B} := \Tr[e^{-\beta H_{\rm B}} ]$ is the thermal state of the environment and $[U_{\rm SB}, H_{\rm S} + H_{\rm B}] = 0$.
\end{lemma}
\begin{proof}
Let us write the states $\rho_{\rm S}$ and $\sigma_{\rm S}$ in their energy eigenbasis as:
\begin{align}
    \rho_{\rm S} = \sum_{i} p_{i} \dyad{i}_{\rm S}, \qquad 
    \sigma_{\rm S} = \sum_{i} q_i \dyad{i}_{\rm S}.
\end{align}
Let $\{r(j|i)\}$ denote the transition probabilities of the map $\mathcal{T}$, i.e. $\mathcal{T}[\rho_{\rm S}] = \sum_{i,j} p_{i}r(j|i) \dyad{j}_{\rm S}$. Then (\ref{eq:lemma1_1}) implies:
\begin{align}
    \label{eq:error_sigma}
    \norm{\mathcal{T}[\rho_{\rm S}] - \sigma_{\rm S}}_1 =  \frac{1}{2} \sum_j  \Big|q_j - \sum_ip_i r(j|i)\Big| \leq \epsilon. 
\end{align}
We now represent our total state $\rho_{\rm S} \ot \gamma_{\rm B}$ in blocks of constant energy. This is a standard approach which is described in more detail e.g. in \cite{Horodecki2013,Lostaglio2019}. To do so we first write the Hamiltonian of the bath as: 
\begin{align}
    \label{eq:heat_bath_ham}
    H_{\rm B} = \sum_{j}\sum_{g}^{g(E_j^{\rm B})}E_j^{\rm B} \dyad{E_j^{\rm B}}_{\rm B},
\end{align}
where $g(E)$ is the degeneracy of the energy level $E$. In what follows we  make two assumptions about the spectrum of the heat bath: $(i)$ all energy gaps are present, i.e. for any two energy levels of the system, $E_i^{\rm S}$ and $E_j^{\rm S}$, there exists $E_k^{\rm B}$ and $E_l^{\rm B}$ such that $E_i^{\rm S}-E_j^{\rm S} = E_k^{\rm B} - E_l^{\rm B}$ and $
(ii)$ degeneracies scale exponentially, that is $g(E + \epsilon) = g(E) e^{\beta \epsilon}$. Both of these  are standard assumptions in thermodynamics. The first condition is valid for heat baths with continuous energy spectra and can be approached arbitrarily well by discrete heath baths. The second condition is naturally satisfied when the average energy of the heat bath is much larger than that of the system.

With (\ref{eq:heat_bath_ham}) in mind we can write the initial state of the environment as: 
\begin{align}
    \tau_{\rm B} = \sum_{j}\sum_{g}^{g\left(E_{j}^{\rm B}\right)} \frac{e^{-\beta E_{j}^{\rm B}}}{Z_{\rm B}}  \dyad{E_j^{\rm B}, g}_{\rm B}
\end{align}
The joint state of the system $\rm S$ and the environment $\rm B$ can be written as:
\begin{align}
    \rho_{\rm S} \ot \tau_{\rm B} &= \sum_{i, j} \sum_{g}^{g\left(E_j^{\rm B}\right)} p_i \frac{e^{-\beta E_j^{\rm B}}}{Z_{\rm B}} \dyad{i}_{\rm S} \ot \dyad{E_j^{\rm B}, g}_{\rm B}\\ &=\sum_{E} \sum_{i} \sum_{g}^{g(E-E_i^{\rm S})} p_i \frac{e^{-\beta E}}{Z_{\rm B}} e^{\beta E_i^{\rm S}} \dyad{i}_{\rm S} \ot \dyad{E - E_i^{\rm S}, g}_{\rm B},
\end{align}
where in the second line we replaced summation over the heat bath energies with a summation over the total energies $E := E_{i}^{\rm S} + E_j^{\rm B}$. Using assumption $(ii)$ and denoting $d_i(E) = g(E)e^{-\beta E_i^{\rm S}}$ leads to:
\begin{align}
    \rho_{\rm S} \ot \tau_{\rm B} = \sum_E  \sum_i \sum_g^{d_i(E)} C(E)  \frac{p_i}{d_i(E)} \dyad{E - E_i^{\rm S}, i, g}_{\rm SB}. 
\end{align}
where we labelled $C(E) := g(E) e^{-\beta E} / Z_{\rm B}$. Now, the assumption that $[U_{\rm SB}, H_{\rm S}+H_{\rm B}] = 0$ implies that the unitary $U_{\rm SB}$ is block-diagonal in the basis of $H_{\rm SB}$. Writting $U_{\rm SB} = \bigoplus_{E} U^{(E)}_{\rm SB}$ we can choose each $U^{(E)}_{\rm SB}$ to be a permutation of eigenvalues within a fixed energy shell determined by $E$. Such a permutation can be described by providing a collection of numbers $\{n_{j|i}(E)\}$ which denote the number of eigenvalues moved from block $i$ to block $j$ in energy shell $E$. Unitarity of $U_{\rm SB}$ and the fact that it is energy-preserving implies that these numbers for all $E$ have to satisfy the constraints: $\sum_{i} n_{j|i}(E) = d_j(E)$ and $n_{j|i}(E) = d_i(E)$. For our purposes we choose these numbers to be:
\begin{align}
    \forall \, E \qquad n_{j|i} (E) = d_i(E) r(j|i). 
\end{align}
It can be easily checked that this is a feasible choice as long as the channel probabilities $\{r(j|i)\}$ preserve the thermal state (which is always the case). The final state of the system and the environment after application of the unitary $U_{\rm SB}$ becomes:
\begin{align}
    U_{\rm SB} \,  \left[\rho_{\rm S} \ot \tau_{\rm B} \right]\, U_{\rm SB}^{\dagger} &= \sum_{E} \sum_{i,j} \sum_{g}^{d_i(E)} C(E)  \frac{p_i}{d_i(E)} \frac{n_{j|i}(E)}{d_j(E)} \dyad{E, j,g}_{\rm SB}\\
    &= \sum_{E} \sum_{i,j} C(E) \frac{p_i}{d_j(E)} r(j|i) \dyad{E, j, g}_{\rm SB} \\
    &=  \sum_{E} \sum_{j}C(E) \frac{q_j'}{d_j(E)} \dyad{E, j,g}_{\rm SB},
\end{align}
where in the third line we labelled $q_j' = \sum_i p_i r(j|i)$. The total error on the system and the bath then reads:
\begin{align}
    \norm{U_{\rm SB} \,\left[\rho_{\rm S} \ot \tau_{\rm B}\right]\, U_{\rm SB}^{\dagger} - \sigma_{\rm S} \ot \tau_{\rm B}}_{1} &= \norm{\sum_{j, E}\!\!\sum_{g}^{d_j(E)} \frac{C(E)}{d_j(E)} \left(q_j' - q_j\right) \dyad{E,j,g}_{\rm SB}}_1\! \\
    &= \sum_{j, E} \sum_{g}^{d_j(E)} \frac{C(E)}{d_j(E)} \cdot \frac{1}{2} |q_j'-q_j| \\
    &= \frac{1}{2} \sum_j |q_j'-q_j| \\
    &\leq \epsilon,
\end{align}
where the inequality follows from (\ref{eq:error_sigma}).
\end{proof}
The second lemma was proven in \cite{Chubb_2019} and provides an estimate on the transformation error when converting between multiple copies of states.  
\begin{lemma}
\label{lemma:nonasymp} 
For any $\alpha \in (0, 1)$ and any states diagonal in the energy basis $\rho = \emph{diag}[\bm{p}]$ and $\sigma = \emph{diag}[\bm{q}]$ there exists a thermal operation $\mathcal{T}$ such that:
\begin{align}
    \norm{\mathcal{T}[\rho^{\ot n} \ot \tau^{\ot n r_n}] - \sigma^{\ot n r_n} \ot \tau^{\ot n
    }}_1 \leq \epsilon_n,
\end{align}
where the error $\epsilon_n$ and the conversion rate $r_n$ are given by:
\begin{align}
    \label{eq:lemma:nonasymp_eqs}
    \epsilon_n &= e^{-n^{\kappa} }, \\
    \label{eq:lemma:nonasymp_rn}
    r_n &= \frac{D(\rho||\tau)}{D(\sigma||\tau)}\left(1 - \sqrt{\frac{2V(\rho||\tau)}{D(\rho||\tau)}} \left|1-\frac{1}{\sqrt{v}}\right| \cdot \frac{1}{\sqrt{n^{1-\kappa}}}  \right) = \frac{D(\rho||\tau)}{D(\sigma||\tau)} - \mathcal{O}\left(\frac{1}{\sqrt{n^{1-\kappa}}}\right).
\end{align}
and where the parameter $v$ is given by:
\begin{align}
    v := \frac{V(\bm{p}||\bm{g}) / D(\bm{p}||\bm{g})}{V(\bm{q}||\bm{g}) / D(\bm{q}||\bm{g})}
\end{align}
\end{lemma}

\section{Details of the main protocol}
\subsection{Setting up the scene}
We begin by describing our main protocol of catalytic state conversion in full detail. Let $\rm S$ label the system, $\rm C$ the catalyst and let $\rm B = \rm B_{1} B_{2}$ denote the ambient heat bath which we decompose into two parts for reasons that will become clear later on. Let $U_{\rm CB_{1}}$  be the unitary which is applied in the first step of the protocol. % and which transforms the joint system $\rm CB_{1}$ into $\rm C'B'_{1}$. The primed notation means that the subsystems may in general have different dimensions before and after the unitary is applied, but such that $\rm CB_{1}$ has the same dimension as $\rm C'B'_{1}. The inverse unitary $U_{\rm CB_1}^{\dagger}$ is the unique unitary operator which satisfies $U_{\rm CB_1} U_{\rm CB_1}^{\dagger} = \mathbb{1}_{\rm CB_1}$. % 
Similarly we define the unitary $V_{\rm SCB_2}$ which is applied in the second step and which acts on $\rm SC$ and the second part of the heat bath $\rm B_2$. To simplify notation we will denote the channels generated by these unitaries with:
\begin{align}
\mathcal{U}_{\rm CB_1}[\cdot] := U_{\rm CB_1} \left(\cdot\right)  U_{\rm CB_1}^{\dagger}, \qquad   \mathcal{U}^{\dagger}_{\rm CB_1}[\cdot] := U_{\rm CB_1}^{\dagger} \left(\cdot\right)  U_{\rm CB_1}, \qquad \mathcal{V}_{\rm SCB_{2}} := V_{\rm SCB_2} \left(\cdot \right) V_{\rm SCB_2}^{\dagger}.
\end{align}
Our protocol starts with a product state $\rho_{\rm SCB}^{(0)} = \rho_{\rm S} \ot \omega_{\rm C}^{\ot n} \ot \tau_{\rm B_1}^{\,} \ot \tau_{\rm B_2}^{\,}$. The full state of the joint system at each step of the protocol can be written as:
\begin{align}
    &\rho_{\rm SCB_{1}B_{2}}^{(1)} = \left(\mathcal{I}_{\rm S} \ot \mathcal{U}_{\rm CB_1}\ot \mathcal{I}_{\rm B_2}\right)[\rho_{\rm SCB}^{(0)}], \qquad 
    \rho_{\rm SCB_{1}B_{2}}^{(2)} = \left(\mathcal{V}_{\rm SCB_{2}}\ot \mathcal{I}_{\rm B_1}\right)[\rho_{\rm SCB_{1}B_{2}}^{(1)}], \\
    &\rho_{\rm SCB_{1}B_{2}}^{(3)} = \left(\mathcal{I}_{\rm S B_2} \ot \mathcal{U}_{\rm CB_1}^{\dagger} \right)[\rho_{\rm SCB_{1}B_{2}}^{(2)}].
\end{align}
Notice that if the unitaries $U_{\rm CB_1}$ and $V_{\rm SCB_2}$ both commute with the total Hamiltonian $H_{\rm SCB}$ then so their composition commutes as well. Hence, tracing out the degrees of freedom associated with the bath leads to a valid thermal operation. The final state of the system and the catalyst after this procedure can be expressed as:
\begin{align}
    \rho_{\rm SC}' 
    = \Tr_{\rm B_1B_2}\left[\rho_{\rm SCB_1B_2}^{(3)}\right] 
    &=  \Tr_{\rm B_1B_2} \left[\left(\mathcal{I}_{\rm S} \ot \mathcal{U}_{\rm CB_1}^{\dagger}\ot \mathcal{I}_{\rm B_2}\right)\left(\mathcal{V}_{\rm SCB_{2}}\ot \mathcal{I}_{\rm B_1}\right)\left(\mathcal{I}_{\rm S B_2} \ot \mathcal{U}_{\rm CB_1}^{} \right)[\rho_{\rm S} \ot \omega_{\rm C}^{\ot n} \ot \tau_{\rm B_1}^{\,} \ot \tau_{\rm B_2}^{\,}]\right] \\
    &=: \mathcal{T}_{\rm SC}[\rho_{\rm S} \ot \omega_{\rm C}^{\ot n}],
\end{align}
where we implicitly defined the effective thermal operation $\mathcal{T}_{\rm SC}$ acting on the system and the catalyst. In what follows we will argue that by an appropriate choice of unitaries $U_{\rm CB_1}$ and $V_{\rm SCB_2}$ we can obtain the approximate catalytic state conversion for a pair of states $\rho_{\rm S}$ and $\sigma_{\rm S}$, i.e.:
\begin{align}
\mathcal{T}_{\rm SC}[\rho_{\rm S} \ot \omega_{\rm C}^{\ot n}] \approx \sigma_{\rm S} \ot \omega_{\rm C}^{\ot n},
\end{align}
Before we get into the details of how to choose these unitaries let us first find a general expression for the errors on the system and the catalyst.

\subsection{Error analysis}
We defined the errors on the system and the catalyst in the following way:
\begin{align}
    \eps{S} := \norm{\Tr_{\rm C}\left[\mathcal{T}_{\rm SC}[\rho_{\rm S} \ot \omega_{\rm C}^{\ot n}]\right] - \sigma_{\rm S}}_1, \qquad
    \eps{C} := \norm{\Tr_{\rm S}\left[\mathcal{T}_{\rm SC}[\rho_{\rm S} \ot \omega_{\rm C}^{\ot n}]\right] - \omega_{\rm C}^{\ot n}}_1
\end{align}
Let us start by writing explicitly the error on the catalyst. To keep the notation compact we define the following effective channel acting on the catalyst:
\begin{align}
    \label{def:EC}
    \mathcal{E}_{\rm C}[\omega] := \Tr_{\rm SB_{2}}\left[ V_{\rm SCB_2} \left( \rho_{\rm S} \ot \omega_{\rm C} \ot \tau_{\rm B_2}\right)V_{\rm SCB_2}^{\dagger} \right].
\end{align}
Using this definition we can write the final error on the catalyst as:
\begin{align}
\eps{C} &= \norm{\Tr_{\rm SB_1B_2} \left[\left(\mathcal{I}_{\rm S} \ot \mathcal{U}_{\rm CB_1}^{\dagger}\ot \mathcal{I}_{\rm B_2}\right)\left(\mathcal{V}_{\rm SCB_{2}}\ot \mathcal{I}_{\rm B_1}\right)\left(\mathcal{I}_{\rm S B_2} \ot \mathcal{U}_{\rm CB_1}^{} \right)[\rho_{\rm S} \ot \omega_{\rm C}^{\ot n} \ot \tau_{\rm B_1}^{\,} \ot \tau_{\rm B_2}^{\,}]\right] - \omega_{\rm C}^{\ot n}}_{1} \\
%%%%%%%%%%%%%%%%%%%%%%%%%%%%%%%%%%%%%%%%%%%%%
&= \norm{\Tr_{\rm B_1} \left[\mathcal{U}_{\rm CB_1}^{\dagger} \left(\mathcal{E}_{\rm C}\ot \mathcal{I}_{\rm B_1}\right) \mathcal{U}_{\rm CB_1} \left[\omega_{\rm C}^{\ot n} \ot \tau_{\rm B_1}^{\,} \right]\right] - \omega_{\rm C}^{\ot n}}_{1} \\
%%%%%%%%%%%%%%%%%%%%%%%%%%%%%%%%%%%%%%%%%%%%%
&=  \norm{\Tr_{\rm B_1} \left[\mathcal{U}_{\rm CB_1}^{\dagger} \left(\mathcal{E}_{\rm C}\ot \mathcal{I}_{\rm B_1}\right) \left(\mathcal{U}_{\rm CB_1} \left[\omega_{\rm C}^{\ot n} \ot \tau_{\rm B_1}^{\,}\right] - \omega_{\rm CB_1}'\right)\right] + \Tr_{\rm B_1} \left[\mathcal{U}_{\rm CB_1}^{\dagger} \left(\mathcal{E}_{\rm C}\ot \mathcal{I}_{\rm B_1}\right) \left[ \omega'_{\rm CB_1} \right]\right] - \omega_{\rm C}^{\ot n} }_1 \\
%%%%%%%%%%%%%%%%%%%%%%%%%%%%%%%%%%%%%%%%%%%%%
&\leq\norm{\Tr_{\rm B_1} \left[\mathcal{U}_{\rm CB_1}^{\dagger} \left(\mathcal{E}_{\rm C}\ot \mathcal{I}_{\rm B_1}\right) \left(\mathcal{U}_{\rm CB_1} \left[\omega_{\rm C}^{\ot n} \ot \tau_{\rm B_1}^{\,}\right] - \omega_{\rm CB_1}'\right)\right]}_1 + \norm{\Tr_{\rm B_1} \left[\mathcal{U}_{\rm CB_1}^{\dagger} \left(\mathcal{E}_{\rm C}\ot \mathcal{I}_{\rm B_1}\right) \left[ \omega'_{\rm CB_1} \right]\right] - \omega_{\rm C}^{\ot n} }_1 \\
%%%%%%%%%%%%%%%%%%%%%%%%%%%%%%%%%%%%%%%%%%%%%
&\leq \norm{\mathcal{U}_{\rm CB_1} \left[\omega_{\rm C}^{\ot n} \ot \tau_{\rm B_1}^{\,}\right] - \omega_{\rm CB_1}'}_1 + \norm{\Tr_{\rm B_1} \left[\mathcal{U}_{\rm CB_1}^{\dagger} \left(\mathcal{E}_{\rm C}\ot \mathcal{I}_{\rm B_1}\right) \left[ \omega'_{\rm CB_1} \right]\right] - \omega_{\rm C}^{\ot n} }_1
\end{align}
Let us label the first of the terms above with $\delta(n)$, i.e. $\delta(n):= \norm{\mathcal{U}_{\rm CB_1} \left[\omega_{\rm C}^{\ot n} \ot \tau_{\rm B_1}^{\,}\right] - \omega_{\rm CB_1}'}_1$. With this in mind and after subtracting and adding $\omega_{\rm CB_1}'$ in the second term we can further write:
%%%%%%%%%%%%%%%%%%%%%%%%%%%%%%%%%%%%%%%%%%%%%
\begin{align}
\eps{C} &\leq \delta(n) + \norm{\Tr_{\rm B_1} \left[\mathcal{U}_{\rm CB_1}^{\dagger} \left( \left(\mathcal{E}_{\rm C}\ot \mathcal{I}_{\rm B_1}\right) \left[ \omega'_{\rm CB_1} \right]-\omega_{\rm CB_1}'\right)\right] + \Tr_{\rm B_1}\left[\mathcal{U}_{\rm CB_1}^{\dagger}[\omega_{\rm CB_1}']\right]- \omega_{\rm C}^{\ot n} }_1 \\
%%%%%%%%%%%%%%%%%%%%%%%%%%%%%%%%%%%%%%%%%%%%%
&\leq \delta(n) + \norm{\left(\mathcal{E}_{\rm C}\ot \mathcal{I}_{\rm B_1}\right)[\omega_{\rm CB_1}'] - \omega_{\rm CB_1}'}_1 + \norm{\Tr_{\rm B_1}\left[\mathcal{U}_{\rm CB_1}^{\dagger}[\omega_{\rm CB_1}'] - \omega_{\rm C}^{\ot n} \ot \tau_{\rm B_1}\right]}_1.
\end{align}
We now label the second term with $\nu(n) := \norm{\left(\mathcal{E}_{\rm C}\ot \mathcal{I}_{\rm B_1}\right)[\omega_{\rm CB_1}'] - \omega_{\rm CB_1}'}_1$ and again use the fact that trace distance decreases under partial trace and is invariant under unitaries. This leads to the following upper bound:
\begin{align}
    \label{eq:tot_err}
    \eps{C} \leq 2 \delta(n) + \nu(n).
\end{align}

This decomposes the total error on the catalyst into two contributions, $\delta(n)$ and $\nu(n)$. The first of these terms accounts for the error due to producing a state different from $\omega_{\rm CB_1}'$ in the preprocesssing step, as well as the failure in recovering the initial state $\omega_{\rm C}^{\ot n} \ot \tau_{\rm B_1}$ after applying channel $\mathcal{E}_{\rm C}$ in the postprocessing step. The second term, $\nu(n)$, is responsible for the disturbance applied to the catalyst and the heat bath $\rm B_1$ in the catalytic step via channel $\mathcal{E}_{\rm C}$ and changes depending on which regime of catalysis one is currently interested in.

Note that so far the state $\omega_{\rm CB_1}'$ is arbitrary and appears in both contributions $\delta(n)$ and $\nu(n)$ to the total error on the catalyst. Let us now decompose the bath $\rm B_1$ into two parts $\rm B_1'$ and $\rm B_1''$ and choose the following state:
\begin{align}
\label{def:omegaCB1}
\omega_{\rm CB_1}' := \eta_{\rm C}^{\ot n r_n} \ot \tau_{\rm B_1'}^{\ot n} \ot \tau_{\rm B_1''},
\end{align}
where $0 \leq r_n \leq 1$ is some real number which potentially depends on $n$ and $\eta_{\rm C}$ is a arbitrary state which we will choose later. 

We now invoke Lemma \ref{lemma:nonasymp} which assures that there is a thermal operation $\mathcal{T}_{\rm CB_1'}$ which transforms $n$ copies of the initial catalyst state $\omega_{\rm C}^{\ot n}$ into another state $\eta_{\rm C}^{\ot n r_n}$ (potentially with different dimensions). Here $r_n$ is the conversion rate which, up to the leading order, depends on non-equilibrium Helmholtz free energies of $\omega$ and $\eta$. Formally thermal operation $\mathcal{T}_{\rm CB_1'}$ satisfies:
\begin{align}
    \norm{\mathcal{T}_{\rm CB_1'}\left[\omega_{\rm C}^{\ot n} \ot \tau_{\rm B_1'}^{\ot n r_n}\right]  - \eta_{\rm C}^{\ot nr_n} \ot \tau_{\rm B_1'}^{\ot n}}_1 \leq \epsilon_n,
\end{align}
where the conversion rate $r_n$ and the transformation error $\epsilon_n$ are given in (\ref{eq:lemma:nonasymp_eqs}). Now we use Lemma \ref{lemma:gentleu} which states that every thermal operation can be written in a ``Stinespring form'' using a global unitary $U_{\rm CB_1'B_1''}$ acting on the system $\rm CB_1'$ and the heat bath $\rm B_1''$, i.e.:
\begin{align}
    \mathcal{T}_{\rm CB_1'}[\rho_{\rm CB_1'}] = \Tr_{\rm B_1''}\left[U_{\rm CB_1'B_1''}\left(\rho_{\rm CB_1'}\ot \tau_{\rm B_1''}\right)U_{\rm CB_1'B_1''}^{\dagger}\right].
\end{align}
Importantly, this unitary does not increase the transformation error $\epsilon_n$ when enlarged to the bigger system $\rm CB_1'B_1''$, that is:
\begin{align}
    \norm{U_{\rm CB_1'B_1''} \left( \omega_{\rm C}^{\ot n} \ot \tau_{\rm B_1'}^{\ot n r_n} \ot \tau_{\rm B_1''}^{} \right)U_{\rm CB_1'B_1''}^{\dagger} - \eta_{\rm C}^{\ot nr_n} \ot \tau_{\rm B_1'}^{\ot n} \ot \tau_{\rm B_1''} }_1  \leq \epsilon_n
\end{align}
This implies that the transformation error induced by the unitary channel $\mathcal{U}_{\rm CB_1}$ is upper bounded as:
\begin{align}
\label{eq:del_err}
\delta(n) \leq \epsilon_n = e^{-n^{\kappa}}.
\end{align}
Up to this point we have the freedom to choose the operation $\mathcal{E}_{\rm C}$ and the state $\eta$, both of which can be thought of as parameters of our protocol. We now discuss two different choices which both lead to quantitatively different results.

\subsubsection{Embezzlement regime}
In this section we will apply a recent result from the theory of quantum communication called ``convex-split lemma'' to show that any state, given that enough copies of it are available, can act as an embezzler for any state transformation. 

Recall that our main goal is to carry out the transformation $\rho \rightarrow \sigma$ on the system $\rm S$ in the same time keeping the error on the catalyst small. To simplify notation let us denote the number of copies of the state $\eta$ obtained after applying the unitary $U_{\rm CB_1}$ with $m = n r_n$. In order to specify the channel $\mathcal{E}_{\rm C}$ defined in $(\ref{def:EC})$ let us first recall the process from the main text:
\begin{align}
   \mathcal{T}^{(\text{mix})} [\cdot] = \frac{1}{m} \sum_{i=0}^{m} S_{(0, i)} (\cdot) S_{(0,i)}^{\dagger},
\end{align}
where $S_{(i, j)}$ is a unitary which swaps subsystems $i$ and $j$ leaving all the remaining systems untouched, i.e.:
\begin{align}
    % S_{(1, i)} \left[\ket{a_1}_{\rm A_1} \ket{a_2}_{\rm A_2} \ldots \ket{a_i}_{\rm A_i} \ldots \ket{a_j}_{\rm A_j} \ldots\right] = \ket{a_1}_{\rm A_1} \ket{a_2}_{\rm A_2} \ldots \ket{a_j}_{\rm A_i} \ldots \ket{a_i}_{\rm A_j} \ldots
    S_{(i, j)} \ket{a_0, a_1, a_2, \ldots a_i, \ldots a_j, \ldots} = \ket{a_0, a_1, a_2, \ldots a_j, \ldots a_i, \ldots}.
\end{align}
We apply this transformation to the state of the system $\rho_{\rm S}$ (which is treated as the zeroth subsystem) and $m$ copies of the catalyst state $\eta_{\rm C}^{\ot m}$ (which are treated as the remaining $m$ subsystems). Due to the Lemma \ref{lemma:gentleu} we can always find a unitary $V_{\rm SCB_{2}}$ which acts on $\rm SC$ and the corresponding part of the heat bath $\rm B_{\rm 2}$ such that, after tracing out the bath, the action on $\rm SC$ is fully described by the channel (\ref{eq:mix_perm}). With this in mind we can write the error term $\nu(n)$ as:
\begin{align}
    \nu(n) &= \norm{\left(\mathcal{E}_{\rm C}\ot \mathcal{I}_{\rm B_1}\right)[\omega_{\rm CB_1}'] - \omega_{\rm CB_1}'}_1 \\
    &= \norm{\left(\mathcal{E}_{\rm C}\ot \mathcal{I}_{\rm B_1}\right)[\eta_{\rm C}^{\ot m} \ot \tau_{\rm B_1'}^{\ot n} \ot \tau_{\rm B_1''}] - \eta_{\rm C}^{\ot m} \ot \tau_{\rm B_1'}^{\ot n} \ot \tau_{\rm B_1''} }_1 \\
    &= \norm{\mathcal{E}_{\rm C} [\eta_{\rm C}^{\ot m}] - \eta_{\rm C}^{\ot m}}_1 \\
    &= \norm{\Tr_{\rm SB_{2}}\left[ V_{\rm SCB_2} \left( \rho_{\rm S} \ot \eta_{\rm C}^{\ot m} \ot \tau_{\rm B_2}\right)V_{\rm SCB_2}^{\dagger} \right] - \eta_{\rm C}^{\ot  m}}_{1}\\
    &= \norm{\Tr_{\rm S}\left[ \mathcal{T}^{\text{(mix)}}_{\rm SC} [\rho_{\rm S} \ot \eta_{\rm C}^{\ot m}]\right] - \eta_{\rm C}^{\ot m}}_{1} \\
    & \leq \norm{ \mathcal{T}^{\text{(mix)}}_{\rm SC} [\rho_{\rm S} \ot \eta_{\rm C}^{\ot m}] - \sigma_{\rm S} \ot \eta_{\rm C}^{\ot m}}_{1}, 
\end{align}
where in the second line we used the explicit form of the state $\omega_{\rm CB_1}'$ given in (\ref{def:omegaCB1}) and in the fourth line we used the definition of the channel $\mathcal{E}_{\rm C}$ from Eq. (\ref{def:EC}). In the fifth line labelled the mixing channel induced by the unitary $V_{\rm SCB_2}$ with $\mathcal{T}_{\rm SC}^{\text{(mix)}}$, i.e. $\mathcal{T}_{\rm SC}^{\text{(mix)}}[\cdot] := \Tr_{\rm B_2}\left[ V_{\rm SCB_2} \left( (\cdot)_{\rm SC} \ot \tau_{\rm B_2}\right)V_{\rm SCB_2}^{\dagger}\right]$ .  The last inequality follows from the fact that the trace distance is contractive when tracing out subsystems. Notice now that so far we have not made the choice for the state $\eta$. As we will soon see, a good choice is to pick the state $\eta$ to be precisely our target state on $\rm S$, i.e. $\eta = \sigma$. In this way the error term $\nu(n)$ becomes:
\begin{align}
    \label{eq:nu_conv_split}
    \nu(n) \leq \norm{ \mathcal{T}^{\text{(mix)}}_{\rm SC} [\rho_{\rm S} \ot \sigma_{\rm C}^{\ot m}] - \sigma_{\rm S} \ot \sigma_{\rm C}^{\ot m}}_{1}
\end{align}
In what follows we will label the subsystems which comprise the catalyst $\rm C$ with $\mathrm{C}_i$ for $i \in [1, m]$, i.e.  $\mathrm{C} = \mathrm{C}_1 \mathrm{C}_2 \ldots \mathrm{C}_m$. To further upper bound $\nu(n)$ we need the following result adapted from \cite{Anshu_2017}. Actually, we will write the lemma in a more specialised form which better demonstrates its applicability to our particular problem.
\begin{lemma}[Convex-split, \cite{Anshu_2017}]
\label{lemma:convex}
Let $\rho$ and $\sigma$ be two quantum states satisfying $\emph{supp}(\rho) \subseteq \emph{supp}(\sigma)$. Then the following state:
\begin{align}
    \mathcal{T}_{\mathrm{SC}} ^{\emph{(mix)}}\left[\rho_{\mathrm{S}} \ot \sigma_{\mathrm{C}}^{\ot m}\right] = \frac{1}{m}\left[\rho_{\mathrm{S}} \ot \sigma_{\mathrm{C}_{1}} \ot \ldots \ot \sigma_{\mathrm{C}_{m}} + \sigma_{\mathrm{S}} \ot \rho_{\mathrm{C}_{1}} \ot \ldots \ot \sigma_{\mathrm{C}_{m}} + \ldots + \sigma_{\mathrm{S}} \ot \sigma_{\mathrm{C}_{1}} \ot \ldots \ot \rho_{\mathrm{C}_{m}} \right]
\end{align}
for large $m$ is close to the state $\sigma_{\rm S} \ot \sigma^{\ot m}_{\rm C}$, that is, for all $m \geq 1$ the following holds:
\begin{align}
    \norm{\mathcal{T}_{\mathrm{SC}} ^{\emph{(mix)}}\left[\rho_{\mathrm{S}} \ot \sigma_{\mathrm{C}}^{\ot m}\right] - \sigma_{\rm S} \ot \sigma_{\rm C}^{\ot m}}_{1} \leq \frac{1}{\sqrt{m}} \cdot 2^{\frac{1}{2} D_{\emph{max}}(\rho||\sigma)}.
\end{align}
where $D_{\emph{max}}(\rho||\sigma) := \log \min \left\{\lambda \,|\, \rho \leq \lambda \sigma \right\}$ is the max-relative entropy.
\end{lemma}
Using the above lemma we can readily bound the error term from (\ref{eq:nu_conv_split}) as:
\begin{align}
 \nu(n) &\leq \frac{1}{\sqrt{m}} \cdot 2^{\frac{1}{2} D_{\text{max}}(\rho||\sigma)} \\
 &=  \sqrt{\frac{2^{D_{\text{max}}(\rho||\sigma)}}{r_n}}  \cdot \frac{1}{\sqrt{n}} \\
 &= c_n \cdot \frac{1}{\sqrt{n}},
\end{align}
where we used $m = n r_n$ and $r_n$ is given in (\ref{eq:lemma:nonasymp_rn}). Notice now that for fixed input and output states $\rho$ and $\sigma$ and for large $n$ the factor $c_n$ is effectively constant. Hence the error term $\nu(n)$ is, up to the leading order in $n$, scaling with $n$ as :
\begin{align}
    \nu(n) \sim \frac{1}{\sqrt{n}}.
\end{align}
Importantly, notice that we have not assumed anything special about the states $\rho$ and $\sigma$ (in fact our only technical assumption is that $\text{supp}({\rho}) \subseteq \text{supp}(\sigma)$ which can be always satisfied by considering arbitrarily small perturbations of the states). Consequently, the two states are arbitrary, which is an embodiment of the fact that in the embezzlement regime partial order between states fully vanishes.
\subsubsection{Genuine catalysis regime}
We now move to the case when the partial order between states is fully retained, i.e. when the second laws of thermodynamics are satisfied by a pair of states $\rho$ and $\sigma$. Recall that we are working with block-diagonal states and hence, any such state can be expressed as a a probability vector. Let us denote:
\begin{align}
    \bm{p} := \diag[\rho_{\rm S}], \qquad \bm{q} := \diag[\sigma_{\rm S}], \qquad \bm{c} := \diag[\eta_{\rm C}], \qquad \bm{g} := \diag[\tau_{\rm S}]
\end{align}
By definition then $\bm{p}$, $\bm{q}$ and $\bm{g}$ are $d_{\rm S}$-dimensional vectors and $\bm{c}$ is a $d_{\rm C}$-dimensional vector. Our only assumption in this section is that the states described by probability vectors $\bm{p}$ and $\bm{q}$ satisfy the second laws of thermodynamics, i.e.:
\begin{align}
    \label{eq:sec_laws_supp}
    \forall\, \alpha \geq 0\qquad F_{\alpha}(\bm{p}, \bm{g}) \geq F_{\alpha}(\bm{q}, \bm{g}),
\end{align}
where the generalized free energies are defined in (\ref{def:gen_fre_ene}). Before we specify the state $\eta_{\rm C}$ and the channel $\mathcal{E}_{\rm C}$ let us rephrase the conditions (\ref{eq:sec_laws_supp}) in a form which relates them with Renyi entropies. 

The generalized free energies $F_{\alpha}$ can be connected with Renyi entropies $H_{\alpha}$ using the so-called \emph{embedding map} introduced in \cite{Brand_o_2015}. Intuitively, the embedding map is is an operation which allows to translate between the microcanonical and macrocanonical descriptions of a thermodynamic system (see e.g. \cite{Egloff2015}). For clarity we also define it here. 
\begin{definition}[Embedding map]
The embedding map $\Gamma_{\bm{d}}: \mathbb{R}^{n} \rightarrow \mathbb{R}^{D}$ is a transformation between vectors such that for $\bm{x} = (x_1, x_2, \ldots, x_n)$ we have:
\begin{align}
    \Gamma_{\bm{d}}\left(\bm{x}\right) := \left( \underbrace{\frac{x_1}{d_1}, \ldots, \frac{x_1}{d_1}}_{d_1 \, \text{terms}}, \underbrace{\frac{x_2}{d_2}, \ldots, \frac{x_2}{d_2}}_{d_2 \, \text{terms}}, \ldots, \underbrace{\frac{x_n}{d_n}, \ldots, \frac{x_n}{d_n}}_{d_n \, \text{terms}}  \right) = \bigoplus_{i = 1}^{n} \left[x_i \cdot \bm{u}_i \right],
\end{align}
where $\bm{d} = (d_1, d_2, \ldots, d_n)$ is a vector of natural numbers which sum to $D = \sum_{i=1}^{n} d_i$ and $\bm{u}_i := \frac{1}{i} \left(1, 1, \ldots,1 \right)$ is an $i-$dimensional uniform distribution. 
\end{definition}
We assume that the thermal state $\bm{g}$ is a vector with rational entries, i.e. there exists a collection of natural numbers $\{d_1, d_2, \ldots, d_{d_{\rm S}}\}$ such that $D = \sum_{i=1}^{d_{\rm S}} d_i$ and:
\begin{align}
    \label{eq:gibbs_rational}
    \bm{g} := \left(\frac{d_1}{D}, \frac{d_2}{D}, \ldots, \frac{d_{d_{\rm S}}}{D}\right).
\end{align}
Taking the thermal state as in (\ref{eq:gibbs_rational}) and applying the embedding map with $\bm{d} = (d_1, d_2, \ldots, d_{d_{\rm S}})$ leads to:
\begin{align}
    \Gamma_{\bm{d}}(\bm{g}) = \bm{u}_{D},
\end{align}
where $\bm{u}_{D}$ is the uniform distribution of dimension $D = \sum_{i=1}^{d_{\rm S}} d_i$. As described in \cite{Brand_o_2015}, the embedding map leads to the following relationship between the generalized free energies $F_{\alpha}$ and Renyi entropies $H_{\alpha}$:
\begin{align}
    \log Z_{\rm S} + \beta F_{\alpha}(\bm{x}, \bm{g})  = \log D - H_{\alpha}\left(\Gamma_{\bm{d}}(\bm{p})\right).
\end{align}
In this way the embedding map also allows to rewrite the second laws purely in terms of the Renyi entropies. Denoting $\widetilde{\bm{p}} := \Gamma_{\bm{d}}(\bm{p})$ and $\widetilde{\bm{q}} := \Gamma_{\bm{d}}(\bm{q})$ the conditions in (\ref{eq:sec_laws_supp}) become:
\begin{align}
    H_{\alpha}(\widetilde{\bm{p}}) \leq  H_{\alpha}(\widetilde{\bm{q}})
\end{align}
The second laws expressed using Renyi entropies not only imply the existence of a state which catalyzes a given transformation. In fact, they also determine when multiple copies of one state can be converted into multiple copies of another state asymptotically. 
This relationship is captured by the following result from majorization theory adapted from  \cite{jensen2019asymptotic} (Proposition 3.2.7):
\begin{lemma}
Let $\widetilde{\bm{p}}$ and $\widetilde{\bm{q}}$ be two probability distributions of dimension $D$ such that:
\begin{align}
    \label{eq:renyi_conds}
    H_{\alpha}(\widetilde{\bm{p}}) < H_{\alpha}(\widetilde{\bm{q}}) \qquad \forall \, \alpha \geq 0,
\end{align}
then for sufficiently large $k$ the following holds:
\begin{align}
    \label{eq:k_copy}
    \widetilde{\bm{p}}^{\,\ot k} \succ \widetilde{\bm{q}}^{\,\ot k}.
\end{align}
\end{lemma}
\noindent Let us note that using the properties of the embedding map it can be easily verified that (\ref{eq:k_copy}) can be equivalently written as:
\begin{align}
    \bm{p}^{\ot k} \succ_{\rm  T} \bm{q}^{\ot k}.
\end{align}
Let us begin by showing that we can replace strict inequalities in the above Lemma with non-strict ones. The argument proceeds similarly as in \cite{Brand_o_2015} (Appendix B, Proposition 3) and we repeat it here for convenience. 

It is easy to see that when (\ref{eq:k_copy}) holds then the non-strict inequalities are satisfied. In order to prove the converse suppose that the conditions (\ref{eq:renyi_conds}) hold for all $\alpha \geq 0$ but with non-strict inequalities. Consider $\bm{q}_{\epsilon} := (1-\epsilon) \bm{q} + \epsilon \bm{u}$, with $\bm{u}$ being a uniform distribution. Since $\bm{q}_{\epsilon}$ is more mixed than $\bm{q}$ we have that $\bm{q} \succ \bm{q}_{\epsilon}$. Since Renyi entropies are (strictly) Schur-concave for any $\alpha > 0$ \cite{Marshall2011}, we have that $H_{\alpha}(\bm{p}) < H_{\alpha}(\bm{q_{\epsilon}})$. Therefore, for a sufficiently large $k$ we have that $\bm{p}^{\ot k} \succ \bm{q}_{\epsilon}^{\ot k}$. Now, since this holds for any $\epsilon \geq 0$ and majorization relation is continuous with respect to probability distributions, we also have that $\bm{p}^{\ot k} \succ \bm{q}^{\ot k}$.

Let us now recall the definition of the \emph{Duan state} $\omega^{\rm D}_k({\rho}, {\sigma})$ \cite{PhysRevA.71.062306}. For convenience we denote $\bm{z}^{\rm D}_k(\widetilde{\bm{p}}, \widetilde{\bm{q}}) = \text{diag}[\omega^{\rm D}_k({\rho}, {\sigma})]$, where:
\begin{align}
    \label{def:dua_sta}
    \bm{z}^{\rm D}_k(\widetilde{\bm{p}}, \widetilde{\bm{q}}) :=\! \frac{1}{k}\!\left[ \widetilde{\bm{p}}^{\,\ot k-1}\! \oplus\! \left(\widetilde{\bm{p}}^{\,\ot k-2}\!\ot \widetilde{\bm{q}} \right)\! \oplus \! \left(\widetilde{\bm{p}}^{\,\ot k-3}\!\ot \widetilde{\bm{q}}^{\, \ot 2} \right)\! \oplus \ldots \oplus \! \left(\widetilde{\bm{p}} \ot \widetilde{\bm{q}}^{\, \ot k-2} \right)\! \oplus \widetilde{\bm{q}}^{\,\ot k-1}\right]
\end{align}
Then, as it was shown in  \cite{PhysRevA.71.062306, PhysRevA.74.042312}, the Duan state $\bm{z}^{\rm D}_k(\widetilde{\bm{p}}, \widetilde{\bm{q}})$ is an exact catalyst for a transformation between states $\widetilde{\bm{p}}$ and $\widetilde{\bm{q}}$ if $k$ copies of $\widetilde{\bm{p}}$ can be converted into $k$ copies of $\widetilde{\bm{q}}$. In other words, if the condition (\ref{eq:k_copy}) is satisfied then the following also holds:
\begin{align}
    \widetilde{\bm{p}} \ot \bm{z}^{\rm D}_k(\widetilde{\bm{p}}, \widetilde{\bm{q}}) \succ  \widetilde{\bm{q}} \ot \bm{z}^{\rm D}_k(\widetilde{\bm{p}}, \widetilde{\bm{q}})
\end{align}
To see this note that if (\ref{eq:k_copy}) is satisfied for some $k$ then we also have:
\begin{align}
    \widetilde{\bm{p}} \ot \bm{z}^{\rm D}_k(\widetilde{\bm{p}}, \widetilde{\bm{q}}) & = \frac{1}{k}\!\left[ \widetilde{\bm{p}}^{\,\ot k}\! \oplus\! \left(\widetilde{\bm{p}}^{\,\ot k-1}\!\ot \widetilde{\bm{q}} \right)\! \oplus \! \left(\widetilde{\bm{p}}^{\,\ot k-2}\!\ot \widetilde{\bm{q}}^{\, \ot 2} \right)\! \oplus \ldots \oplus \! \left(\widetilde{\bm{p}}^{\, \ot 2} \ot \widetilde{\bm{q}}^{\, \ot k-2} \right)\! \oplus
    \left(\widetilde{\bm{p}}\ot \widetilde{\bm{q}}^{\,\ot k-1}\right)\right] \\
    %%%%%%%%%%%%%%%%%%%%%%%%%%%%%
    &\succ \frac{1}{k}\!\left[ \widetilde{\bm{q}}^{\,\ot k}\! \oplus\! \left(\widetilde{\bm{p}}^{\,\ot k-1}\!\ot \widetilde{\bm{q}} \right)\! \oplus \! \left(\widetilde{\bm{p}}^{\,\ot k-2}\!\ot \widetilde{\bm{q}}^{\, \ot 2} \right)\! \oplus \ldots \oplus \! \left(\widetilde{\bm{p}}^{\, \ot 2} \ot \widetilde{\bm{q}}^{\, \ot k-2} \right)\! \oplus
    \left(\widetilde{\bm{p}}\ot \widetilde{\bm{q}}^{\,\ot k-1}\right)\right] \\
     %%%%%%%%%%%%%%%%%%%%%%%%%%%%%
    &= \widetilde{\bm{q}} \ot \bm{z}^{\rm D}_k(\widetilde{\bm{p}}, \widetilde{\bm{q}})
\end{align}
Hence, if the second laws are satisfied there exists large enough $k$ such that $k$ copies of $\widetilde{\bm{p}}$ can be converted into $k$ copies of $\widetilde{\bm{q}}$. This, on the other hand, means that there exists a special catalyst (the Duan state) which can catalyze the transformation from $\widetilde{\bm{p}}$ to  $\widetilde{\bm{q}}$ \emph{without} any disturbance. In this way we obtain the following theorem:
\begin{theorem}
    \label{theorem_k_copy}
    Let $\bm{p}$ and $\bm{q}$ be two probability distributions of dimension $d_{\rm S}$ such that:
\begin{align}
    \qquad F_{\alpha}({\bm{p}}, \bm{g}) \geq F_{\alpha}({\bm{q}}, \bm{g}) \qquad  \forall \, \alpha \geq 0,
\end{align}
then for sufficiently large $k$ the following holds:
\begin{align}
    \label{eq:theorem_k_copy}
    \bm{p} \ot \bm{z}^{\rm D}_k({\bm{p}}, \bm{q}) \, \succ_{\small \rm T}\, \bm{q} \ot \bm{z}^{\rm D}_k({\bm{p}}, {\bm{q}}),
\end{align}
where $\bm{z}^{\rm D}_k({\bm{p}}, \bm{q})$ is the Duan state defined in (\ref{def:dua_sta}).
\end{theorem}
We can now return to our main protocol. We choose the Duan state corresponding to $\bm{p}$ and $\bm{q}$ as the intermediate state, i.e.:
\begin{align}
    \label{eq:eta_cat}
    \eta_{\rm C} = \omega_{k}^{\rm D}(\rho, \sigma).
\end{align}
In fact we only need one copy of the state $\eta_{\rm C}$ to transform $\rho_{\rm S}$ into $\sigma_{\rm S}$ and so we choose $m = 1$. In order to specify the intermediate transformation $\mathcal{E}_{\rm C}$ let us recall that it is fully specified by the unitary $V_{\rm SCB_2}$ and defined as:
\begin{align}
    \mathcal{E}_{\rm C} &:= \Tr_{\rm SB_{2}}\left[ V_{\rm SCB_2} \left( \rho_{\rm S} \ot \eta_{\rm C}^{\ot m} \ot \tau_{\rm B_2}\right)V_{\rm SCB_2}^{\dagger} \right] \\
    &= \mathcal{T}_{\rm SC}^{\text{(cat)}}[\rho_{\rm S} \ot \eta_{\rm C}^{\ot m}],
\end{align}
where we labelled the thermal operation induced by the unitary $V_{\rm SCB_2}$ with:
\begin{align}
    \mathcal{T}_{\rm SC}^{\text{(cat)}}[\cdot] := \Tr_{\rm B_2}\left[ V_{\rm SCB_2} \left( (\cdot)_{\rm SC} \ot \tau_{\rm B_2}\right)V_{\rm SCB_2}^{\dagger}\right].
\end{align}
Theorem \ref{theorem_k_copy} assures that for $\eta$ as chosen in (\ref{eq:eta_cat}) and $m \geq 1$ the state $\rho_{\rm S} \ot \eta_{\rm C}$ thermo-majorizes the state $\sigma_{\rm S} \ot \eta_{\rm C}$. This, due to the fundamental result of \cite{Horodecki2013}, also means that there exists a thermal operation connecting this two states and hence we choose $\mathcal{T}_{\rm SC}^{\text{(cat)}}$ to be precisely this transformation. As as a result we have:
\begin{align}
    \label{eq:tau_SC}
    \mathcal{T}_{\rm SC}^{\text{(cat)}}[\rho_{\rm S} \ot \omega_k^{\rm D}(\rho_{\rm S}, \sigma_{\rm S})] = \sigma_{\rm S} \ot \omega_k^{\rm D}(\rho, \sigma).
\end{align}
Notice that even though we chose the transformation $\mathcal{T}_{\rm SC}^{\text{(cat)}}$ implicitly, it can still be constructed using the methods described e.g. in \cite{Renes2016} (p. 29, below Lemma 7, see also the references therein), i.e. by finding an appropriate Gibbs-stochastic matrix and then using the methods described in the proof of Lemma \ref{lemma:gentleu} to construct a permutation unitary which leads to the desired thermal operation.   

Importantly, the catalytic transformation from (\ref{eq:tau_SC}) is exact, i.e. it does not induce any new error on the system nor the catalyst. As a result the error term $\nu(n)$ vanishes:
\begin{align}
    \nu(n) &\leq \norm{ \mathcal{T}_{\rm SC}^{\text{(cat)}}[\rho_{\rm S} \ot \omega_k^{\rm D}(\rho, \sigma)] - \sigma_{\rm S} \ot \omega_k^{\rm D}(\rho, \sigma)}_{1} \\
    &=0,
\end{align}
and so the only error on the catalyst system is due to the pre and post-processing steps of the main protocol (\ref{eq:tot_err}).

\section{\chg{Numerical analysis for higher dimensional systems}}
In this section we present the results of a supplementary numerical computation that provides further evidence that catalytic universality is a generic phenomenon. In particular, we compute the fraction $f(\bm{p})$ from Sec. IV.C for different dimension of the system $d_{\rm S}$ and different choice of the error parameter $\mu$ and threshold value $\gamma_{\rm thld}$. The results are presented in Fig. \ref{fig:supplementary_numerics}.   
\begin{figure}
    \centering
    \includegraphics[width=\linewidth]{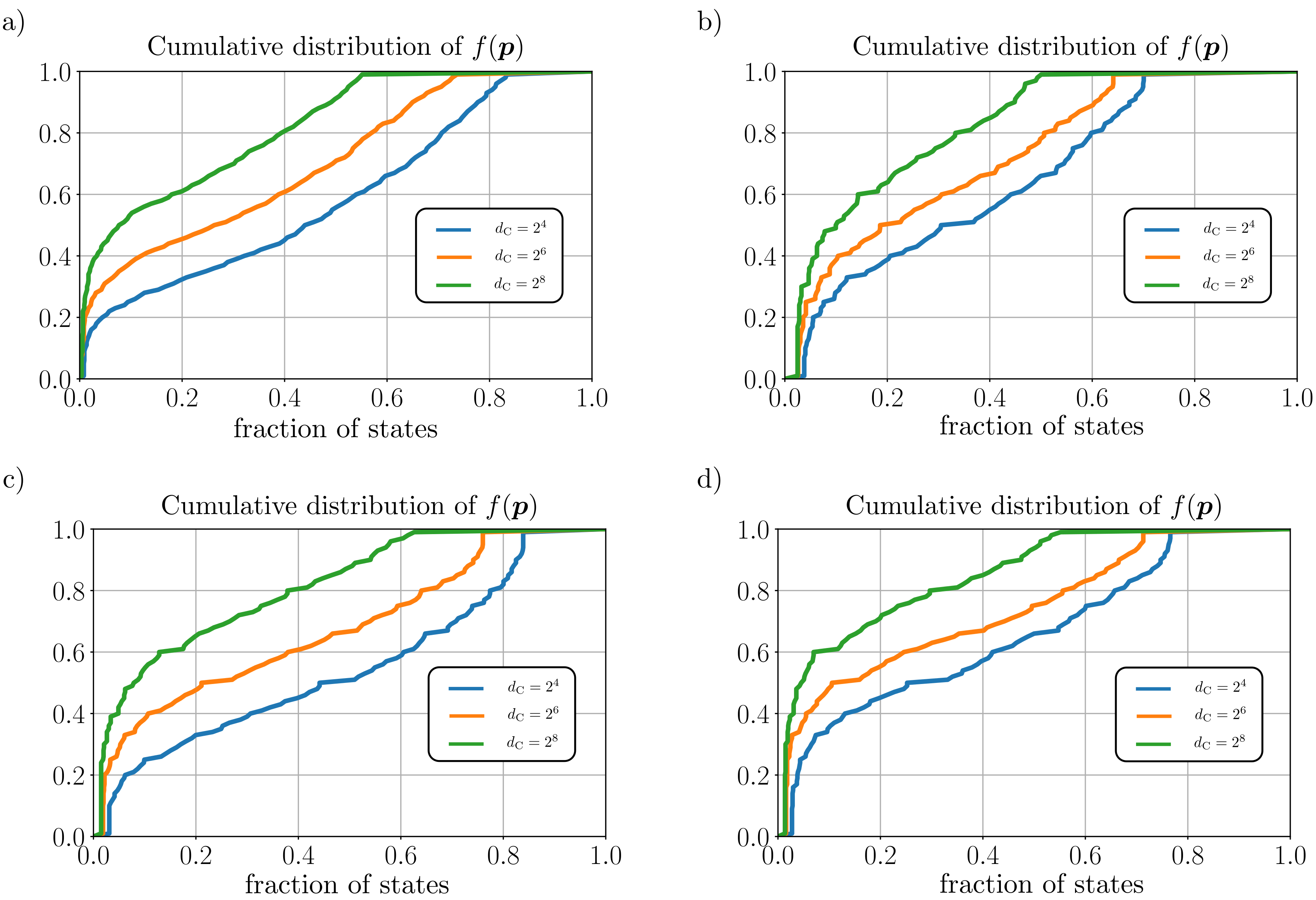}
    \caption{Cumulative distribution of $f(\bm{p})$ for different choices of parameters: $(a)$ $[d_{\rm S} = 4, \mu = 0.05, \gamma_{\rm thld} = 0.9]$, $(b)$ $[d_{\rm S} = 5, \mu = 0.05, \gamma_{\rm thld} = 0.9]$, $(c)$ $[d_{\rm S} = 4, \mu = 0.1, \gamma_{\rm thld} = 0.8]$ $(d)$ $[d_{\rm S} = 5, \mu = 0.1, \gamma_{\rm thld} = 0.8]$. }
    \label{fig:supplementary_numerics}
\end{figure}

\end{document}